\pdfoutput=1
\documentclass[10pt]{article}

\usepackage[numbers]{natbib}

\usepackage[utf8]{inputenc}
\usepackage[T1]{fontenc}
\usepackage{hyperref}
\usepackage{url}
\usepackage{booktabs}
\usepackage{amsfonts}
\usepackage{nicefrac}
\usepackage{microtype}
\usepackage{nicefrac,xfrac}
\usepackage{geometry} 
\usepackage{multirow}
\usepackage{lipsum}
\usepackage{latexsym,dsfont}
\usepackage{graphicx}
\usepackage[font=normalsize]{caption}
\usepackage{subcaption}
\usepackage{amsmath}
\usepackage{mathtools}
\usepackage{amssymb}
\usepackage{bbold}
\usepackage{enumitem}   
\usepackage{algorithmic}
\usepackage{algorithm}
\usepackage{tcolorbox}

\def \E{\mathbb E}
\usepackage{color}

\newcommand*{\email}[1]{
    \normalsize\href{mailto:#1}{\texttt{#1}}
    }

\usepackage{amsthm}
\newtheorem{theorem}{Theorem}
\newtheorem{remark}{Remark}
\newtheorem{lemma}{Lemma}
\newtheorem{definition}{Definition}

\DeclareMathOperator*{\argmin}{argmin}

\title{Discrete-Valued Latent Preference Matrix Estimation\\ with Graph Side Information}

\date{University of Wisconsin-Madison}
\author{
  Changhun Jo \\ \email{cjo4@wisc.edu} \and   Kangwook Lee \\ \email{kangwook.lee@wisc.edu}
}

\begin{document}

\maketitle

\begin{abstract}
Incorporating graph side information into recommender systems has been widely used to better predict ratings, but relatively few works have focused on theoretical guarantees. Ahn et al. (2018) firstly characterized the optimal sample complexity in the presence of graph side information, but the results are limited due to strict, unrealistic assumptions made on the unknown latent preference matrix and the structure of user clusters. In this work, we propose a new model in which 1) the unknown latent preference matrix can have any discrete values, and 2) users can be clustered into multiple clusters, thereby relaxing the assumptions made in prior work. Under this new model, we fully characterize the optimal sample complexity and develop a computationally-efficient algorithm that matches the optimal sample complexity. Our algorithm is robust to model errors and outperforms the existing algorithms in terms of prediction performance on both synthetic and real data.
\end{abstract}

\section{Introduction}
\label{sec:1}

Recommender systems provide suggestions for items based on users' decisions such as ratings given to those items.
Collaborative filtering is a popular approach to designing recommender systems~\cite{Herlocker_1999,Sarwar_2001, Linden_2003, Rennie_2005, Salakhutdinov_2007, Salakhutdinov_2008, Agarwal_2010, Davenport_2014}.
However, collaborative filtering suffers from the well-known cold start problem since it relies only on past interactions between users and items. With the exponential growth of social media, recommender systems have started to use a social graph to resolve the cold start problem. For instance, \citet{Jamali_2010} provide an algorithm that handles the cold start problem by exploiting social graph information.

While a lot of works have improved the performance of algorithms by incorporating graph side information into recommender systems~\cite{Jamali_2009_1, Jamali_2009_2, Jamali_2010, Cai_2011, Ma_2011, Yang_2012, Yang_2013, Kalofolias2014MC}, relatively few works have focused on justifying theoretical guarantees of the performance~\cite{Chiang_2015, Rao_2015, Ahn_2018}. One notable exception is the recent work of \citet{Ahn_2018}, which finds the minimum number of observed ratings for reliable recovery of the latent preference matrix with social graph information and partial observation of the rating matrix. They also provide an efficient algorithm with low computational complexity.
However, the assumptions made in this work are too strong to reflect the real-world data. In specific, they assume that each user rates each item either $+1$ (like) or $-1$ (dislike), and that the observations are noisy so that they can be flipped with probability $\theta \in (0,\frac{1}{2})$. This assumption can be interpreted as each user rates each item $+1$ with probability $1-\theta$ or $\theta$. 
Note that this parametric model is very limited, so the discrepancy between the model and the real world occurs; if a user likes item a, b and c with probability $\sfrac{1}{4}, \sfrac{1}{3}$ and $\sfrac{3}{4}$ respectively, then the model cannot represent this case well (see Remark \ref{rmk:3} for a detailed description).

\begin{figure}[t]
    \centering
    \begin{subfigure}[b]{0.49\columnwidth}
        \includegraphics[width=\textwidth]{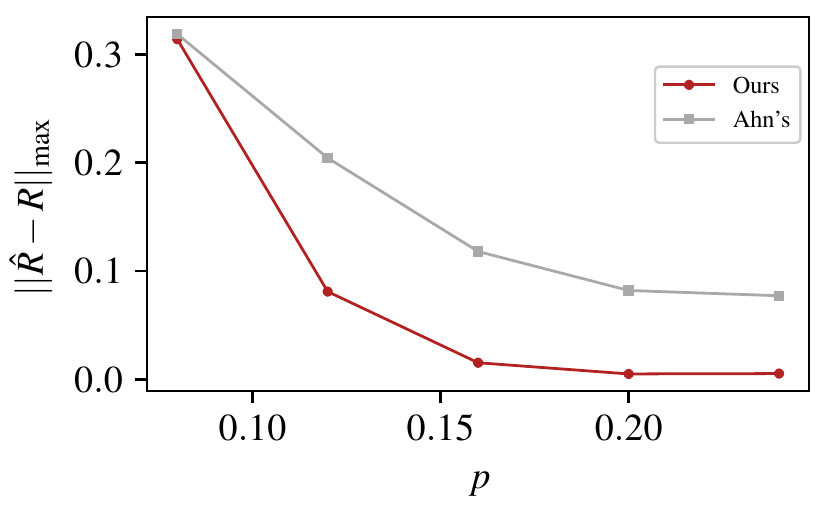}
        \caption{Synthetic $N^\Omega$ + synthetic $G$}
        \label{fig:1a}
    \end{subfigure}
    \begin{subfigure}[b]{0.49\columnwidth}
        \includegraphics[width=\textwidth]{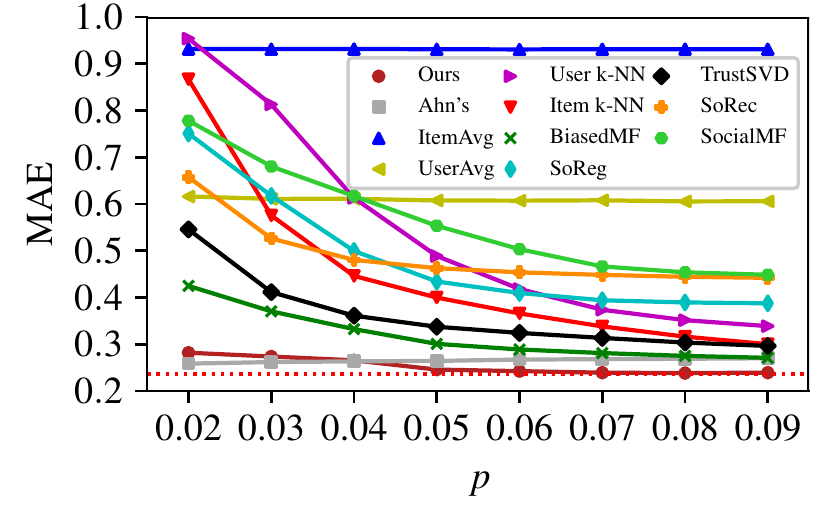}
        \caption{Synthetic $N^\Omega$ + real $G$~\cite{51inAhn}}
        \label{fig:1b}
    \end{subfigure}
    \caption{Performance comparison of various algorithms for latent preference estimation with graph side information. The $x$-axis is the probability of observing each rating ($p$), and the $y$-axis is the estimation error measured in $\|\cdot\|_\text{max}$ or the mean absolute error (MAE).
    (a) Our algorithm vs \cite{Ahn_2018} where $d=2$, $p_1 = 0.3$, $p_2 = 0.62$. \cite{Ahn_2018} performs badly due to the asymmetry of latent preference levels. (b) Our algorithm vs various algorithms proposed in the literature on real graph data and synthetic ratings. Observe that ours strictly outperforms all the existing algorithms for almost all tested values of $p$.}
    \label{fig:1}
\end{figure}

\begin{table*}[t]
\caption{MAE comparison with other algorithms on real $N^{\Omega}$ + real $G$~\cite{Massa_2007, Massa_2008}}
\label{table:1}
\vspace{-0.15in}
\begin{center}
\begin{scriptsize}
\begin{sc}
\begin{tabular}{ccccccccccc}
\toprule
ItemAvg & UserAvg & User k-NN & Item k-NN & BiasedMF & SocialMF & SoRec & SoReg & TrustSVD & Ahn's & Ours\\
\midrule
0.547 & 0.731 & 0.614 & 0.664 & 0.592 & 0.591 & 0.592 & 0.576 & 0.567 & 0.567 & $\mathbf{0.540}$\\
\bottomrule
\end{tabular}
\end{sc}
\end{scriptsize}
\end{center}
\vskip -0.1in
\end{table*}
This motivates us to propose a general model that better represents real data.
Specifically, we assume that user $i$ likes item $j$ with probability $R_{ij}$, which we call user $i$'s latent preference level on item $j$, and $R_{ij}$ belongs to the discrete set $\{p_1, \dots, p_d\}$ where $d\ge 1$ and $0< p_1 < \dots < p_d < 1$.    
As $d$ can be any positive integer, our generalized model can reflect various preference levels on different items.    
In addition to that, we assume that the social graph information follows the Stochastic Block Model (SBM)~\cite{Holland_1983}, and the social graph is correlated with the latent preference matrix $R$ in a specific way, which we will detail in Sec.~\ref{sec:3}.
Under this highly generalized model, we fully characterize the optimal sample complexity required for estimating the latent preference matrix $R$.
To the best of our knowledge, this work is the first theoretical work that shows the optimal sample complexity of latent preference estimation with graph side information without making strict assumptions on the rating generation model, made in all the prior work~\cite{Ahn_2018,8636058,Elmahdy2020Hierarchical,changhonew}.
We also develop an algorithm with low computational complexity, and our algorithm is shown to consistently outperform all the proposed algorithms in the literature including those of~\cite{Ahn_2018} on synthetic/real data.

To further highlight the limitation of the proposed algorithms developed under the strict assumptions used in the literature, we present various experimental results in Fig.~\ref{fig:1}.
(We will revisit the experimental setting in Sec.~\ref{sec:6}.)
In Fig.~\ref{fig:1a}, we compare our algorithm with that of~\cite{Ahn_2018} on synthetic rating $N^\Omega$ and synthetic graph $G$. 
Here, we set $d=2, p_1 = 0.3,p_2 = 0.62$, i.e., the symmetry assumption $p_1 + p_2 = 1$ does not hold anymore.
We can see that our algorithm significantly outperforms the algorithm of~\cite{Ahn_2018} in terms of the estimation error for all tested values of $p$, where $p$ denotes the probability of observing each rating.
This clearly shows that their algorithm quickly breaks down even when the modeling assumption is just slightly off.
Shown in Fig.~\ref{fig:1b} is the performance of various algorithms on synthetic rating/real graph, and we observe that the estimation error of~\cite{Ahn_2018} increases as the observation rate $p$ increases unlike all the other algorithms. 
(We discuss why this unexpected phenomenon happens in more details in Sec.~\ref{sec:6}.)
On the other hand, our algorithm outperforms all the existing baseline algorithms for almost all tested values of $p$ and does not exhibit any unexpected phenomenon. In Table~\ref{table:1}, we observe that our algorithm outperforms all the other algorithms even on real rating/real graph data,
although the improvement is not significant than the one for synthetic rating/real graph data.
These results demonstrate the practicality of our new algorithm, which is developed under a more realistic model without limiting assumptions.

This paper is organized as follows. 
Related works are given in Sec.~\ref{sec:2}. We propose a generalized problem formulation for a recommender system with social graph information in Sec.~\ref{sec:3}. Sec.~\ref{sec:4} characterizes the optimal sample complexity with main theorems. In Sec.~\ref{sec:5}, we propose an algorithm with low time complexity and provide a theoretical performance guarantee. In Sec.~\ref{sec:6}, experiments are conducted on synthetic and real data to compare the performance between our algorithm and existing algorithms in the literature. Finally, we discuss our results in Sec.~\ref{sec:7}.
All the proofs and experimental details are given in the appendix.
\subsection{Notation}

Let $[n]=\{1,2,\dots,n\}$ where $n$ is a positive integer, and let $\mathbb{1}(\cdot )$ be the indicator function. An undirected graph $G$ is a pair $(V,E)$ where $V$ is a set of vertices and $E$ is a set of edges. For two subsets $X$ and $Y$ of the vertex set $V$, $e(X,Y)$ denotes the number of edges between $X$ and $Y$.

\section{Related Work}\label{sec:2}

Collaborative filtering has been widely used to design recommender systems. There are two types of methods commonly used in collaborative filtering; neighborhood-based method and matrix factorization-based method. The neighborhood-based approach predicts users' ratings by finding similarity between users~\cite{Herlocker_1999}, or by finding similarity between items~\cite{Sarwar_2001, Linden_2003}. In the matrix factorization-based approach, it assumes users' latent preference matrix is of a certain structure, e.g., low rank, so the latent preference matrix can be decomposed into two matrices of low dimension~\cite{Rennie_2005, Salakhutdinov_2007, Salakhutdinov_2008, Agarwal_2010}. In particular, \citet{Davenport_2014} consider binary (1-bit) matrix completion and show that the maximum likelihood estimate is accurate under suitable conditions.

Since collaborative filtering relies solely on past interactions between users and items, it suffers from the cold start problem; collaborative filtering cannot provide a recommendation for new users since the system does not have enough information. A lot of works have been done to resolve this issue by incorporating social graph information into recommender systems. In specific, the social graph helps neighborhood-based method to find better neighborhood~\cite{Jamali_2009_1, Jamali_2009_2, Yang_2012, Yang_2013}. Some works add social regularization terms to the matrix factorization method to improve the performance~\cite{Cai_2011, Jamali_2010, Ma_2011, Kalofolias2014MC}.

Few works have been conducted to provide theoretical guarantees of their models that consider graph side information. \citet{Chiang_2015} consider a model that incorporates general side information into matrix completion, and provide statistical guarantees. \citet{Rao_2015} derive consistency guarantees for graph regularized matrix completion.

Recently, several works have studied the binary rating estimation problem with the aid of social graph information~\cite{Ahn_2018,8636058,Zhang_2020,Elmahdy2020Hierarchical,changhonew}.
These works characterize the optimal sample complexity as the minimum number of observed ratings for reliable recovery of a latent preference matrix under various settings, and find how much the social graph information reduces the optimal sample complexity. 
In specific, \citet{Ahn_2018} study the case where users are clustered in two equal-sized groups, and \citet{8636058} generalize the results of \cite{Ahn_2018} to the multi-cluster case. 
\citet{Zhang_2020, changhonew} study the problem where both user-to-user and item-to-item similarity graphs are available. Lastly, \citet{Elmahdy2020Hierarchical} adopt the hierarchical stochastic block model to handle the case where each cluster can be grouped into sub-clusters.
However, all of these works require strict assumptions on the rating generation model, which is too limited to well capture the real-world data.

Our problem can also be viewed as ``node label inference on SBM,'' where nodes are users, edges are for social connections, node labels are $m$-dimensional rating vectors (consisting of $-1, 0, 1$), and node label distributions are determined by the latent preference matrix. Various works have studied recovery of clusters in SBM in the presence of node labels~\cite{J.Yang_2013, Saad_2018} or edge labels~\cite{Heimlicher_2012, Jog_2015, Yun_2016}. While their goal is recovery of clusters, \citet{Jiaming_2014} study ``edge label inference on SBM'' whose goal is to recover edge label distributions as well as clusters.

\begin{remark}
\emph{While our problem shares high similarities with ``edge label'' inference on SBM, studied in~\cite{Jiaming_2014}, there exist some critical differences. 
To see the difference, consider a very sparse graph where many nodes are isolated. Edge label inference is impossible in this regime since there is no observed information about those isolated nodes (see Thm.~2 in~\cite{Jiaming_2014} for more details). 
On the other hand, in node labelled cases, we still observe information about isolated nodes from their node labels, so it is possible to infer node label distributions as long as we observe enough number of node labels.}
\end{remark}

\section{Problem Formulation}\label{sec:3}
Let $[n]$ be the set of users, and let $[m]$ be the set of items where $m$ can scale with $n$. For $i\in [n]$ and $j\in [m]$, $R_{ij}$ denotes user $i$'s \emph{latent preference level} on item $j$, that is, user $i$'s rating on item $j$ is $+1$ (like) with probability $R_{ij}$ or $-1$ (dislike) with probability $1-R_{ij}$. We assume that latent preference levels take values in the discrete set $\{p_1, p_2, \dots , p_d\}$ where $d\ge 1$ and $0< p_1 < \dots < p_d < 1$. The \emph{latent preference matrix} $R$ is the $n\times m$ matrix whose $(i,j)$-th entry is $R_{ij}$. The \emph{latent preference vector} of user $i$ is the $i$-th row of $R$.

We further assume that $n$ users are clustered into $K$ clusters, and the users in the same cluster have the same latent preference vector. More precisely, let $C:[n]\rightarrow[K]$ be the cluster assignment function where $C(i)=k$ if user $i$ belongs to the $k$-th cluster. The inverse image $C^{-1}(\{k\})$ is the set of users whose cluster assignment is $k$, so the users in $C^{-1}(\{k\})$ have the same latent preference vector by the assumption. We denote the latent preference vector of the users in $C^{-1}(\{k\})$ by $u_k$ for $k\in [K]$. Note that the latent preference matrix $R$ can be completely recovered with the cluster assignment function $C:[n]\rightarrow[K]$ and the corresponding preference vectors $u_1, \dots, u_K$.

As the latent preference vector and the cluster assignment function are generally unknown in the real world, we estimate them with observed ratings on items and the social graph.

\paragraph{Observed rating matrix $N^\Omega$} We assume that we observe binary ratings of users independently with probability $p$ where $p\in [0,1]$. We denote a set of observed entries by $\Omega$ which is a subset of $[n]\times [m]$.
Then, the $(i,j)$-th entry of the observed rating matrix $N^{\Omega}$ is defined by user $i$'s rating on item $j$ if $(i,j)\in\Omega$ and $0$ otherwise.
That is, $(N^{\Omega})_{ij} \overset{\text{iid}}{\sim} \text{Bern}(p)\cdot (2 \text{Bern}(R_{ij})-1)$.

\begin{figure*}[t]
	\includegraphics[width=0.95\textwidth]{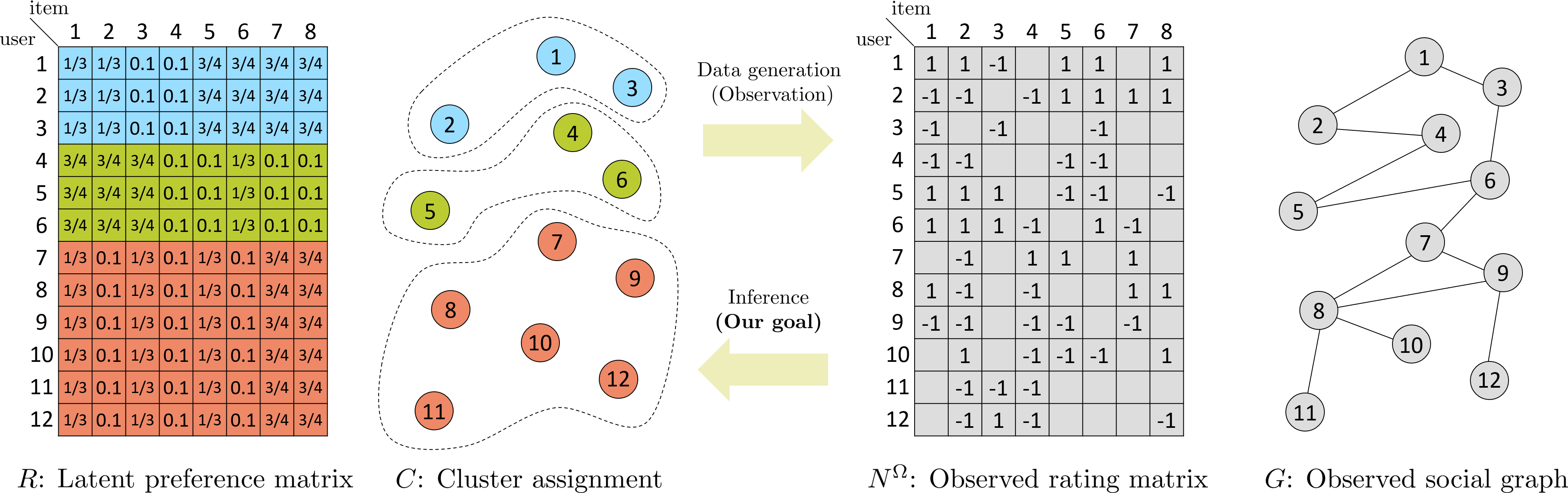}
    \centering
    \caption{A toy example of our model where $n=12, m=8, d=3, p_1 = 0.1, p_2 = \frac{1}{3}, p_3 = \frac{3}{4}, K=3, C^{-1}(\{1\})=\{1,2,3\}, C^{-1}(\{2\})=\{4,5,6\}, C^{-1}(\{3\})=\{7,8,9,10,11,12\}, p=0.5, \alpha = 0.6, \beta = 0.1$.}
    \label{fig:model}
\end{figure*}

\paragraph{Observed social graph $G$} We observe the social graph $G=([n], E)$ on $n$ users, and we further assume that the graph is generated as per the stochastic block model (SBM)~\cite{Holland_1983}. Specifically, we consider the symmetric SBM. If two users $i$ and $j$ are from the same cluster, an edge between them is placed with probability $\alpha$, independently of the others. If they are from the different clusters, the probability of having an edge between them is $\beta$, where $\alpha \ge \beta$.

Fig.~\ref{fig:model} provides a toy example that visualizes how our observation model is realized by the latent preference matrix and the cluster assignment. Given this observation model, the goal of latent preference estimation with graph side information is to find an estimator $\psi(N^{\Omega}, G)$ that estimates the latent preference matrix $R$.

\begin{remark}
\emph{(Why binary rating?) Binary rating has its critical applications such as click/impression-based advertisement recommendation, in which only $-1$ (shown, not clicked), $0$ (not shown), $1$ (shown, clicked) information is available. Moreover, binary rating is gaining increasing interests in the  industry due to its simplicity and robustness. This is precisely why Youtube and Netflix, two of the largest media recommendation systems, have discarded their ``star rating systems'' and employed binary ratings in 2009~\cite{Youtube_2009} and in 2017~\cite{Netflix_2017}, respectively.}
\end{remark}

\begin{remark}\label{rmk:3}
\emph{\citet{Ahn_2018} assume that each user rates each item either $+1$ (like) or $-1$ (dislike), and that the observations are noisy so that they can be flipped with probability $\theta \in (0,\frac{1}{2})$. 
This assumption can be interpreted as each user rates each item $+1$ with probability $1-\theta$ (when the user's true rating is $+1$) or $\theta$ (when the user's true rating is $-1$). 
Therefore, our model reduces to the model of~\cite{Ahn_2018} by setting $d=2, p_1=\theta, p_2=1-\theta, K=2, |C^{-1}(\{1\})|=|C^{-1}(\{2\})|=\frac{n}{2}$.
As mentioned in Sec.~\ref{sec:1}, the parametric model used in \cite{Ahn_2018} is very limited.
For example, consider the following two latent preference matrices $R_1 =\begin{bmatrix}
\sfrac{1}{4} & \sfrac{1}{4} & \sfrac{3}{4} & \sfrac{3}{4}\\
\sfrac{1}{4} & \sfrac{1}{4} & \sfrac{1}{4} & \sfrac{1}{4}
\end{bmatrix}, R_2 =\begin{bmatrix}
\sfrac{1}{3} & \sfrac{1}{4} & \sfrac{3}{4} & \sfrac{3}{4}\\
\sfrac{1}{3} & \sfrac{1}{4} & \sfrac{1}{4} & \sfrac{1}{4}
\end{bmatrix}$ where $n=2, m=4$. Then $R_1$ can be represented by the model used in [Ahn et al., 2018] with $\theta = \frac{1}{4}$, but $R_2$ cannot be handled by their model with any choice of $\theta$ since there are more than two latent preference levels in $R_2$.}
\end{remark}

\begin{remark}
\emph{Without graph observation, our observation model reduces to a special case of the observation model for the binary ($1$-bit) matrix completion shown in Sec. 2.1. of~\cite{Davenport_2014}.}
\end{remark}

\section{Fundamental Limit on Sample Complexity}\label{sec:4}
We now characterize the fundamental limit on the sample complexity. We first focus on the two equal-sized clusters case (i.e., $K=2, |C^{-1}(\{1\})|=|C^{-1}(\{2\})|=\frac{n}{2}$) and will extend the results to the multi-cluster case.
We use $A_R$ and $B_R$ for the ground-truth clusters and $u_R$ and $v_R$ for the corresponding latent preference vectors, respectively.
We define the worst-case error probability as follows. 

\begin{definition}\label{def:1}
(Worst-case probability of error for two equal-sized clusters) Let $\gamma$ be a fixed number in $(0,1)$ and $\psi$ be an estimator that outputs a latent preference matrix in $\{p_1,p_2, \dots, p_d\}^{n\times m}$ based on $N^{\Omega}$ and $G$. 
We define the worst-case probability of error $P_e^{\gamma}(\psi):= \max\big\{\Pr(\psi(N^\Omega, G)\neq R) : R\in \{p_1, p_2, \dots , p_d\}^{n\times m}, \|u_R-v_R\|_0=\lceil \gamma m \rceil \big\}$ where $\|\cdot\|_0$ is the hamming distance.
\end{definition}

A latent preference level $p_i \in [p_1, \dots, p_d]$ implies that the probability of choosing $(+1, -1)$ is $(p_i, 1-p_i)$ respectively, so it corresponds to a discrete probability distribution $(p_i, 1-p_i)$.
For two latent preference levels $p_i, p_j \in [p_1, \dots, p_d]$,
\emph{the Hellinger distance} between two discrete probability distributions $(p_i, 1-p_i)$ and $(p_j, 1-p_j)$, denoted $d_H(p_i, p_j)$, is
$$\frac{1}{\sqrt{2}}\sqrt{(\sqrt{p_i}-\sqrt{p_j})^2+(\sqrt{1-p_i}-\sqrt{1-p_j})^2}.$$
Then, \emph{the minimum Hellinger distance} of the set of discrete-valued latent preference levels $\{p_1, \dots, p_d\}$, denoted $d^{\min}_H$, is 
$$\min \{d_H(p_i, p_j) : i \neq j \in [d]\}. $$
Below is our main theorem that characterizes a sharp threshold of $p$, the probability of observing each rating of users, for reliable recovery as a function of $n, m, \gamma, \alpha, \beta, d^{\min}_H$.

\begin{theorem}\label{thm:1}
Let $K=2, |C^{-1}(\{1\})|=|C^{-1}(\{2\})|=\frac{n}{2}, \gamma \in (0,1)$, $m=\omega (\log n)$, $\log m=o(n)$, 
$I_s\footnote{\citet{Ahn_2018} made implicit assumptions that $\alpha,\beta\rightarrow 0$ and $\frac{\alpha}{\beta}\rightarrow 1$ as $n\rightarrow \infty$. These assumptions are used when they approximate $-2\log ( 1- d^2_H(\alpha, \beta) )=(1+o(1))(\sqrt{\alpha}-\sqrt{\beta})^2$. The approximation does not hold without above assumptions, in explicit, $-2\log ( 1- d^2_H(\alpha, \beta) )=(\sqrt{\alpha}-\sqrt{\beta})^2 \left\{\frac{(\sqrt{\alpha}+\sqrt{\beta})^2}{4\beta(1-\beta)}+o(1)\right\}$ (see the appendix for the derivation). 
The MLE achievability part of our theorem does not make any implicit assumptions, and the results hold for any $\alpha$ and $\beta$ with our modified definition of $I_s:=-2\log ( 1- d^2_H(\alpha, \beta) )$.}:=-2\log \big( 1- d^2_H(\alpha, \beta) \big).$
Then, the following holds for arbitrary $\epsilon>0$.

\noindent (I) if $p\ge \frac{1}{(d^{\min}_H)^2} \max\left\{\frac{(1+\epsilon) \log n - \frac{n}{2}I_s}{\gamma m},\frac{(1+\epsilon)2\log m}{n}\right\}$,
then there exists an estimator $\psi$ that outputs a latent preference matrix in $\{p_1,p_2, \dots, p_d\}^{n\times m}$ based on $N^{\Omega}$ and $G$ such that $P_e^{\gamma}(\psi)\rightarrow 0$ as $n \rightarrow \infty$.

\noindent (II) if $p\le \frac{1}{(d^{\min}_H)^2} \max\left\{\frac{(1-\epsilon) \log n - \frac{n}{2}I_s}{\gamma m},\frac{(1-\epsilon)2\log m}{n}\right\}$ and $\alpha = O(\frac{\log n}{n})$, then $P_e^{\gamma}(\psi)\nrightarrow 0$ as $n \rightarrow \infty$ for any $\psi$.
\end{theorem}

\begin{remark}
\emph{We note that our technical contributions lie in the proof of Thm.~\ref{thm:1}. In specific, we find the upper bound of the probability of error in Lem.~3 by using the results of Lem.~1,~2, and we made nontrivial technical contributions as we need to handle a significantly larger set of candidate latent preference matrices.}
\end{remark}

Theorem~\ref{thm:1} shows that $\frac{1}{(d^{\min}_H)^2} \max\left\{\frac{\log n - \frac{n}{2}I_s}{\gamma m},\frac{2\log m}{n}\right\}$ can be used as a sharp threshold for reliable recovery of the latent preference matrix. As $nmp$ is the expected number of observed entries, we define the optimal sample complexity for two-cluster cases as follows.

\begin{definition}\label{def:2}
$p^*_{(\gamma)}:=\frac{1}{(d^{\min}_H)^2} \max\left\{\frac{\log n - \frac{n}{2}I_s}{\gamma m},\frac{2\log m}{n}\right\}$ denotes the optimal observation rate. Then $nmp^*_{(\gamma)} = \frac{1}{(d^{\min}_H)^2} \max\left\{\frac{1}{\gamma} (  n\log n - \frac{1}{2} n^2 I_s),2m\log m\right\}$ denotes the optimal sample complexity for two-cluster cases.
\end{definition}

The optimal sample complexity for two-cluster cases is written as a function of $p_1, ..., p_d$, so the dependency on $d$ is implicit. To see the dependency clearly, we can set $p_i = \frac{i}{d+1} $. This gives us $p^*_{(\gamma)} \approx 2d^2 \max \Big\{ \frac{\log n- \frac{n}{2} I_s}{\gamma m}, \frac{2\log m}{n} \Big\}$, and $p^*_{(\gamma)}$ increases as a quadratic function of $d$.

\begin{remark}
\emph{(How does the graph information reduce the optimal sample complexity?) One can observe that $I_s$ decreases as $\alpha$ and $\beta$ get closer to each other, and $I_s=0$ when $\alpha = \beta$. Hence $I_s$ measures the quality of the graph information. If we consider the case that does not employ the graph information, it is equivalent to the case of $\alpha = \beta$ ($I_s=0$) in our model, thereby getting the optimal sample complexity of $\frac{1}{(d^{\min}_H)^2} \max\left\{\frac{1}{\gamma} n\log n ,2m\log m\right\}$. Therefore, exploiting the graph information results in the reduction of the optimal sample complexity by $\frac{1}{(d^{\min}_H)^2} \frac{1}{2 \gamma}n^2 I_s$ provided that $\frac{1}{\gamma} n\log n > 2m\log m$. Note that the optimal sample complexity stops decreasing when $I_s$ is larger than a certain threshold which implies the gain is saturated.}
\end{remark}

\begin{remark}
\emph{If we set $d=2, p_1=\theta, p_2=1-\theta$, then $(d^{\min}_H)^2=1-2\sqrt{\theta (1-\theta)}=(\sqrt{1-\theta}-\sqrt{\theta})^2$. Plugging this into the result of Theorem~\ref{thm:1}, we get $p^*_{(\gamma)}=\frac{1}{(\sqrt{1-\theta}-\sqrt{\theta})^2}\max\left\{\frac{\log n - \frac{n}{2}I_s}{\gamma m},\frac{2\log m}{n}\right\}$, recovering the main theorem of \cite{Ahn_2018} as a special case of our result.}
\end{remark}

Our results can be extended to the case of multiple (possibly unequal-sized) clusters by combining the technique developed in Theorem~\ref{thm:1} and the technique of~\cite{8636058}. Suppose $d_H(p_i, p_j)$ achieves the minimum Hellinger distance when $p_i = p_{d_0}, p_j = p_{d_0 +1}$.
Define $\mathbb{p}:\{p_1, \dots, p_d\}^m \rightarrow \{p_{d_0}, p_{d_0+1}\}^m$ that maps a latent preference vector to a latent preference vector consisting of latent preference levels $\{p_{d_0}, p_{d_0+1}\}$. In explicit, $\mathbb{p}$ sends each coordinate $x_i$ of a latent preference vector to $p_{d_0}$ if $x_i \le p_{d_0}$; $p_{d_0+1}$ if $x_i \ge p_{d_0+1}$. We now present the extended result below, while deferring the the proof to the appendix.

\begin{theorem}\label{thm:2} Let $m=\omega (\log n)$, $\log m=o(n)$, $c_k = |C^{-1}(\{k\})|, c_{i,j}=\frac{c_i +c_j}{2}$, $ \underset{n\rightarrow \infty}{\lim\inf} \frac{c_k}{n}>0~\text{for all}~k\in [K]$, $ \underset{m\rightarrow \infty}{\lim\inf} \frac{\|\mathbb{p}(u_i)-\mathbb{p}(u_j)\|_0}{m}>0~\text{for all}~i\neq j \in[K]$. Then, the following holds for arbitrary $\epsilon>0$.

\noindent (I) (achievability)\\ If $p\ge$ $\frac{1}{(d^{\min}_H)^2} \max \Big\{\underset{i\neq j \in [K]}{\max}\Big\{\frac{(1+\epsilon)\log n - c_{i,j}I_s}{\|\mathbb{p}(u_i)-\mathbb{p}(u_j)\|_0}\Big\},\underset{k\in [K]}{\max}\Big\{\frac{(1+\epsilon)\log m}{c_k}\Big\}\Big\}$,
then there exists an estimator $\psi$ such that $\Pr(\psi(N^\Omega, G)\neq R)\rightarrow 0$ as $n \rightarrow \infty$.

\noindent (II) (impossibility)\\
Suppose $R\in \{p_{d_0},p_{d_0+1}\}^{n\times m}$, $\alpha = O(\frac{\log n}{n})$. If $p\le$ $\frac{1}{(d^{\min}_H)^2} \max \Big\{\underset{i\neq j \in [K]}{\max}\Big\{\frac{(1-\epsilon)\log n - c_{i,j}I_s}{\|\mathbb{p}(u_i)-\mathbb{p}(u_j)\|_0}\Big\}, \underset{k\in [K]}{\max}\Big\{\frac{(1-\epsilon)\log m}{c_k}\Big\}\Big\}$, then $\Pr(\psi(N^\Omega, G)\neq R)\nrightarrow 0$ as $n \rightarrow \infty$ for any $\psi$.
\end{theorem}

\begin{remark}
\emph{One can observe that Theorem~\ref{thm:1} is a special case of Theorem~\ref{thm:2} by setting $K=2, c_1=c_2=\frac{n}{2}, \|\mathbb{p}(u_1)-\mathbb{p}(u_2)\|_0=\lceil \gamma m \rceil$.}
\end{remark}

\begin{remark}
\emph{In light of Theorem 14 in~\cite{Abbe_2018}, we conjecture that our results can be extended to asymmetric SBMs with a new definition of $I_s$ involving Chernoff-Hellinger divergence.}
\end{remark}

\section{Our Proposed Algorithm}\label{sec:5}
In this section, we develop a computationally efficient algorithm that can recover the latent preference matrix $R$ without knowing the latent preference levels $\{p_1, \dots, p_d\}$. We then provide a theoretical guarantee that if $p\ge \frac{1}{(d^{\min}_H)^2} \max\left\{\frac{(1+\epsilon) \log n - \frac{n}{2}I_s}{\gamma m},\frac{(1+\epsilon)2\log m}{n}\right\}$ for some $\epsilon >0$, then the proposed algorithm recovers the latent preference matrix with high probability. Now we provide a high-level description of our algorithm while deferring the pseudocode to the appendix.
\\
\vspace{0.5in}

\noindent \textbf{Algorithm description}
\vspace{-0.08in}
\paragraph{Input:} $N^{\Omega}\in \{-1, 0, +1\}^{n\times m}$, $G=([n], E)$, $K$, $d$, $\ell_{\max}$
\vspace{-0.12in}    
\paragraph{Output:} Clusters of users $A_1^{(\ell_{\max})},\dots, A_K^{(\ell_{\max})}$, latent preference vectors $\hat{u_1}^{(\ell_{\max})},\dots, \hat{u_K}^{(\ell_{\max})}$
\vspace{-0.12in}    
\paragraph{Stage 1. Partial recovery of clusters} We run a spectral method~\cite{Gao_2017} on $G$ to get an initial clustering result $A_1^{(0)},\dots, A_K^{(0)}$. Unless $\alpha$ is too close to $\beta$, this stage will give us a reasonable clustering result, with which we can kick-start the entire estimation procedure. Other clustering algorithms~\cite{Abbe_2015, Chin_2015,Krzakala_2013,Lei_2015} can also be used for this stage.
\vspace{-0.12in}    
\paragraph{Stage 2} We iterate \textbf{Stage 2-(i)} and \textbf{Stage 2-(ii)} for $\ell= 1,\dots, \ell_{\max}$.
\vspace{-0.12in}
\paragraph{Stage 2-(i). Recovery of latent preference vectors} In the $\ell$-th iteration step, this stage takes the clustering result $A_1^{(\ell -1)},\dots, A_K^{(\ell -1)}$ and rating data $N^{\Omega}$ as input and outputs the estimation of latent preference vectors $\hat{u_1}^{(\ell)},\dots, \hat{u_K}^{(\ell)}$.
    
First, for each cluster $A_k^{(\ell -1)}$, we estimate the latent preference levels for $d\lceil \log m \rceil$ randomly chosen items with replacement. The estimation of a latent preference level can be easily done by computing the ratio of ``the number of $+1$ ratings'' to ``the number of observed ratings (i.e., nonzero ratings)'' for each item within the cluster $A_k^{(\ell -1)}$. Now we have $K d\lceil \log m \rceil$ number of estimations, and these estimations will be highly concentrated around the latent preference levels $p_1,\dots, p_d$ under our modeling assumptions (see the appendix for the mathematical justifications). After running a distance-based clustering algorithm (see the pseudocode for details), we take the average within each cluster to get the estimations $\hat{p_1}^{(\ell)},\dots, \hat{p_d}^{(\ell)}$.

Given the estimations $\hat{p_1}^{(\ell)},\dots, \hat{p_d}^{(\ell)}$ and the clustering result $A_1^{(\ell -1)},\dots, A_K^{(\ell -1)}$, we estimate latent preference vectors $\hat{u_1}^{(\ell)},\dots, \hat{u_K}^{(\ell)}$ by maximizing the likelihood of the observed rating matrix $N^{\Omega}$ and the observed social graph $G = ([n], E)$.
In specific, the $j$-th coordinate of $\hat{u_k}^{(\ell)}$ is $\underset{\hat{p_h}^{(\ell)}:h\in [d]}{\arg \min}~\hat{L}(\hat{p_h}^{(\ell)}; A^{(\ell-1)}_k, j)$ where $\hat{L}(\hat{p_h}^{(\ell)}; A^{(\ell-1)}_k, j) := \underset{i\in A_k^{(\ell-1)}}{\sum}\big\{\mathbb{1}(N^{\Omega}_{ij}=1)(-\log \hat{p_h}^{(\ell)})+$ $\mathbb{1}(N^{\Omega}_{ij}=-1) (-\log (1-\hat{p_h}^{(\ell)}))\big\}$.
\vspace{-0.12in}
\paragraph{Stage 2-(ii). Refinement of clusters} In the $\ell$-th iteration step, this stage takes the clustering result $A_1^{(\ell -1)},\dots, A_K^{(\ell -1)}$, the estimation of latent preference vectors $\hat{u_1}^{(\ell)},\dots, \hat{u_K}^{(\ell)}$, rating data~$N^{\Omega}$, graph data $G$ as input and outputs the refined clustering result $A_1^{(\ell)},\dots, A_K^{(\ell)}$.
    
We first compute $\hat{\alpha}, \hat{\beta}$ that estimate $\alpha, \beta$ based on the clustering result $A_1^{(\ell -1)},\dots, A_K^{(\ell -1)}$ and the number of edges within a cluster and across clusters. Let $A_k^{(\ell -1, 0)} := A_k^{(\ell -1)}$ for $k\in [K]$.
Then $A_k^{(\ell -1, 0)}$'s are iteratively refined by $T = \lceil \log_2 n \rceil$ times of refinement steps as follows.
    
Suppose we have a clustering result $A_k^{(\ell -1, t-1)}$'s from the $(t-1)$-th refinement step where $t= 1, \dots, T$.
Given the estimations $\hat{\alpha}, \hat{\beta}$, the estimated latent preference vectors $\hat{u_1}^{(\ell)},\dots, \hat{u_K}^{(\ell)}$, and the clustering result $A_1^{(\ell -1, t-1)},\dots, A_K^{(\ell -1, t-1)}$, we find the refined clustering result $A_1^{(\ell -1, t)},\dots, A_K^{(\ell -1, t)}$  by updating each user's affiliation.
Specifically, for each user $i$, we put user~$i$ to $A_{k^*}^{(\ell -1, t)}$ where $k^* := \underset{k\in [K]}{\argmin}~\hat{L}(A^{(\ell -1, t-1)}_k; i)$ and $\hat{L}(A^{(\ell -1, t-1)}_k; i):=$ 
$- \underset{k'\neq k}{\sum}\Big\{\log(\hat{\beta})e(\{i\},A_{k'}^{(\ell -1, t-1)})+\log(1-\hat{\beta})\big\{|A_{k'}^{(\ell -1, t-1)}|-e(\{i\},A_{k'}^{(\ell -1, t-1)})\big\}\Big\}-$\\
$\Big\{\underset{j:N^{\Omega}_{ij}=1}{\sum}\log(\hat{u_k}^{(\ell)})_{j}+\underset{j:N^{\Omega}_{ij}=-1}{\sum}\log(1-(\hat{u_k}^{(\ell)})_{j})\Big\}$
$-\log(\hat{\alpha})e(\{i\},A_k^{(t-1)})- \log(1-\hat{\alpha})\big\{|A_k^{(\ell -1, t-1)}|-e(\{i\},A_k^{(\ell -1, t-1)})\big\}$.\\
($(\hat{u_k}^{(\ell)})_{j}$ denotes the $j$-th coordinate of $\hat{u_k}^{(\ell)}$.)
In each refinement step, the number of mis-clustered users will decrease provided that estimations $\hat{u_k}^{(\ell)}$'s, $\hat{\alpha}, \hat{\beta}$ are close enough to their true values (see the appendix for the mathematical justifications).
    
After $T$ times of refinement steps, we let $A_k^{(\ell)} := A_k^{(\ell - 1, T)}$ for $k\in [K]$. Finally, this stage outputs the refined clustering result $A_1^{(\ell)},\dots, A_K^{(\ell)}$.

\begin{remark}
\emph{The computational complexity of our algorithm can be computed as follows; $O(|E|\log n)$ for \textbf{Stage 1} via the power method~\cite{Boutsidis_2015}, $O(|\Omega|)$ for \textbf{Stage 2-(i)}, $O((|\Omega|+|E|) \log n)$ for \textbf{Stage 2-(ii)}. As $\ell_{max}$ is constant, the linear factor of $\ell_{max}$ is omitted in the computational complexity of  \textbf{Stage 2-(i),(ii)}.  Overall, our algorithm has low computational complexity of $O((|\Omega|+|E|)\log n)$.}
\end{remark}

\begin{remark}
\emph{We note that our technical contributions lie in the analysis of \textbf{Stage 2-(i)} while the analysis of \textbf{Stage~1}  and \textbf{Stage 2-(ii)} is similar to those in~\cite{Ahn_2018, 8636058}. In specific, we sample $O(\lceil \log m \rceil)$ number of items in \textbf{Stage 2-(i)} to get estimations of the latent preference levels and Lem.~8 ensures that those estimations are located in the $o(1)$-radius neighborhoods of ground-truth latent preference levels with high probability. Then Lem.~9 ensures that estimations of latent preference vectors converges to ground-truth latent preference vectors with high probability.}
\end{remark}

For the two equal-sized clusters case, the following theorem asserts that our algorithm will successfully estimate the latent preference matrix with high probability as long as the sampling probability is slightly above the optimal threshold. We defer the proof to the appendix.

\begin{theorem}\label{thm:3}
Let $\ell_{\max} = 1, K=2, |C^{-1}(\{1\})|=|C^{-1}(\{2\})|=\frac{n}{2}, \gamma \in (0,1)$, $m=\omega (\log n)$, $\log m=o(n)$, $(\sqrt{\alpha}-\sqrt{\beta})^2=\omega(\frac{1}{n})$, $m=O(n)$, and $\alpha = O(\frac{\log n}{n})$. 
Let $\phi_j$ be the ratio of the number of $p_j$'s among $(u_R)_{1},\dots,(u_R)_{m},(v_R)_{1},\dots,(v_R)_{m}$ to $2m$ for $j=1,\dots,d$, and assume that $\phi_j \nrightarrow 0$ as $n\rightarrow \infty$. If $$p\ge \frac{1}{(d^{\min}_H)^2} \max\left\{\frac{(1+\epsilon) \log n - \frac{n}{2}I_s}{\gamma m},\frac{2(1+\epsilon)\log m}{n}\right\}$$ for some $\epsilon >0$, then our algorithm outputs $\hat{R}$ where the following holds with probability approaching to $1$ as $n$ goes to $\infty$ : $\|\hat{R}-R\|_{\max} := \underset{(i,j)\in [n]\times [m]}{\max}|\hat{R}_{ij}-R_{ij}|=o(1)$.
\end{theorem}

\begin{remark}\label{rmk:12}
\emph{As our algorithm makes use of only graph data at Stage~1, the initial clustering result highly depends on the quality of graph data $I_s$.
In the extreme cases where only rating data are available, Stage~1 will output a meaningless clustering result.
As the performance of Stage~2 depends on the success of Stage~1, our algorithm may not work well even if the observation rate $p$ is above the optimal rate.
In Sec.~\ref{sec:E}, we suggest an alternative algorithm, which utilizes both rating and graph data at Stage~1.
Analyzing the performance of this new algorithm is an interesting open problem.
}
\end{remark}

\section{Experimental Results}\label{sec:6}
In this section, we run several experiments to evaluate the performance of our proposed algorithm.
Denoting by $\hat{R}$ the output of an estimator, the estimation quality is measured by the max norm of the error matrix, i.e., $\|\hat{R}-R\|_{\max} := \underset{(i,j)\in [n]\times [m]}{\max}|\hat{R}_{ij}-R_{ij}|$.
For each observation rate $p$, we
generate synthetic data $(N^{\Omega}, G)$ $100$ times at random and then report the average errors.

\subsection{Non-asymptotic Performance of Our Algorithm }\label{sec:6:1}

Shown in Fig.~\ref{:fig:2a} is the probability of error $\Pr( \psi_1(N^{\Omega},G)\neq R)$ of our algorithm for $(n,m,K,d) = (10000,5000,2,3)$ and various combinations of $(I_s, p)$. To measure $\Pr( \psi_1(N^{\Omega},G)\neq R)$, we allow our algorithm to have access to the latent preference levels $(p_1, p_2, p_3) = (0.2, 0.5, 0.7)$ in Stage 2.  We draw $p^*_{\gamma}$ as a red line. While the theoretical guarantee of our algorithm is valid when $n, m$ go to $\infty$, Fig.~\ref{:fig:2a} shows that Theorem~\ref{thm:1} predicts the optimal observation rate $p^*_{\gamma}$ with small error for sufficiently large $n, m$. One can observe a sharp phase transition around $p^*_{\gamma}$.

\begin{figure}[t]
    \centering
    \begin{subfigure}[b]{0.24\columnwidth}
        \includegraphics[width=\columnwidth]{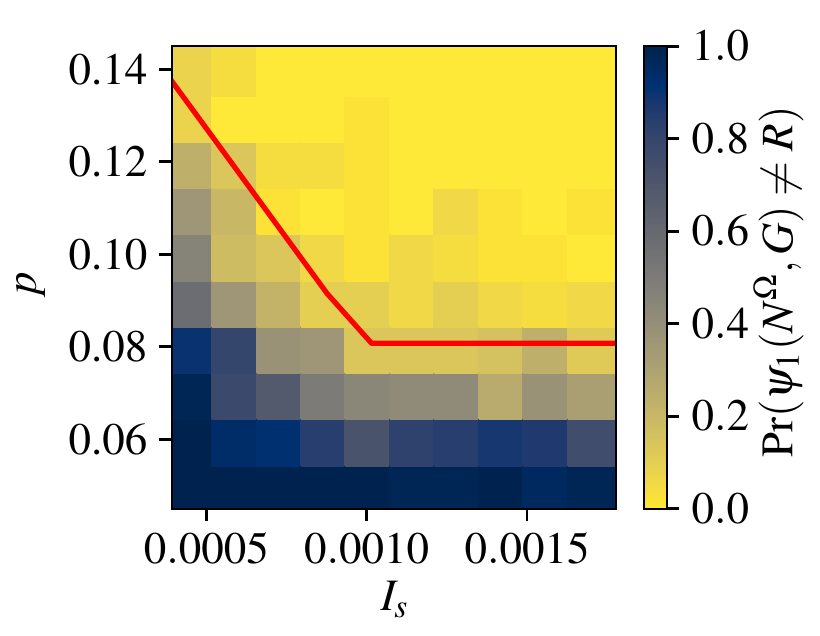}
        \caption{Phase transition}
        \label{:fig:2a}
    \end{subfigure}
    \begin{subfigure}[b]{0.24\columnwidth}
        \includegraphics[width=\columnwidth]{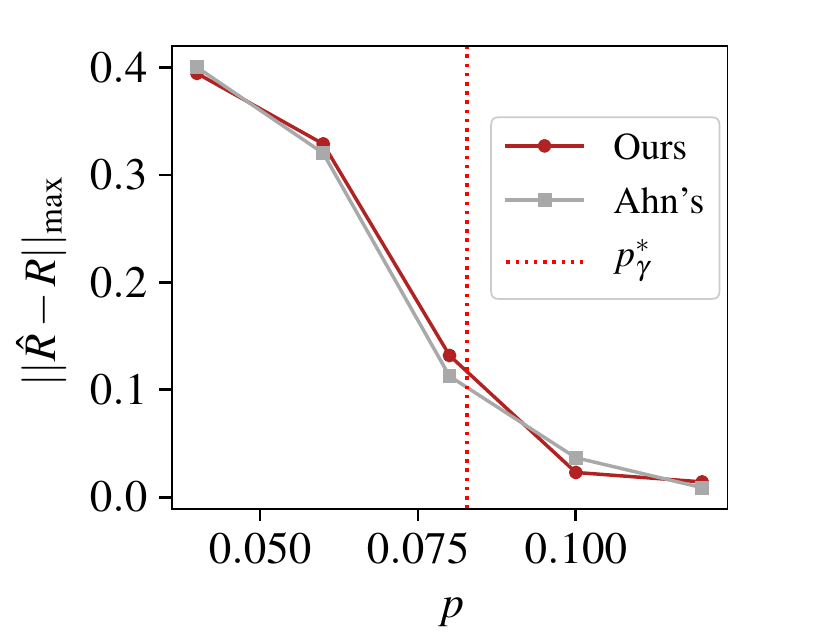}
        \caption{Symmetric levels}
        \label{:fig:2b}
    \end{subfigure}
    \begin{subfigure}[b]{0.24\columnwidth}
        \includegraphics[width=\columnwidth]{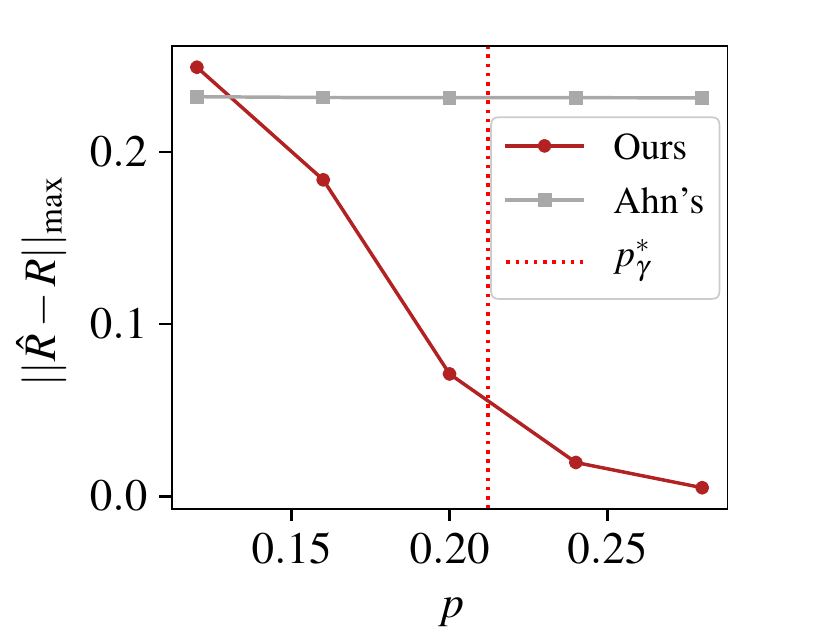}
        \caption{Asymmetric levels}
        \label{:fig:2c}
    \end{subfigure}
    \begin{subfigure}[b]{0.24\columnwidth}
        \includegraphics[width=\columnwidth]{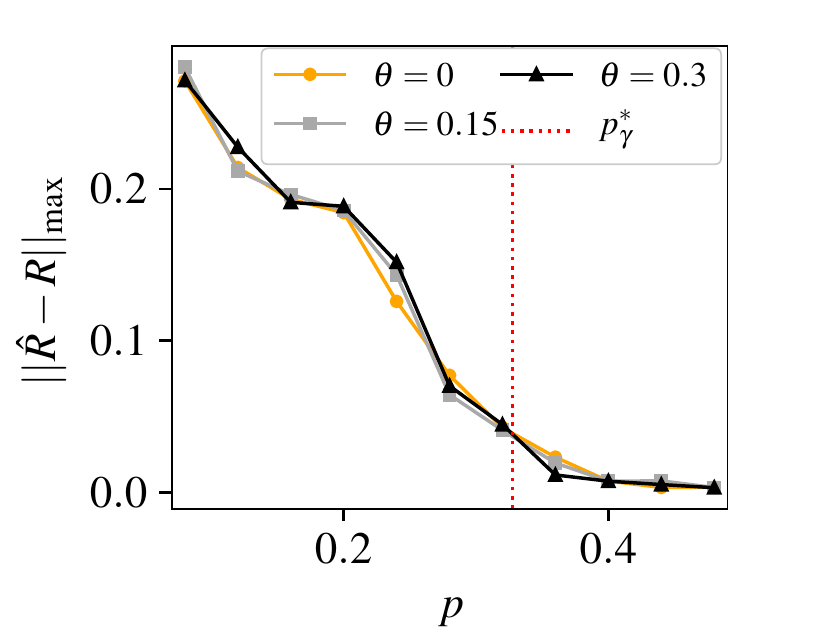}
        \caption{Noisy SBM}
        \label{:fig:2d}
    \end{subfigure} 
    \vspace{-0.1in}
    \caption{(a) Non-asymptotic performance of our algorithm. One can observe a sharp phase transition around $p^*_{\gamma}$ (b),~(c) Limitation of the symmetric level model. (b) When the latent preference levels are symmetric ($p_1 = 0.3$ and $p_2 = 0.7$), our algorithm and the algorithm proposed in \cite{Ahn_2018} achieve the same estimation errors. (c) When the latent preference levels are not symmetric ($p_1 = 0.3$ and $p_2 = 0.55$), our algorithm significantly outperforms the one proposed in \cite{Ahn_2018}. (d) Estimation error as a function of observation rate $p$ when graph data is generated as per noisy stochastic block models. Observe that our algorithm is robust to model errors.}
    \label{fig:2}
    \vspace{-0.1in}
\end{figure}

\subsection{Limitation of the Symmetric Latent Preference Levels}\label{sec:6:2}

As described in Sec.~\ref{sec:1}, the latent preference matrix model studied in~\cite{Ahn_2018} assumes that the latent preference level must be either $\theta$ or $1-\theta$ for some $\theta$, which is fully symmetric. In this section, we show that this model cannot be
applied unless the symmetry assumption perfectly holds. Let $(K, d,n,m,\gamma, \alpha, \beta)=(2, 2, 2000, 1000, \frac{1}{4}, 0.7, 0.3)$. Shown in Fig.~\ref{:fig:2b}, Fig.~\ref{fig:1a}, Fig.~\ref{:fig:2c} are the estimation errors of our algorithm and that of the algorithm proposed in~\cite{Ahn_2018} for various pairs of $(p_1, p_2)$. (1) Fig.~\ref{:fig:2b} shows the result for $(p_1, p_2) = (0.3, 0.7)$ where the latent preference levels are perfectly symmetric, and the two algorithms perform exactly the same. (2) Fig.~\ref{fig:1a} shows the result for $(p_1, p_2) = (0.3, 0.62)$ where the latent preference levels are slightly asymmetric. The estimation error of the algorithm of~\cite{Ahn_2018} is much larger than ours for all tested values of $p$. (3) Shown in Fig.~\ref{:fig:2c} are the experimental results with $(p_1, p_2) = (0.3, 0.55)$. Observe that the gap between these two algorithms becomes even larger, and the algorithm of~\cite{Ahn_2018} seems not able to output a reliable estimation of the latent preference matrix due to its limited modeling assumption.

\subsection{Robustness to Model Errors}\label{sec:6:3}

We show that while the theoretical guarantee of our algorithm holds only for a certain data generation model, our algorithm is indeed robust to model errors and can be applied to a wider range of data generation models.
Specifically, we add noise to the stochastic block model as follows.
If two users $i$ and $j$ are from the same cluster, we place an edge with probability $\alpha+q_{ij}$, independently of other edges, where $q_{ij} \overset{i.i.d.}{\sim} U[-\theta, \theta]$ for some constant $\theta$. 
Similarly, if they are from the two different clusters, the probability of having an edge between them is $\beta+q_{ij}$.
Under this \emph{noisy} stochastic block model, we generate data and measure the estimation errors with $(K, d,p_1,p_2,p_3,n,m,\gamma, \alpha, \beta)=(2, 3, 0.2, 0.5, 0.7, 2000, 1000, \frac{1}{4}, 0.7, 0.3)$, $\theta=0, 0.15, 0.3.$ Fig.~\ref{:fig:2d} shows that the performance of our algorithm is not affected by the model noise, implying the model robustness of our algorithm.
The result for $\theta=0.3$ is even more interesting since the level of noise is so large that $\alpha+q_{ij}$ can become even lower than $\beta+q_{i'j'}$ for some $(i,j)$ and $(i',j')$. 
However, even under this extreme condition, our algorithm successfully recovers the latent preference matrix.

\subsection{Real-World Data Experiments}\label{sec:6:4}

The experimental result given in Sec.~\ref{sec:6:3} motivated us to evaluate the performance of our algorithm when real-world graph data is given as graph side information. 
First, we take Facebook graph data~\cite{51inAhn} as graph side information (which has a $3$-cluster structure) and generate binary ratings as per our discrete-valued latent preference model ( $(p_1, p_2, p_3) = (0.05, 0.5, 0.95)$ ). We use $80\%$ (randomly sampled) of $\Omega$ as a training set $(\Omega_{tr})$ and the remaining $20\%$ of $\Omega$ as a test set $(\Omega_{t})$. We use mean absolute error (MAE)  $\frac{1}{|{\Omega}_{t}|} \sum_{(i,j)\in {\Omega}_{t}}\|N^{\Omega}_{ij}-(2 \hat{R}_{ij}-1)\|$ for the performance metric.\footnote{We compute the difference between $N^{\Omega}_{ij}$ and $2 \hat{R}_{ij}-1$ for fair comparison since $N^{\Omega}_{ij}\in \{\pm 1\}, \hat{R}_{ij}\in [0, 1]$.} Then we compare the performance of our algorithm with other algorithms in the literature.\footnote{We compare our algorithm with the algorithm of~\cite{Ahn_2018}, item average, user average, user k-NN (nearest neighbors), item k-NN, BiasedMF~\cite{Koren_2008}, SocialMF~\cite{Jamali_2010}, SoRec~\cite{Ma_2008}, SoReg~\cite{Ma_2011}, Trust SVD~\cite{Guo_2015}. Except for ours and that of~\cite{Ahn_2018}, we adopt implementations from LibRec~\cite{Guo_2015_lib}.} Fig.~\ref{fig:1b} shows that our algorithm outperforms other baseline algorithms for almost all tested values of~$p$. The red dotted line is the expected value of MAE of the optimal estimator (see the appendix for a detailed explanation) which means our algorithm shows near-optimal performance. Unlike other algorithms, MAE of~\cite{Ahn_2018} increases as $p$ increases. One explanation is that the algorithm of~\cite{Ahn_2018} cannot properly handle $d \ge 3$ cases due to its limited modeling assumption. 

\begin{remark}
\emph{While our algorithm shows near-optimal performance with $\ell_{\max}=1$ for synthetic data, Fig.~\ref{fig:3a} shows that our algorithm does not work well with $\ell_{\max}=1$  for real-world data. This phenomenon can be explained as follows. If $\ell_{\max}=1$, the estimations of latent preference vectors are only based on the result of the Stage 1. For real-world graph data, the clustering result of the Stage 1 may not be close to the ground-truth clusters, thereby resulting in bad estimations of latent preference vectors in Stage 2-(i). Surprisingly, our algorithm shows  near-optimal performance with $\ell_{\max}=2$ even for real-world graph data (see  Fig.~\ref{fig:1b}).
Unlike ours, the algorithm of~\cite{Ahn_2018} shows no difference between $\ell_{\max}=1$ and $\ell_{\max}=2$. }
\end{remark}

Furthermore, we evaluate the performance of our algorithm on a real rating/real graph dataset called Epinions~\cite{Massa_2007, Massa_2008}. We use $5$-fold cross-validation to determine hyperparameters. Then we compute MAE for a randomly sampled test set (with 500 iterations). Shown in Table~\ref{table:1} are MAE's for various algorithms.
Although the improvement is not significant than the one for synthetic rating/real graph data,
our algorithm outperforms all the other algorithms.
Note that all the experimental results presented in the prior work are based on synthetic rating~\cite{Ahn_2018,8636058,changhonew}, and this is the first real rating/real graph experiment that shows the practicality of binary rating estimation with graph side information.

\begin{figure}[t]
    \centering
    \begin{subfigure}{0.32\columnwidth}
    \captionsetup{justification=centering, margin = {17pt, 0pt}}
    \includegraphics[width=\textwidth]{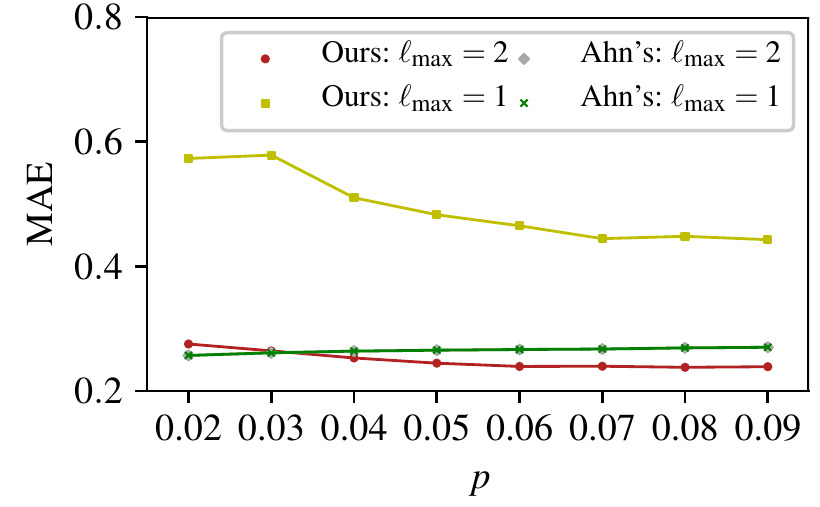}
    \subcaption{$\ell_{\max} = 1$ vs $\ell_{\max} = 2$}
    \label{fig:3a}
    \end{subfigure}    
    \begin{subfigure}{0.32\columnwidth}
    \captionsetup{justification=centering, margin = {30pt, 0pt}}
    \includegraphics[width=\textwidth]{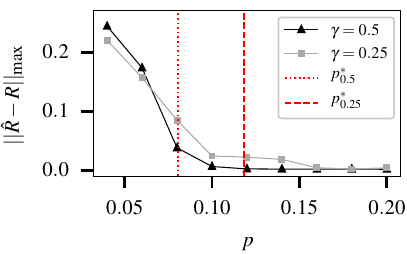}
    \subcaption{$\gamma = 0.5$ vs $\gamma = 0.25$}
    \label{fig:3b}
    \end{subfigure}
    \begin{subfigure}{0.32\columnwidth}
    \captionsetup{justification=centering, margin = {30pt, 0pt}}
    \includegraphics[width=\textwidth]{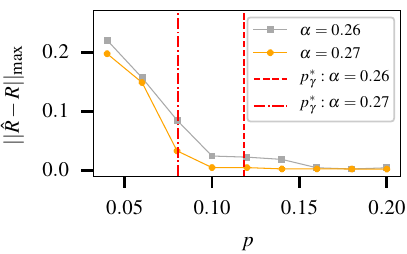}
    \subcaption{$\alpha = 0.26$ vs $\alpha = 0.27$} \label{fig:3c}
    \end{subfigure}
    \vspace{-0.1in}
    \caption{(a) Estimation error as a function of observation rate $p$ for the Facebook graph data~\cite{51inAhn} with different values of $\ell_{\max}$. (b), (c) Estimation error as a function of observation probability $p$ with different values of $\gamma$ and $\alpha$.    The $x$-axis is the probability of observing each rating ($p$), and the $y$-axis is the estimation error measured in $\|\cdot\|_\text{max}$.
}
    \label{fig:3}
    \vspace{-0.1in}
\end{figure}

\subsection{Estimation Error with Different Values of $\gamma$ and $I_s$}\label{sec:6:5}

We corroborate Theorem~\ref{thm:3}.
More specifically, we observe how the estimation error behaves as a function of $p$ when $\gamma$ and $(\alpha, \beta)$ varies. 
Let $d=3, p_1=0.2, p_2=0.5, p_3=0.7, n=10000, m=5000$. 
We first compare cases for $(\alpha, \beta, \gamma)=(0.26, 0.23,0.5)$ and $(0.26, 0.23,0.25)$.
Shown in Fig.~\ref{fig:3b} is the estimation error as a function of $p$. 
We draw $p^{*}_{\gamma}$ as  dotted vertical lines.
One can see from the figure that the estimation error for $(\alpha, \beta, \gamma)=(0.26, 0.23,0.5)$ is lower than that for $(\alpha, \beta, \gamma)=(0.26, 0.23,0.25)$ for all tested values of $p$. 
This can be explained by the fact that $p^{*}_{\gamma}$ decreases as $\gamma$ increases, as stated in Theorem~\ref{thm:3}. 
We also compare cases for $(\alpha, \beta, \gamma)=(0.26, 0.23,0.25)$ and $(0.27, 0.23,0.25)$. 
Note that the only difference between these cases is the value of $\alpha$. 
By Theorem~\ref{thm:3}, we have $p^{*}_{\gamma}=0.118$ for the former case, and $p^{*}_{\gamma}=0.081$ for the latter case. 
That is, a larger value of $\alpha$ implies a higher quality of graph side information, i.e., the graph side information is more useful for predicting the latent preference matrix $R$. 
Fig.~\ref{fig:3c} shows the estimation error as a function of $p$, and we can see that even a small increase in the quality of the graph can result in a significant decrease in $p^{*}_{\gamma}$. 

\vspace{-0.1in}
\section{Conclusion}\label{sec:7}
We studied the problem of estimating the latent preference matrix whose entries are discrete-valued given a partially observed binary rating matrix and graph side information.
We first showed that the latent preference matrix model adopted in existing works is highly limited, and proposed a generalized data generation model.
We characterized the optimal sample complexity that guarantees perfect recovery of latent preference matrix, and showed that this optimal complexity also serves as a tight lower bound, i.e., no estimation algorithm can achieve perfect recovery below the optimal sample complexity.
We also proposed a computationally efficient estimation algorithm.
Our analysis showed that our proposed algorithm can perfectly estimate the latent preference matrix if the sample complexity is above the optimal sample complexity.
We provided experimental results that corroborate our theoretical findings, highlight the importance of our relaxed modeling assumptions, imply the robustness of our algorithm to model errors, and compare our algorithm with other algorithms on real-world data.

\paragraph{Acknowledgements} 
This material is based upon work supported by NSF Award DMS-2023239.

\bibliography{ref}

\begin{thebibliography}{48}
\providecommand{\natexlab}[1]{#1}
\providecommand{\url}[1]{\texttt{#1}}
\expandafter\ifx\csname urlstyle\endcsname\relax
  \providecommand{\doi}[1]{doi: #1}\else
  \providecommand{\doi}{doi: \begingroup \urlstyle{rm}\Url}\fi

\bibitem[Abbe(2018)]{Abbe_2018}
Emmanuel Abbe.
\newblock Community detection and stochastic block models: Recent developments.
\newblock \emph{Journal of Machine Learning Research}, 18\penalty0
  (177):\penalty0 1--86, 2018.

\bibitem[Abbe and Sandon(2015)]{Abbe_2015}
Emmanuel Abbe and Colin Sandon.
\newblock Community detection in general stochastic block models: Fundamental
  limits and efficient algorithms for recovery.
\newblock In \emph{Proceedings of the 2015 IEEE 56th Annual Symposium on
  Foundations of Computer Science}, FOCS, page 670–688. IEEE, 2015.

\bibitem[Abbe et~al.(2016)Abbe, Bandeira, and Hall]{Abbe_2016}
Emmanuel Abbe, Afonso~S. Bandeira, and Georgina Hall.
\newblock Exact recovery in the stochastic block model.
\newblock In \emph{IEEE Transactions on Information Theory}, volume~62, pages
  471--487, 2016.

\bibitem[Agarwal and Chen(2010)]{Agarwal_2010}
Deepak Agarwal and Bee-Chung Chen.
\newblock flda: Matrix factorization through latent dirichlet allocation.
\newblock In \emph{Proceedings of the Third ACM International Conference on Web
  Search and Data Mining}, pages 91--100, 2010.

\bibitem[Ahn et~al.(2018)Ahn, Lee, Cha, and Suh]{Ahn_2018}
Kwangjun Ahn, Kangwook Lee, Hyunseung Cha, and Changho Suh.
\newblock Binary rating estimation with graph side information.
\newblock In \emph{Advances in Neural Information Processing Systems 31}, pages
  4272--4283, 2018.

\bibitem[Awasthi and Sheffet(2012)]{Awasthi_2012}
Pranjal Awasthi and Or~Sheffet.
\newblock Improved spectral-norm bounds for clustering.
\newblock In \emph{Approximation, Randomization, and Combinatorial
  Optimization. Algorithms and Techniques}, pages 37--49, 2012.

\bibitem[Boutsidis et~al.(2015)Boutsidis, Gittens, and
  Kambadur]{Boutsidis_2015}
Christos Boutsidis, Alex Gittens, and Prabhanjan Kambadur.
\newblock Spectral clustering via the power method - provably.
\newblock In \emph{Proceedings of the 32nd International Conference on
  International Conference on Machine Learning}, page 40–48, 2015.

\bibitem[Cai et~al.(2011)Cai, He, Han, and Huang]{Cai_2011}
Deng Cai, Xiaofei He, Jiawei Han, and Thomas~S. Huang.
\newblock Graph regularized nonnegative matrix factorization for data
  representation.
\newblock \emph{IEEE Transactions on Pattern Analysis and Machine
  Intelligence}, 33\penalty0 (8):\penalty0 1548--1560, 2011.

\bibitem[Center(2017)]{Netflix_2017}
Netflix~Media Center.
\newblock Goodbye stars, hello thumbs.
\newblock
  \url{https://media.netflix.com/en/company-blog/goodbye-stars-hello-thumbs},
  2017.

\bibitem[Chiang et~al.(2015)Chiang, Hsieh, and Dhillon]{Chiang_2015}
Kai-Yang Chiang, Cho-Jui Hsieh, and Inderjit~S Dhillon.
\newblock Matrix completion with noisy side information.
\newblock In \emph{Advances in Neural Information Processing Systems 28}, pages
  3447--3455, 2015.

\bibitem[Chin et~al.(2015)Chin, Rao, and Vu]{Chin_2015}
Peter Chin, Anup Rao, and Van Vu.
\newblock Stochastic block model and community detection in sparse graphs: A
  spectral algorithm with optimal rate of recovery.
\newblock In \emph{Proceedings of The 28th Conference on Learning Theory},
  volume~40, pages 391--423. PMLR, 2015.

\bibitem[Davenport et~al.(2014)Davenport, Plan, van~den Berg, and
  Wootters]{Davenport_2014}
Mark~A. Davenport, Yaniv Plan, Ewout van~den Berg, and Mary Wootters.
\newblock {1-Bit matrix completion}.
\newblock \emph{Information and Inference: A Journal of the IMA}, 3\penalty0
  (3):\penalty0 189--223, 2014.

\bibitem[Elmahdy et~al.(2020)Elmahdy, Ahn, Suh, and
  Mohajer]{Elmahdy2020Hierarchical}
Adel Elmahdy, Junhyung Ahn, Changho Suh, and Soheil Mohajer.
\newblock Matrix completion with hierarchical graph side information.
\newblock In \emph{Advances in Neural Information Processing Systems 34}, 2020.

\bibitem[Gao et~al.(2017)Gao, Ma, Zhang, and Zhou]{Gao_2017}
Chao Gao, Zongming Ma, Anderson~Y. Zhang, and Harrison~H. Zhou.
\newblock Achieving optimal misclassification proportion in stochastic block
  models.
\newblock \emph{J. Mach. Learn. Res.}, page 1980–2024, 2017.

\bibitem[Gruber(2017)]{Youtube_2009}
John Gruber.
\newblock Why youtube switched from 5-star ratings to thumbs up/down in 2009.
\newblock
  \url{https://daringfireball.net/linked/2017/03/18/youtube-thumbs-stars},
  2017.

\bibitem[Guo et~al.(2015{\natexlab{a}})Guo, Zhang, Sun, and
  Yorke-Smith]{Guo_2015_lib}
Guibing Guo, Jie Zhang, Zhu Sun, and Neil Yorke-Smith.
\newblock Librec: A java library for recommender systems.
\newblock In \emph{UMAP Workshops}, volume 1388 of \emph{CEUR Workshop
  Proceedings}, 2015{\natexlab{a}}.

\bibitem[Guo et~al.(2015{\natexlab{b}})Guo, Zhang, and Yorke-Smith]{Guo_2015}
Guibing Guo, Jie Zhang, and Neil Yorke-Smith.
\newblock Trustsvd: Collaborative filtering with both the explicit and implicit
  influence of user trust and of item ratings.
\newblock In \emph{AAAI}, pages 123--129, 2015{\natexlab{b}}.

\bibitem[Heimlicher et~al.(2012)Heimlicher, Lelarge, and
  Massoulié]{Heimlicher_2012}
Simon Heimlicher, Marc Lelarge, and Laurent Massoulié.
\newblock Community detection in the labelled stochastic block model.
\newblock \emph{NIPS Workshop: Algorithmic and Statistical Approaches for Large
  Social Networks}, 2012.

\bibitem[Herlocker et~al.(1999)Herlocker, Konstan, Borchers, and
  Riedl]{Herlocker_1999}
Jonathan~L. Herlocker, Joseph~A. Konstan, Al~Borchers, and John Riedl.
\newblock An algorithmic framework for performing collaborative filtering.
\newblock In \emph{Proceedings of the 22Nd Annual International ACM SIGIR
  Conference on Research and Development in Information Retrieval}, pages
  230--237, 1999.

\bibitem[Holland et~al.(1983)Holland, Laskey, and Leinhardt]{Holland_1983}
Paul~W. Holland, Kathryn~Blackmond Laskey, and Samuel Leinhardt.
\newblock Stochastic blockmodels: First steps.
\newblock \emph{Social networks}, 5\penalty0 (2):\penalty0 109--137, 1983.

\bibitem[Jamali and Ester(2009{\natexlab{a}})]{Jamali_2009_1}
Mohsen Jamali and Martin Ester.
\newblock Trustwalker: A random walk model for combining trust-based and
  item-based recommendation.
\newblock In \emph{Proceedings of the 15th ACM SIGKDD International Conference
  on Knowledge Discovery and Data Mining}, pages 397--406, 2009{\natexlab{a}}.

\bibitem[Jamali and Ester(2009{\natexlab{b}})]{Jamali_2009_2}
Mohsen Jamali and Martin Ester.
\newblock Using a trust network to improve top-n recommendation.
\newblock In \emph{Proceedings of the Third ACM Conference on Recommender
  Systems}, pages 181--188, 2009{\natexlab{b}}.

\bibitem[Jamali and Ester(2010)]{Jamali_2010}
Mohsen Jamali and Martin Ester.
\newblock A matrix factorization technique with trust propagation for
  recommendation in social networks.
\newblock In \emph{Proceedings of the Fourth ACM Conference on Recommender
  Systems}, pages 135--142, 2010.

\bibitem[Jog and Loh(2015)]{Jog_2015}
Varun Jog and Po-Ling Loh.
\newblock Recovering communities in weighted stochastic block models.
\newblock In \emph{53rd Annual Allerton Conference on Communication, Control,
  and Computing}, pages 1308--1315, 2015.

\bibitem[Kalofolias et~al.(2014)Kalofolias, Bresson, Bronstein, and
  Vandergheynst]{Kalofolias2014MC}
Vassilis Kalofolias, Xavier Bresson, Michael Bronstein, and Pierre
  Vandergheynst.
\newblock Matrix completion on graphs.
\newblock \emph{NIPS Workshop, Out of the Box: Robustness in High Dimension},
  2014.

\bibitem[Koren(2008)]{Koren_2008}
Yehuda Koren.
\newblock Factorization meets the neighborhood: A multifaceted collaborative
  filtering model.
\newblock In \emph{Proceedings of the 14th ACM SIGKDD International Conference
  on Knowledge Discovery and Data Mining}, page 426–434, 2008.

\bibitem[Krzakala et~al.(2013)Krzakala, Moore, Mossel, Neeman, Sly,
  Zdeborov{\'a}, and Zhang]{Krzakala_2013}
Florent Krzakala, Cristopher Moore, Elchanan Mossel, Joe Neeman, Allan Sly,
  Lenka Zdeborov{\'a}, and Pan Zhang.
\newblock Spectral redemption in clustering sparse networks.
\newblock \emph{Proceedings of the National Academy of Sciences}, 110\penalty0
  (52):\penalty0 20935--20940, 2013.

\bibitem[Lei and Rinaldo(2015)]{Lei_2015}
Jing Lei and Alessandro Rinaldo.
\newblock Consistency of spectral clustering in sparse stochastic block models.
\newblock \emph{The Annals of Statistics}, 43\penalty0 (1):\penalty0 215--237,
  2015.

\bibitem[Linden et~al.(2003)Linden, Smith, and York]{Linden_2003}
Greg Linden, Brent Smith, and Jeremy York.
\newblock Amazon.com recommendations: Item-to-item collaborative filtering.
\newblock \emph{IEEE Internet Computing 7}, pages 76--80, 2003.

\bibitem[Ma et~al.(2008)Ma, Yang, Lyu, and King]{Ma_2008}
Hao Ma, Haixuan Yang, Michael~R. Lyu, and Irwin King.
\newblock Sorec: Social recommendation using probabilistic matrix
  factorization.
\newblock In \emph{Proceedings of the 17th ACM Conference on Information and
  Knowledge Management}, page 931–940, 2008.

\bibitem[Ma et~al.(2011)Ma, Zhou, Liu, Lyu, and King]{Ma_2011}
Hao Ma, Dengyong Zhou, Chao Liu, Michael~R. Lyu, and Irwin King.
\newblock Recommender systems with social regularization.
\newblock In \emph{Proceedings of the Fourth ACM International Conference on
  Web Search and Data Mining}, pages 287--296, 2011.

\bibitem[Massa and Avesani(2007)]{Massa_2007}
Paolo Massa and Paolo Avesani.
\newblock Trust-aware recommender systems.
\newblock In \emph{Proceedings of the 2007 ACM Conference on Recommender
  Systems}, page 17–24, 2007.

\bibitem[Massa et~al.(2008)Massa, Souren, Salvetti, and Tomasoni]{Massa_2008}
Paolo Massa, Kasper Souren, Martino Salvetti, and Danilo Tomasoni.
\newblock Trustlet, open research on trust metrics.
\newblock In \emph{BIS}, 2008.

\bibitem[Rao et~al.(2015)Rao, Yu, Ravikumar, and Dhillon]{Rao_2015}
Nikhil Rao, Hsiang-Fu Yu, Pradeep~K Ravikumar, and Inderjit~S Dhillon.
\newblock Collaborative filtering with graph information: Consistency and
  scalable methods.
\newblock In \emph{Advances in Neural Information Processing Systems 28}, pages
  2107--2115, 2015.

\bibitem[Rennie and Srebro(2005)]{Rennie_2005}
Jasson D.~M. Rennie and Nathan Srebro.
\newblock Fast maximum margin matrix factorization for collaborative
  prediction.
\newblock In \emph{Proceedings of the 22nd International Conference on Machine
  Learning}, pages 713--719, 2005.

\bibitem[Saad and Nosratinia(2018)]{Saad_2018}
Hussein Saad and Aria Nosratinia.
\newblock Community detection with side information: Exact recovery under the
  stochastic block model.
\newblock \emph{IEEE Journal of Selected Topics in Signal Processing},
  12\penalty0 (5):\penalty0 944--958, 2018.

\bibitem[Salakhutdinov and Mnih(2007)]{Salakhutdinov_2007}
Ruslan Salakhutdinov and Andriy Mnih.
\newblock Probabilistic matrix factorization.
\newblock In \emph{Proceedings of the 20th International Conference on Neural
  Information Processing Systems}, pages 1257--1264, 2007.

\bibitem[Salakhutdinov and Mnih(2008)]{Salakhutdinov_2008}
Ruslan Salakhutdinov and Andriy Mnih.
\newblock Bayesian probabilistic matrix factorization using markov chain monte
  carlo.
\newblock In \emph{Proceedings of the 25th International Conference on Machine
  Learning}, pages 880--887, 2008.

\bibitem[Sarwar et~al.(2001)Sarwar, Karypis, Konstan, and Riedl]{Sarwar_2001}
Badrul Sarwar, George Karypis, Joseph Konstan, and John Riedl.
\newblock Item-based collaborative filtering recommendation algorithms.
\newblock In \emph{Proceedings of the 10th International Conference on World
  Wide Web}, pages 285--295, 2001.

\bibitem[Traud et~al.(2012)Traud, Mucha, and Porter]{51inAhn}
Amanda~L. Traud, Peter~J. Mucha, and Mason~A. Porter.
\newblock Social structure of facebook networks.
\newblock \emph{Physica A: Statistical Mechanics and its Applications},
  391\penalty0 (16):\penalty0 4165--4180, 2012.

\bibitem[Xu et~al.(2014)Xu, Massoulié, and Lelarge]{Jiaming_2014}
Jiaming Xu, Laurent Massoulié, and Marc Lelarge.
\newblock Edge label inference in generalized stochastic block models: from
  spectral theory to impossibility results.
\newblock In \emph{Proceedings of Machine Learning Research}, volume~35, pages
  903--920, 2014.

\bibitem[Yang et~al.(2013{\natexlab{a}})Yang, McAuley, and
  Leskovec]{J.Yang_2013}
Jaewon Yang, Julian McAuley, and Jure Leskovec.
\newblock Community detection in networks with node attributes.
\newblock In \emph{2013 IEEE 13th International Conference on Data Mining},
  pages 1151--1156, 2013{\natexlab{a}}.

\bibitem[Yang et~al.(2012)Yang, Steck, Guo, and Liu]{Yang_2012}
Xiwang Yang, Harald Steck, Yang Guo, and Yong Liu.
\newblock On top-k recommendation using social networks.
\newblock In \emph{Proceedings of the Sixth ACM Conference on Recommender
  Systems}, pages 67--74, 2012.

\bibitem[Yang et~al.(2013{\natexlab{b}})Yang, Guo, and Liu]{Yang_2013}
Xiwang Yang, Yang Guo, and Yong Liu.
\newblock Bayesian-inference-based recommendation in online social networks.
\newblock \emph{IEEE Trans. Parallel Distrib. Syst.}, pages 642--651,
  2013{\natexlab{b}}.

\bibitem[{Yoon} et~al.(2018){Yoon}, {Lee}, and {Suh}]{8636058}
J.~{Yoon}, K.~{Lee}, and C.~{Suh}.
\newblock On the joint recovery of community structure and community features.
\newblock In \emph{56th Annual Allerton Conference on Communication, Control,
  and Computing}, pages 686--694, 2018.

\bibitem[Yun and Proutiere(2016)]{Yun_2016}
Se-Young Yun and Alexandre Proutiere.
\newblock Optimal cluster recovery in the labeled stochastic block model.
\newblock In \emph{Proceedings of the 30th International Conference on Neural
  Information Processing Systems}, page 973–981, 2016.

\bibitem[Zhang et~al.(2020)Zhang, Suh, Suh, and Tan]{Zhang_2020}
Qiaosheng Zhang, Geewon Suh, Changho Suh, and Vincent Tan.
\newblock Mc2g: An efficient algorithm for matrix completion with social and
  item similarity graphs.
\newblock \emph{arXiv preprint arXiv:2006.04373}, 2020.

\bibitem[Zhang et~al.(2021)Zhang, Tan, and Suh]{changhonew}
Qiaosheng Zhang, Vincent Tan, and Changho Suh.
\newblock Community detection and matrix completion with social and item
  similarity graphs.
\newblock \emph{IEEE Transactions on Signal Processing}, 2021.

\end{thebibliography}
\bibliographystyle{plainnat}

\newpage
\appendix

\setcounter{theorem}{0} 
\setcounter{lemma}{0} 
\setcounter{definition}{0} 
\begin{center}
\textbf{\huge Appendix}
\end{center}
\section{Reproducing Our Simulation Results}
We provide our Python implementation of our algorithm as well as that of \cite{Ahn_2018} so that one can easily reproduce all of our experimental results.
Our code is available at \url{https://github.com/changhunjo0927/Discrete-Valued_Latent_Preference}, and one can easily reproduce any figure simply by opening the corresponding subfolder and run three or four Jupyter notebooks in order.  (Click `Run All' in each Jupyter notebook).
Then, simulation results will be saved as a figure within the subfolder.
While the values reported in our figures were the average performance over $T=100$ random runs, the default configuration in our codes is $T = 2$.
This way, one can quickly reproduce rough versions of our figures in a few minutes on a typical machine.
If one wants to reproduce more precise simulations results, one may want to change the value of $T$ from $2$ to $100$ by modifying the first cell of each Jupyter notebook.

\section{Pseudocode of Proposed Algorithm} \label{sec:B}
See Alg.~\ref{alg:1} for the pseudocode of our algorithm proposed in Sec.~\ref{sec:5}.

\begin{algorithm}
   \caption{}
   \label{alg:1}
\begin{algorithmic}
   \STATE {\bfseries Input:} $N^{\Omega}\in \{-1,0,+1\}^{n\times m}$, $G=([n], E), K, d, \ell_{\max}$
   \STATE {\bfseries Output:} Clusters of users $A_1^{(\ell_{\max})},A_2^{(\ell_{\max})}, \dots, A_K^{(\ell_{\max})}$
   \STATE \qquad \quad ~~~Latent preference vectors $\hat{u_1}^{(\ell_{\max})},\hat{u_2}^{(\ell_{\max})},\dots, \hat{u_K}^{(\ell_{\max})}$
	\STATE {\bfseries Stage 1 (Partial recovery of clusters): }
	\STATE Run a spectral method on G, and get a clustering result $A_1^{(0)}, A_2^{(0)},\dots, A_K^{(0)}$.
	\FOR{$\ell=1$ {\bfseries to} $\ell_{\max}$}
   \STATE {\bfseries Stage 2-(i) (Recovery of latent preference vectors):}
	\FOR{$k=1$ {\bfseries to} $K$}
	\FOR{$t=1$ {\bfseries to} $m_0 (:=d\lceil \log m \rceil)$}
   \STATE Sample $j^{(k)}_t \sim \text{unif}\{1,m\}$.
   \STATE $a_{j^{(k)}_t} \leftarrow \frac{\underset{i\in A_k^{(\ell-1)}}{\sum}\mathbb{1}(N_{ij^{(k)}_t}^{\Omega}=1)}{\underset{i\in A_k^{(\ell-1)}}{\sum}\mathbb{1}(N_{ij^{(k)}_t}^{\Omega}=\pm 1)}$
   \ENDFOR
   \ENDFOR
   \STATE Sort $\{a_{j^{(k)}_t} : k = 1,\dots, K, t = 1, \dots, m_0\}$ in ascending order, get $b_1,\dots,b_{Km_0}$.
   \STATE $S\leftarrow\{1,\dots,Km_0-1\}$
   \FOR{$t=1$ {\bfseries to} $d-1$}
   \STATE $z_t \leftarrow \underset{j\in S}{\arg\max}(b_{j+1}-b_{j})$
   \STATE $S\leftarrow S\setminus \{z_t\}$
   \ENDFOR
   \STATE Sort $z_1,\dots,z_{d-1}$ in ascending order, get $r_1,\dots,r_{d-1}$.
   \STATE $r_d \leftarrow Km_0, l_1 \leftarrow 1$
   \STATE $l_{i+1} \leftarrow r_i+1$ for $i=1,\dots,d-1$.
   \STATE $\hat{p_h}^{(\ell)} \leftarrow \frac{\sum_{i=l_h}^{r_h}b_i}{r_h-l_h}$ for $h=1,\dots,d$.
   \FOR{$k=1$ {\bfseries to} $K$}
   \FOR{$j=1$ {\bfseries to} $m$}
   \STATE $(\hat{u_k}^{(\ell)})_{j} \leftarrow \underset{\hat{p_h}:h\in [d]}{\arg\min}~\hat{L}(\hat{p_h}^{(\ell)}; A^{(\ell-1)}_k, j) $
     \ENDFOR
           \ENDFOR
   \STATE $\hat{u_k}^{(\ell)} \leftarrow ((\hat{u_k}^{(\ell)})_{1},\dots,(\hat{u_k}^{(\ell)})_{m})$ for $k = 1,\dots, K$.
   
   \STATE {\bfseries Stage 2-(ii) (Exact recovery of clusters):}
   \STATE $\hat{\alpha} \leftarrow \frac{\underset{k \in [K] }{\sum} e(A_k^{(\ell-1)},A_k^{(\ell-1)}) }{\underset{k \in [K] }{\sum}	\binom{ |A_k^{(\ell-1)}| }{2}}, \hat{\beta} \leftarrow \frac{ \underset{k_1\neq k_2 \in [K] }{\sum} e(A_{k_1}^{(\ell-1)},A_{k_2}^{(\ell-1)})}{ \underset{k_1\neq k_2 \in [K] }{\sum} |A_{k_1}^{(\ell-1)}| |A_{k_2}^{(\ell-1)}| }$
   \FOR{$t=1$ {\bfseries to} $T (:=\lceil \log_2 n \rceil)$}
   \FOR{$k=1$ {\bfseries to} $K$}
   \STATE $A_k^{(\ell-1, t)}\leftarrow \emptyset$
   \ENDFOR
   \FOR{$i=1$ {\bfseries to} $n$}
   \STATE $k^* \leftarrow \underset{k\in [K]}{\arg\min}~\hat{L}(A^{(\ell -1, t-1)}_k; i)$   
   \STATE $A_{k^*}^{(\ell-1, t)}\leftarrow A_{k^*}^{(\ell-1, t)}\cup \{i\}$
   \ENDFOR
   \ENDFOR
   \STATE $A_k^{(\ell)} \leftarrow A_k^{(\ell-1, T)}$ for $k = 1,\dots, K$.
   \ENDFOR
\end{algorithmic}
\end{algorithm} 
 
\section{Additional Experimental Details}
In this section, we provide additional experimental details deferred to the appendix. Let $\mathbf{1}_{k,l}$ be a $k\times l$ matrix whose entries are all equal to $1$.

\subsection{Sec.~\ref{sec:6:1}}
For Fig.~\ref{:fig:2a}, we used 
$$R=
\left[
\begin{array}{c|c|c|c}
0.2\cdot \mathbf{1}_{5000,1250} & 0.5\cdot \mathbf{1}_{5000,1250} & 0.5\cdot \mathbf{1}_{5000,1250} & 0.7\cdot \mathbf{1}_{5000,1250}\\
\hline
0.2\cdot \mathbf{1}_{5000,1250} & 0.5\cdot \mathbf{1}_{5000,1250} & 0.7\cdot \mathbf{1}_{5000,1250} & 0.5\cdot \mathbf{1}_{5000,1250}
\end{array}
\right]
$$
as a latent preference matrix, and $(K, d, \ell_{\max}) = (2,3,1)$.

\subsection{Sec.~\ref{sec:6:2}}
For Fig.~\ref{:fig:2b}, we used
$$R=
\left[
\begin{array}{c|c|c|c}
0.3\cdot \mathbf{1}_{1000,250} & 0.3\cdot \mathbf{1}_{1000,250} & 0.3\cdot \mathbf{1}_{1000,250} & 0.3\cdot \mathbf{1}_{1000,250}\\
\hline
0.3\cdot \mathbf{1}_{1000,250} & 0.3\cdot \mathbf{1}_{1000,250} & 0.3\cdot \mathbf{1}_{1000,250} & 0.7\cdot \mathbf{1}_{1000,250}
\end{array}
\right]
$$
as a latent preference matrix, and $(K, d, \ell_{\max}) = (2,2,1)$.

For Fig.~\ref{fig:1a}, we used
$$R=
\left[
\begin{array}{c|c|c|c}
0.3\cdot \mathbf{1}_{1000,250} & 0.3\cdot \mathbf{1}_{1000,250} & 0.3\cdot \mathbf{1}_{1000,250} & 0.3\cdot \mathbf{1}_{1000,250}\\
\hline
0.3\cdot \mathbf{1}_{1000,250} & 0.3\cdot \mathbf{1}_{1000,250} & 0.3\cdot \mathbf{1}_{1000,250} & 0.62\cdot \mathbf{1}_{1000,250}
\end{array}
\right]
$$
as a latent preference matrix, and $(K, d, \ell_{\max}) = (2,2,1)$.

For Fig.~\ref{:fig:2c}, we used
$$R=
\left[
\begin{array}{c|c|c|c}
0.3\cdot \mathbf{1}_{1000,250} & 0.3\cdot \mathbf{1}_{1000,250} & 0.3\cdot \mathbf{1}_{1000,250} & 0.3\cdot \mathbf{1}_{1000,250}\\
\hline
0.3\cdot \mathbf{1}_{1000,250} & 0.3\cdot \mathbf{1}_{1000,250} & 0.3\cdot \mathbf{1}_{1000,250} & 0.55\cdot \mathbf{1}_{1000,250}
\end{array}
\right]
$$
as a latent preference matrix, and $(K, d, \ell_{\max}) = (2,2,1)$.

\subsection{Sec.~\ref{sec:6:3}}

For Fig.~\ref{:fig:2d}, we used
$$R=
\left[
\begin{array}{c|c|c|c}
0.2\cdot \mathbf{1}_{1000,250} & 0.5\cdot \mathbf{1}_{1000,250} & 0.5\cdot \mathbf{1}_{1000,250} & 0.7\cdot \mathbf{1}_{1000,250}\\
\hline
0.2\cdot \mathbf{1}_{1000,250} & 0.5\cdot \mathbf{1}_{1000,250} & 0.5\cdot \mathbf{1}_{1000,250} & 0.5\cdot \mathbf{1}_{1000,250}
\end{array}
\right]
$$
as a latent preference matrix, and $(K, d, \ell_{\max}) = (2,3,1)$.

\subsection{Sec.~\ref{sec:6:4}}\label{sec:c:4}

\textbf{Synthetic rating/real graph experiment}\\
For Fig.~\ref{fig:1b}, we take the Facebook graph data~\cite{51inAhn} as graph side information. In specific, we use the social graph of $1637$ students in Vassar College; an edge is placed between two students if they are friends in Facebook.
Students are clustered by the year they entered the college; 467, 590, 580 students in each year. On top of the social graph, we generate binary ratings as per our discrete-valued latent preference model. We used
$$R=
\left[
\begin{array}{c|c|c|c|c}
0.05\cdot \mathbf{1}_{467,200} & 0.05\cdot \mathbf{1}_{467,200} & 0.5\cdot \mathbf{1}_{467,200} & 0.95\cdot \mathbf{1}_{467,200} & 0.95\cdot \mathbf{1}_{467,200}\\
\hline
0.05\cdot \mathbf{1}_{590,200} & 0.95\cdot \mathbf{1}_{590,200} & 0.05\cdot \mathbf{1}_{590,200} & 0.05\cdot \mathbf{1}_{590,200} & 0.05\cdot \mathbf{1}_{590,200}\\
\hline
0.95\cdot \mathbf{1}_{580,200} & 0.95\cdot \mathbf{1}_{580,200} & 0.95\cdot \mathbf{1}_{580,200} & 0.95\cdot \mathbf{1}_{580,200} & 0.95\cdot \mathbf{1}_{580,200}\\
\end{array}
\right]
$$
as a latent preference matrix, and $(K, d, \ell_{\max}) = (3,3,2)$.

We computed the expected value of MAE of the optimal estimator in Sec.~\ref{sec:6:4} as follows. Suppose there exists an optimal estimator $\psi^{*}$ in the sense that $\psi^{*}(N^{\Omega}, G)_{ij}=R_{ij}$ for all $i, j$. Then 
\begin{align*}
\text{the expected value of test MAE of $\psi^{*}$} &= \E \Big[\frac{1}{|{\Omega}_{t}|} \sum_{(i,j)\in {\Omega}_{t}}|N^{\Omega}_{ij}-(2 \psi^{*}(N^{\Omega}, G)_{ij}-1)|\Big]\\
&=\E \Big[\frac{1}{|{\Omega}_{t}|} \sum_{(i,j)\in {\Omega}_{t}}|N^{\Omega}_{ij}-(2 R_{ij} -1)|\Big]\\
&=\frac{1}{|{\Omega}_{t}|} \sum_{(i,j)\in {\Omega}_{t}}\E \big[|N^{\Omega}_{ij}-(2 R_{ij} -1)|\big]\\
&=\frac{1}{|{\Omega}_{t}|} \sum_{(i,j)\in {\Omega}_{t}}\E \big[|2 \text{Bern}(R_{ij})-1-(2 R_{ij} -1)|\big]\\
&~~~~(\because \text{conditioned on}~(i,j)\in \Omega, N^{\Omega}_{ij} = 2 \text{Bern}(R_{ij})-1)\\
&= \E \big[|2 \text{Bern}(R_{ij})-1-(2 R_{ij} -1)|\big]\\
&= 2 \E \big[| \text{Bern}(R_{ij})- R_{ij} |\big]\\
&= 2 \frac{1}{1637 \cdot 1000}\sum_{\substack{1\le i \le 1637 \\ 1\le j\le 1000}} \{R_{ij}\cdot |1- R_{ij}| + (1- R_{ij})\cdot |0- R_{ij}| \}\\
&\approx 0.236
\end{align*}

\noindent \textbf{Real rating/real graph experiment}\\
We used a real ratings/real graph dataset called Epinions~\cite{Massa_2007, Massa_2008} that consists of $40163$ users and $139738$ items with rating and graph data. We preprocess this dataset as follows. First, the rating scale of this dataset is from +1 to +5, so we regard +1/+2 as dislike(-1), +4/+5 as like(+1) and ignore +3's (i.e., we treat +3's as unobserved ratings). After the first step, the observation rate of ratings is about $0.00012$ which is too small for meaningful analysis. This is why we add the following preprocessing steps. 1) Find $100$ most frequently rated items. 2) Find $1000$ users who rated above $100$ items most frequently. 3) Find a subset of users that shows a cluster structure via spectral clustering. As a result, we get a preprocessed dataset that consists of $290$ users and $100$ items.

\subsection{Sec.~\ref{sec:6:5}}

For the black curve with triangle markers in Fig.~\ref{fig:3b}, we used
$$R=
\left[
\begin{array}{c|c|c|c}
0.2\cdot \mathbf{1}_{5000,1250} & 0.5\cdot \mathbf{1}_{5000,1250} & 0.5\cdot \mathbf{1}_{5000,1250} & 0.7\cdot \mathbf{1}_{5000,1250}\\
\hline
0.2\cdot \mathbf{1}_{5000,1250} & 0.5\cdot \mathbf{1}_{5000,1250} & 0.7\cdot \mathbf{1}_{5000,1250} & 0.5\cdot \mathbf{1}_{5000,1250}
\end{array}
\right]
$$
as a latent preference matrix and $(K, d, \ell_{\max}, \alpha, \beta) = (2, 2, 1, 0.26, 0.23)$.\\
For the gray curve with square markers in Fig.~\ref{fig:3b} and Fig.~\ref{fig:3c}, we used
$$R=
\left[
\begin{array}{c|c|c|c}
0.2\cdot \mathbf{1}_{5000,1250} & 0.5\cdot \mathbf{1}_{5000,1250} & 0.5\cdot \mathbf{1}_{5000,1250} & 0.7\cdot \mathbf{1}_{5000,1250}\\
\hline
0.2\cdot \mathbf{1}_{5000,1250} & 0.5\cdot \mathbf{1}_{5000,1250} & 0.5\cdot \mathbf{1}_{5000,1250} & 0.5\cdot \mathbf{1}_{5000,1250}
\end{array}
\right]
$$
as a latent preference matrix and $(K, d, \ell_{\max},\alpha, \beta) = (2, 2, 1, 0.26, 0.23)$.\\
For the yellow curve with circle markers in Fig.~\ref{fig:3c}, we used
$$R=
\left[
\begin{array}{c|c|c|c}
0.2\cdot \mathbf{1}_{5000,1250} & 0.5\cdot \mathbf{1}_{5000,1250} & 0.5\cdot \mathbf{1}_{5000,1250} & 0.7\cdot \mathbf{1}_{5000,1250}\\
\hline
0.2\cdot \mathbf{1}_{5000,1250} & 0.5\cdot \mathbf{1}_{5000,1250} & 0.5\cdot \mathbf{1}_{5000,1250} & 0.5\cdot \mathbf{1}_{5000,1250}
\end{array}
\right]
$$
as a latent preference matrix and $(K, d, \ell_{\max},\alpha, \beta) = (2, 2, 1, 0.27, 0.23)$.

\section{Proof of Theorems}
In this section, we provide proofs of theorems.

\subsection{Proof of Theorem~\ref{thm:1}}
We first recall some definitions and the main theorem.

\begin{definition}[Worst-case probability of error]
Let $\gamma$ be a fixed number in $(0,1)$ and $\psi$ be an estimator that outputs a latent preference matrix in $\{p_1,p_2, \dots, p_d\}^{n\times m}$ based on $N^{\Omega}$ and $G$. 
We define the worst-case probability of error $P_e^{\gamma}(\psi):= \max\big\{\Pr(\psi(N^\Omega, G)\neq R) : R\in \{p_1, p_2, \dots , p_d\}^{n\times m}, \|u_R-v_R\|_0=\lceil \gamma m \rceil \big\}$ where $\|\cdot\|_0$ is the hamming distance.
\end{definition}

\begin{theorem}
Let $K=2, |C^{-1}(\{1\})|=|C^{-1}(\{2\})|=\frac{n}{2}, \gamma \in (0,1)$, $m=\omega (\log n)$, $\log m=o(n)$, 
$I_s :=-2\log \big( 1- d^2_H(\alpha, \beta) \big).$
Then, the following holds for arbitrary $\epsilon>0$;

(I) if $p\ge \frac{1}{(d^{\min}_H)^2} \max\left\{\frac{(1+\epsilon) \log n - \frac{n}{2}I_s}{\gamma m},\frac{(1+\epsilon)2\log m}{n}\right\}$,
then there exists an estimator $\psi$ that outputs a latent preference matrix in $\{p_1,p_2, \dots, p_d\}^{n\times m}$ based on $N^{\Omega}$ and $G$ such that $P_e^{\gamma}(\psi)\rightarrow 0$ as $n \rightarrow \infty$

(II) if $p\le \frac{1}{(d^{\min}_H)^2} \max\left\{\frac{(1-\epsilon) \log n - \frac{n}{2}I_s}{\gamma m},\frac{(1-\epsilon)2\log m}{n}\right\}$ and $\alpha = O(\frac{\log n}{n})$, then $P_e^{\gamma}(\psi)\nrightarrow 0$ as $n \rightarrow \infty$ for any $\psi$.
\end{theorem}

\begin{definition}
$p^*_{(\gamma)}:=\frac{1}{(d^{\min}_H)^2} \max\left\{\frac{\log n - \frac{n}{2}I_s}{\gamma m},\frac{2\log m}{n}\right\}$ denotes the optimal observation rate.\\ Then $nmp^*_{(\gamma)} =\frac{1}{(d^{\min}_H)^2} \max\left\{\frac{1}{\gamma} (  n\log n - \frac{1}{2} n^2 I_s),2m\log m\right\}$ denotes the optimal sample complexity.
\end{definition}

We will first show that the maximum likelihood estimator $\psi_{ML}$ satisfies (I), and then show that there does not exist an estimator $\psi$ satisfying (II).

\textbf{(I) MLE Achievability}

\noindent\textbf{Overview of the proof}: We show that if the observation rate is above a certain threshold, then the worst-case probability of error approaches to $0$ as $n\rightarrow \infty$ for MLE. 
In specific, we show the following. 
Given observed ratings and graph side information, Lemma~\ref{lem:1} shows the negative log-likelihood of a latent preference matrix can be written in a compact form. 
Then Lemma~\ref{lem:2} represents the probability of the event ``the likelihood of a candidate latent preference matrix is greater than that of the ground-truth latent preference matrix'' in a compact form. 
In Lemma~\ref{lem:3}, we apply Chernoff bounds to the result of Lemma~\ref{lem:2} to get an upper bound of the probability of error. 
Then we finally show that the worst case probability of error approaches to $0$ as $n\rightarrow \infty$ by applying the union bound.
To get a tight bound, we enumerate all possible latent preference matrices and group them into four distinct types based on Definition~\ref{def:A:3}. 
Note that our technical contributions lie in the proofs of Lemma~\ref{lem:1}, Lemma~\ref{lem:2} and Lemma~\ref{lem:3} in which we must consider a significantly larger set of candidate latent preference matrices compared to the symmetric case.
In Remark~\ref{remark:11}, we give a detailed explanation of our definition of $I_s$. The following diagram visualizes the proof dependencies.

\begin{tcolorbox}
\textbf{Proof dependencies}:\\
(I) of Theorem 1 $\longleftarrow$ Equation (1)\\
\phantom ~~~~~~Equation (1) $\longleftarrow$ Lemma 2, 3\\
\phantom ~~~~~~~~~~~~ Lemma 2 $\longleftarrow$ Lemma 1\\
\phantom ~~~~~~~~~~~~ Lemma 3 $\longleftarrow$ Chernoff bounds
\end{tcolorbox}

Let $R$ be an arbitrary ground-truth latent preference matrix satisfying $\|u_R-v_R\|_0=\lceil \gamma m \rceil$ and assume our model is generated as per $R$ (i.e., user $i$ likes item $j$ with probability $R_{ij}$ and $n$ users are clustered into $A_R$ and $B_R$). By switching the order of items (columns of the latent preference matrix) if necessary, we can assume $(u_R)_j=(v_R)_j$ for $j=1,2,\dots, (m-\lceil \gamma m\rceil)$ and $(u_R)_j \neq (v_R)_j$ for $j=(m-\lceil \gamma m\rceil +1),\dots, m$. By switching the order of users (rows of the latent preference matrix) if necessary, we can also assume that $A_R=[\frac{n}{2}], B_R=[n]\setminus [\frac{n}{2}]$. We will first find the upper bound of $\Pr(\psi_{ML}(N^\Omega, G)\neq R)$ for arbitrary $R$, and show that the upper bound approaches to $0$ as $n$ approaches to infinity. We need following lemmas.

\begin{lemma}\label{lem:1}
Let $L(X)$ be the negative log-likelihood of a latent preference matrix $X$ for given $N^{\Omega}$ and $G$. Then $$L(X)=\underset{1\le t \le d}{\Sigma}\Big[\underset{(i,j)\in \Omega_t}{\Sigma}\left\{\mathbb{1}(N^{\Omega}_{ij}=1)(-\log p_t) + \mathbb{1}(N^{\Omega}_{ij}=-1) (-\log (1-p_t))\right\}\Big]+\log \Big(\frac{\alpha(1-\beta)}{(1-\alpha)\beta}\Big) e(A_X,B_X)+c,$$ where $$c:=-\log \Big\{ (1-p)^{nm-|\Omega|}p^{|\Omega|} \alpha^{|E|} (1-\alpha)^{2\binom{n/2}{2}} (1-\beta)^{(\frac{n}{2})^2} \Big\},$$ and $\Omega_t :=\{(i,j)\in \Omega : X_{ij}=p_t\}$.
\end{lemma}

\begin{proof}
The likelihood of latent preference matrix X given $N^{\Omega}$ and G is $\Pr(N^{\Omega}, G|X)=\Pr(N^{\Omega}|X)\Pr(G|X)$. It is clear that 
$$\Pr(N^{\Omega}|X)=(1-p)^{nm-|\Omega|}p^{|\Omega|}\underset{1\le t \le d}{\Pi}\Big[\underset{(i,j)\in \Omega_t}{\Pi}\big\{p_t^{\mathbb{1}(N^{\Omega}_{ij}=1)}(1-p_t)^{\mathbb{1}(N^{\Omega}_{ij}=-1)}\big\}\Big]$$
where $\Omega_t :=\{(i,j)\in \Omega : X_{ij}=p_t\}$, and
$$\Pr(G|X)=\alpha^{|E|-e(A_X,B_X)}(1-\alpha)^{2\binom{n/2}{2} -(|E|-e(A_X,B_X))} \beta^{e(A_X,B_X)}(1-\beta)^{(\frac{n}{2})^2-e(A_X,B_X)}$$
Then the negative log-likelihood of X can be computed as follows.
\begin{align*}
L(X)&=-\log (\Pr(N^{\Omega}, G|X))\\
&=\underset{1\le t \le d}{\Sigma}\Big[\underset{(i,j)\in \Omega_t}{\Sigma}\big\{\mathbb{1}(N^{\Omega}_{ij}=1)(-\log p_t) + \mathbb{1}(N^{\Omega}_{ij}=-1) (-\log (1-p_t))\big\}\Big]+\log \big(\frac{\alpha(1-\beta)}{(1-\alpha)\beta}\big) e(A_X,B_X)+c
\end{align*}
where $c=-\log \big\{ (1-p)^{nm-|\Omega|}p^{|\Omega|} \alpha^{|E|} (1-\alpha)^{2\binom{n/2}{2}} (1-\beta)^{(\frac{n}{2})^2} \big\}$.
\end{proof}

\begin{definition}\label{def:A:3}
(i) $\chi(k,a_1, a_2, b_1, b_2):= \big\{X\in \{p_1, p_2, \dots , p_d\}^{n\times m} : |A_X\setminus A_R|=|B_X\setminus B_R|=k,$ $u_X$ differs from $u_R$ at $a_1$ coordinates among the first $m-\lceil \gamma m\rceil$ coordinates, $u_X$ differs from $u_R$ at $a_2$ coordinates among the last $\lceil \gamma m \rceil$ coordinates, $v_X$ differs from $v_R$ at $b_1$ many coordinates among the first $m-\lceil \gamma m\rceil$ coordinates, and $v_X$ differs from $v_R$ at $b_2$ many coordinates among the last $\lceil \gamma m \rceil$ coordinates.$ \big\}$. (ii) Let $\mathcal{I}$ to be the index set of $\chi$, namely, $\mathcal{I}:=\big\{(k,a_1,a_2,b_1,b_2)\neq (0,0,0,0,0) : 0\le k\le \frac{n}{4}, 0\le a_1,b_1\le m-\lceil \gamma m \rceil, 0\le a_2,b_2\le \lceil \gamma m \rceil \big\}$.
\end{definition}

Note that we can assume $k\le \frac{n}{4}$ by switching the role of $A_X$ and $B_X$ if necessary.

\begin{lemma}\label{lem:2}
For $X\in \chi(k,a_1, a_2, b_1, b_2)$, 
\begin{align*}
\Pr(L(X)\le L(R))&=\Pr\bigg(\underset{1\le a\neq b \le d}{\Sigma}\Big[\underset{(i,j)\in Q_{ab}}{\Sigma}\mathbb{P}_{ij}\big\{\mathbb{P}_{a,ij}\log \frac{p_b}{p_a} + (1-\mathbb{P}_{a,ij}) \log \frac{1-p_b}{1-p_a} \big\}\Big]\\
&~~~~~~~~~~~~~~~~+\log \Big(\frac{\alpha(1-\beta)}{(1-\alpha)\beta}\Big) \underset{1\le i \le 2(\frac{n}{2}-k)k}\Sigma(B_i-A_i)\ge 0\bigg),    
\end{align*}
where $$Q_{ab}:=\big\{(i,j)\in [n]\times [m] : R_{ij}=p_a, X_{ij}=p_b\big\},$$ and $A_i\overset{i.i.d.}{\sim} \text{Bern}(\alpha)$, $B_i\overset{i.i.d.}{\sim} \text{Bern}(\beta)$, $\mathbb{P}_{ij}\overset{i.i.d.}{\sim} \text{Bern}(p)$, $\mathbb{P}_{a, ij}\overset{i.i.d.}{\sim} \text{Bern}(p_a)$.
\end{lemma}

\begin{proof}
By Lemma~\ref{lem:1}, $$L(R)=\underset{1\le a \le d}{\Sigma}\Big[\underset{(i,j)\in \Omega_a^{(R)}}{\Sigma}\big\{\mathbb{1}(N^{\Omega}_{ij}=1)(-\log p_a) + \mathbb{1}(N^{\Omega}_{ij}=-1) (-\log (1-p_a))\big\}\Big]+\log \Big(\frac{\alpha(1-\beta)}{(1-\alpha)\beta}\Big) e(A_R,B_R)+c$$ and $$L(X)=\underset{1\le b \le d}{\Sigma}\Big[\underset{(i,j)\in \Omega_b^{(X)}}{\Sigma}\big\{\mathbb{1}(N^{\Omega}_{ij}=1)(-\log p_b) + \mathbb{1}(N^{\Omega}_{ij}=-1) (-\log (1-p_b))\big\}\Big]+\log \Big(\frac{\alpha(1-\beta)}{(1-\alpha)\beta}\Big) e(A_X,B_X)+c$$
where $\Omega_a^{(R)} :=\big\{(i,j)\in \Omega : R_{ij}=p_a \big\}$ and $\Omega_b^{(X)} :=\big\{(i,j)\in \Omega : X_{ij}=p_b \big\}$. Then
\begin{align*}
L(R)-L(X)&= \underset{1\le a, b \le d}{\Sigma}\Big[\underset{(i,j)\in \Omega_{ab}}{\Sigma}\big\{\mathbb{1}(N^{\Omega}_{ij}=1)(-\log p_a +\log p_b) + \mathbb{1}(N^{\Omega}_{ij}=-1) (-\log (1-p_a)+\log (1-p_b))\big\}\Big]\\
&~~~~~~+\log \big(\frac{\alpha(1-\beta)}{(1-\alpha)\beta}\big) \{e(A_R,B_R)-e(A_X,B_X)\}
\end{align*}
where $\Omega_{ab} :=\{(i,j)\in \Omega : R_{ij}=p_a, X_{ij}=p_b\}$. 
It is clear that 
\begin{align*}
&\underset{1\le a, b \le d}{\Sigma}\Big[\underset{(i,j)\in \Omega_{ab}}{\Sigma}\big\{\mathbb{1}(N^{\Omega}_{ij}=1)(-\log p_a +\log p_b) + \mathbb{1}(N^{\Omega}_{ij}=-1) (-\log (1-p_a)+\log (1-p_b))\big\}\Big]\\
&=\underset{1\le a\neq b \le d}{\Sigma}\Big[\underset{(i,j)\in \Omega_{ab}}{\Sigma}\big\{\mathbb{1}(N^{\Omega}_{ij}=1)(-\log p_a +\log p_b) + \mathbb{1}(N^{\Omega}_{ij}=-1) (-\log (1-p_a)+\log (1-p_b))\big\}\Big]\\
&=\underset{1\le a\neq b \le d}{\Sigma}\Big[\underset{(i,j)\in Q_{ab}}{\Sigma}\mathbb{P}_{ij}\big\{\mathbb{P}_{a,ij}\log \frac{p_b}{p_a} + (1-\mathbb{P}_{a,ij}) \log \frac{1-p_b}{1-p_a} \big\}\Big],
\end{align*} where $Q_{ab}:=\{(i,j)\in [n]\times [m] : R_{ij}=p_a, X_{ij}=p_b\}$, $\mathbb{P}_{ij}\overset{i.i.d.}{\sim} \text{Bern}(p)$, $\mathbb{P}_{a, ij}\overset{i.i.d.}{\sim} \text{Bern}(p_a)$. Note that our ground-truth latent preference matrix is $R$, hence $\mathbb{1}(N_{ij}^{\Omega}=1)=\mathbb{P}_{ij} \mathbb{P}_{a, ij}$ and $\mathbb{1}(N_{ij}^{\Omega}=-1)=\mathbb{P}_{ij}(1-\mathbb{P}_{a, ij})$ for $(i,j)\in Q_{ab}$).

Note that $X\in \chi(k,a_1, a_2, b_1, b_2)$, $|A_X\setminus A_R|=|B_X\setminus B_R|=k$, $|A_X\cap A_R|=|B_X\cap B_R|=\frac{n}{2}-k,$
$$e(A_R,B_R)=e(B_X\setminus B_R,B_X\cap B_R)+e(B_X\setminus B_R, A_X\setminus A_R)+e(A_X\cap A_R,B_X\cap B_R)+e(A_X\cap A_R,A_X\setminus A_R),$$ and
$$e(A_X,B_X)=e(B_X\setminus B_R,A_X\cap A_R)+e(B_X\setminus B_R, A_X\setminus A_R)+e(B_X\cap B_R,A_X\cap A_R)+e(B_X\cap B_R,A_X\setminus A_R),$$ so we can get
\begin{align*}
e(A_R,B_R)-e(A_X,B_X)&=e(B_X\setminus B_R,B_X\cap B_R)+e(A_X\cap A_R,A_X\setminus A_R)-\{e(B_X\setminus B_R,A_X\cap A_R)\\
&~~~~+e(B_X\cap B_R,A_X\setminus A_R)\}=\underset{1\le i \le 2(\frac{n}{2}-k)k}\Sigma B_i-\underset{1\le i \le 2(\frac{n}{2}-k)k}\Sigma A_i,
\end{align*}
where $A_i\overset{i.i.d.}{\sim} \text{Bern}(\alpha)$, $B_i\overset{i.i.d.}{\sim} \text{Bern}(\beta)$.
\end{proof}

\begin{lemma}\label{lem:3}
Let $X\in \chi(k,a_1, a_2, b_1, b_2), Q_{ab}:=\{(i,j)\in [n]\times [m] : R_{ij}=p_a, X_{ij}=p_b\}$, $A_i\overset{i.i.d.}{\sim} \text{Bern}(\alpha)$, $B_i\overset{i.i.d.}{\sim} \text{Bern}(\beta)$, $\mathbb{P}_{ij}\overset{i.i.d.}{\sim} \text{Bern}(p)$, $\mathbb{P}_{a, ij}\overset{i.i.d.}{\sim} \text{Bern}(p_a)$. Then $\Pr\Big(\underset{1\le a\neq b \le d}{\Sigma}\Big[\underset{(i,j)\in Q_{ab}}{\Sigma}\mathbb{P}_{ij}\big\{\mathbb{P}_{a,ij}\log \frac{p_b}{p_a} + (1-\mathbb{P}_{a,ij}) \log \frac{1-p_b}{1-p_a} \big\}\Big]+\log \big(\frac{\alpha(1-\beta)}{(1-\alpha)\beta}\big) \underset{1\le i \le 2(\frac{n}{2}-k)k}\Sigma(B_i-A_i)\ge 0\Big)\le \{1-p(d^{\min}_H)^2\}^{D_X}\{1- d^2_H(\alpha, \beta)\}^{4(\frac{n}{2}-k)k}$ where $D_X:=|\{(i,j)\in [n]\times [m] : R_{ij}\neq X_{ij}\}|$.
\end{lemma}

\begin{proof}
Let \\
$$Z:=\underset{1\le a\neq b \le d}{\Sigma}\Big[\underset{(i,j)\in Q_{ab}}{\Sigma}\mathbb{P}_{ij}\big\{\mathbb{P}_{a,ij}\log \frac{p_b}{p_a} + (1-\mathbb{P}_{a,ij}) \log \frac{1-p_b}{1-p_a} \big\}\Big]+\log \big(\frac{\alpha(1-\beta)}{(1-\alpha)\beta}\big) \underset{1\le i \le 2(\frac{n}{2}-k)k}\Sigma(B_i-A_i)$$\\
Then 
\begin{eqnarray}
&\quad&\Pr(Z\ge 0) = 
\Pr(e^{\frac{1}{2}Z}\ge 1)\nonumber \\
&\le& \E[e^{\frac{1}{2}Z}] \qquad (\because \text{Markov's inequality}) \nonumber \\
&=& \underset{1\le a\neq b \le d}{\Pi}\underset{(i,j)\in Q_{ab}}{\Pi}\big\{(1-p)+p(\sqrt{p_a p_b}+\sqrt{(1-p_a)(1-p_b)})\big\}\underset{1\le i \le 2(\frac{n}{2}-k)k}\Pi(\sqrt{\alpha\beta}+\sqrt{(1-\alpha)(1-\beta)})^2 \nonumber \\
&=& \underset{1\le a\neq b \le d}{\Pi}\underset{(i,j)\in Q_{ab}}{\Pi}\big\{1-p d^2_H(p_a, p_b)  \big\}\underset{1\le i \le 2(\frac{n}{2}-k)k}\Pi(1- d^2_H(\alpha, \beta) )^2 \nonumber \\
&\le& \underset{1\le a\neq b \le d}{\Pi}\underset{(i,j)\in Q_{ab}}{\Pi}\{1-p (d^{\min}_H)^2 \}\underset{1\le i \le 2(\frac{n}{2}-k)k}\Pi(1- d^2_H(\alpha, \beta))^2 \nonumber \\
&=& \{1-p (d^{\min}_H)^2\}^{|\{(i,j)\in [n]\times [m] : R_{ij}\neq X_{ij}\}|}\{1- d^2_H(\alpha, \beta)\}^{4(\frac{n}{2}-k)k} \nonumber\\
&=& \{1-p (d^{\min}_H)^2\}^{D_X}\{1- d^2_H(\alpha, \beta)\}^{4(\frac{n}{2}-k)k}. \nonumber
\end{eqnarray}
\end{proof}

\begin{remark} \label{remark:11}
We give a detailed explanation regarding the definition of $I_s$ in our paper. In the proof of Lemma 3 in~\cite{Ahn_2018}, they made implicit assumptions that $\alpha,\beta\rightarrow 0$ and $\frac{\alpha}{\beta}\rightarrow 1$ as $n\rightarrow 0$, and used these assumptions when they approximate $-2\log \big( 1- d^2_H(\alpha, \beta) \big)=(1+o(1))(\sqrt{\alpha}-\sqrt{\beta})^2$. The approximation does not hold without above assumptions. In general,
\begingroup
\allowdisplaybreaks
\begin{align*}
1- d^2_H(\alpha, \beta) &= \sqrt{\alpha \beta}+\sqrt{(1-\alpha)(1-\beta)}\\
&=\sqrt{(\beta +y) \beta}+\sqrt{(1-\beta-y)(1-\beta)} \quad (y:=\alpha-\beta) \\
&=\beta\sqrt{(1 +\frac{y}{\beta}) }+(1-\beta)\sqrt{(1-\frac{y}{1-\beta})} \\
&=\beta\Big\{1+\frac{1}{2}\frac{y}{\beta}-\frac{1}{8}(\frac{y}{\beta})^2+O(y^3)\Big\}+(1-\beta)\Big\{1-\frac{1}{2}\frac{y}{1-\beta}-\frac{1}{8}(\frac{y}{1-\beta})^2+O(y^3)\Big\}\\
&~~~~(\because \sqrt{1+x}=1+\frac{x}{2}-\frac{x^2}{8}+O(x^3) )\\
&=1-\frac{y^2}{8\beta(1-\beta)}+O(y^3)\\
&=1-\frac{(\alpha-\beta)^2}{8\beta(1-\beta)}+O((\alpha-\beta)^3)\\
&=1-\frac{(\sqrt{\alpha}-\sqrt{\beta})^2(\sqrt{\alpha}+\sqrt{\beta})^2}{8\beta(1-\beta)}+O((\sqrt{\alpha}-\sqrt{\beta})^3)    
\end{align*}
\endgroup
Assuming $\sqrt{\alpha}-\sqrt{\beta}=o(1)$ (which is true when $I_s=o(1)$),
\begin{align*}
-2\log \big( 1- d^2_H(\alpha, \beta) \big) &=-2\Big\{-\frac{(\sqrt{\alpha}-\sqrt{\beta})^2(\sqrt{\alpha}+\sqrt{\beta})^2}{8\beta(1-\beta)}+O((\sqrt{\alpha}-\sqrt{\beta})^3)\Big\}\\
&~~~~(\because \log (1+x)=x+O(x^2))\\
&=(\sqrt{\alpha}-\sqrt{\beta})^2 \Big\{\frac{(\sqrt{\alpha}+\sqrt{\beta})^2}{4\beta(1-\beta)}+o(1)\Big\}
\end{align*}

Hence we get $-2\log \big( 1- d^2_H(\alpha, \beta) \big) =(\sqrt{\alpha}-\sqrt{\beta})^2 \big\{\frac{(\sqrt{\alpha}+\sqrt{\beta})^2}{4\beta(1-\beta)}+o(1)\big\}$. Note that  If $\alpha,\beta\rightarrow 0$ and $\frac{\alpha}{\beta}\rightarrow 4$ as $n\rightarrow \infty$, then $-2\log \big( 1- d^2_H(\alpha, \beta) \big)=(\sqrt{\alpha}-\sqrt{\beta})^2 \{\frac{9}{4}+o(1)\}$ which is different from $(\sqrt{\alpha}-\sqrt{\beta})^2 \{1+o(1)\}$ which means the approximation used in~\cite{Ahn_2018} depends on the asymptotic behavior of $\alpha, \beta$. This is why we introduce a modified definition of $I_s:=-2\log \big( 1- d^2_H(\alpha, \beta) \big)$, then our achievability result holds for any $\alpha$ and $\beta$.
\end{remark}

The event ``$\psi_{ML}(N^\Omega, G)\neq R$" occurs only if there exists a latent preference matrix $X$ whose likelihood is greater than $R$'s (in other words, $L(X)\le L(R)$ since $L(\cdot)$ is the negative log-likelihood). Let $[L(X)\le L(R)]$ denotes the event ``$L(X)\le L(R)$". Then 

\begin{align*}
\Pr(\psi_{ML}(N^\Omega, G)\neq R)&\le \Pr\big(\underset{X\neq R}{\cup}[L(X)\le L(R)]\big)\\
&\overset{\text{union bound}}{\le} \underset{X\neq R}{\Sigma}\Pr(L(X)\le L(R))\\
&\overset{\text{Lemma~\ref{lem:2}}}{=} \underset{z \in \mathcal{I}}{\Sigma}\underset{X\in \chi (z)}{\Sigma} \Pr\Big(\underset{1\le a\neq b \le d}{\Sigma}\Big[\underset{(i,j)\in Q_{ab}}{\Sigma}\mathbb{P}_{ij}\big\{\mathbb{P}_{a,ij}\log \frac{p_b}{p_a} + (1-\mathbb{P}_{a,ij}) \log \frac{1-p_b}{1-p_a} \big\}\Big]\\
&~~~~~~~~~~~~~~~~~~~~~~~~~~~~~~~~~~~~+\log \Big(\frac{\alpha(1-\beta)}{(1-\alpha)\beta}\Big) \underset{1\le i \le 2(\frac{n}{2}-k_z)k_z}\Sigma(B_i-A_i)\ge 0\Big)\\ &\overset{\text{Lemma~\ref{lem:3}}}{\le} \underset{z \in \mathcal{I}}{\Sigma}\underset{X\in \chi (z)}{\Sigma} \{1-p (d^{\min}_H)^2 \}^{D_X}\{1- d^2_H(\alpha, \beta) \}^{4(\frac{n}{2}-k_z)k_z}
\end{align*}

where $D_X:=|\{(i,j)\in [n]\times [m] : R_{ij}\neq X_{ij}\}|$ and $k_z$ is the first coordinate of $z=(k_z,a_1,a_2,b_1,b_2)$. A direct calculation yields $D_X\ge k_z\{a_1+(\lceil \gamma m\rceil-a_2)+b_1+(\lceil \gamma m \rceil-b_2)\}+(\frac{n}{2}-k_z)(a_1+a_2+b_1+b_2) =: D_z$, so we have an upper bound of $\Pr(\psi_{ML}(N^\Omega, G)\neq R)$ as follows:
\begin{eqnarray}\label{eq:1}
\Pr(\psi_{ML}(N^\Omega, G)\neq R) &\le& \underset{z \in \mathcal{I}}{\Sigma}\underset{X\in \chi (z)}{\Sigma} \{1-p (d^{\min}_H)^2\}^{D_X}\{1- d^2_H(\alpha, \beta) \}^{4(\frac{n}{2}-k_z)k_z} \nonumber\\
&\le& \underset{z \in \mathcal{I}}{\Sigma}\underset{X\in \chi (z)}{\Sigma} \{1-p (d^{\min}_H)^2\}^{D_z}\{1- d^2_H(\alpha, \beta)\}^{4(\frac{n}{2}-k_z)k_z} \nonumber\\
&=&\underset{z \in \mathcal{I}}{\Sigma}|\chi(z)| \{1-p (d^{\min}_H)^2\}^{D_z}\{1- d^2_H(\alpha, \beta)\}^{4(\frac{n}{2}-k_z)k_z}
\end{eqnarray}

We now show the upper bound of $\Pr(\psi_{ML}(N^\Omega, G)\neq R)$ approaches to $0$ as $n\rightarrow \infty$. Note that\\
$p\ge \frac{1}{(d^{\min}_H)^2} \max\Big\{\frac{(1+\epsilon) \log n - \frac{n}{2}I_s}{\gamma m},\frac{(1+\epsilon)2\log m}{n}\Big\}\Leftrightarrow$  $\frac{1}{2}nI_s+\gamma m p (d^{\min}_H)^2\ge (1+\epsilon)\log n$ and $\frac{1}{2}n p (d^{\min}_H)^2\ge (1+\epsilon)\log m$. Since the RHS of (\ref{eq:1}) increases as $p$ decreases, it suffices to consider the case when $p=O(\frac{\log n}{m}+\frac{\log m}{n})=o(1)$, which implies
$$\log (1-p (d^{\min}_H)^2)=-p (d^{\min}_H)^2+O(p^2)=-p (d^{\min}_H)^2(1+o(1))$$
Hence the RHS of (\ref{eq:1}) can be represented as 
\begin{equation}
\underset{z \in \mathcal{I}}{\Sigma}|\chi(z)| e^{-D_z p (d^{\min}_H)^2(1+o(1))-2(\frac{n}{2}-k_z)k_z I_s (1+o(1))}= \underset{z \in \mathcal{I}}{\Sigma}|\chi(z)| e^{(1+o(1))(-D_z I_r-2(\frac{n}{2}-k_z)k_z I_s) }\label{eq:2}
\end{equation}
where $I_r:=p (d^{\min}_H)^2$. For a constant $\delta \in (0, \text{min}\{\gamma, 1-\gamma\})$ (the exact value of $\delta$ will be determined later), define $$\mathcal{J}:=\{(k,a_1,a_2,b_1,b_2)\in \mathcal{I} : a_1,a_2,b_1,b_2<\delta m\},~ \mathcal{K}:=\{(k,a_1,a_2,b_1,b_2)\in \mathcal{I} : k<\delta n\}.$$
Now we show the RHS of (\ref{eq:2}) approaches to $0$ as $n\rightarrow \infty$ by dividing it into four partial sums over $\mathcal{I}\setminus (\mathcal{J}\cup \mathcal{K}), \mathcal{J}\setminus \mathcal{K}, \mathcal{K}\setminus \mathcal{J}, \mathcal{J}\cap \mathcal{K}$.\\

\begin{itemize}
    \item{Case 1. $\mathcal{I}\setminus (\mathcal{J}\cup \mathcal{K})$:} For $z=(k_z,a_1,a_2,b_1,b_2) \in \mathcal{I}\setminus (\mathcal{J}\cup \mathcal{K})$, $\delta n \le k_z \le \frac{n}{4}$ since $z\notin \mathcal{K}$. So $2(\frac{n}{2}-k_z)k_z \ge 2(\frac{n}{2}-\frac{n}{4})\delta n=\frac{\delta}{2}n^2$. As $z\notin \mathcal{J}$, we can assume $a_1\ge \delta m$ without loss of generality, which implies $D_z=k_z\{a_1+(\lceil \gamma m\rceil-a_2)+b_1+(\lceil \gamma m \rceil-b_2)\}+(\frac{n}{2}-k_z)(a_1+a_2+b_1+b_2)\ge (\frac{n}{2}-k_z)a_1 \ge (\frac{n}{2}-\frac{n}{4})\delta m=\frac{\delta}{4}nm$. Then 
\begingroup
\allowdisplaybreaks
\begin{align*}
&\quad\underset{z \in \mathcal{I}\setminus (\mathcal{J}\cup \mathcal{K})}{\Sigma}|\chi(z)| e^{(1+o(1))(-D_z I_r -2(\frac{n}{2}-k_z)k_z I_s)}\\
&\le \underset{z \in \mathcal{I}\setminus (\mathcal{J}\cup \mathcal{K})}{\Sigma}|\chi(z)| e^{(1+o(1))(-\frac{\delta}{4}nm I_r-\frac{\delta}{2}n^2 I_s)} \\
&= e^{(1+o(1))(-\frac{\delta}{4}nm I_r-\frac{\delta}{2}n^2 I_s)} \underset{z \in \mathcal{I}\setminus (\mathcal{J}\cup \mathcal{K})}{\Sigma}|\chi(z)| \\
&\le e^{(1+o(1))(-\frac{\delta}{4}nm I_r-\frac{\delta}{2}n^2 I_s)} \underset{z \in \mathcal{I}}{\Sigma}|\chi(z)| \\
&\le e^{(1+o(1))(-\frac{\delta}{4}nm I_r-\frac{\delta}{2}n^2 I_s)} 2^n d^{2m} \\ 
&~~~~(\because \text{the total number of latent preference matrices is bounded by } 2^n d^{2m}) \\ 
&\le e^{(1+o(1))(-n(\frac{\delta}{2}n I_s+\frac{\delta}{8}m I_r) -m(\frac{\delta}{8}n I_r))} e^{n\log 2+2m\log d} \\ 
&= e^{(1+o(1))(-n(\Omega(\log n)) -m(\Omega(\log m))+(n\log 2+2m\log d)} \rightarrow 0 \text{ as } n\rightarrow \infty. \\
&~~~~(\because m=w(\log n),\quad m \rightarrow \infty \text{ as } n \rightarrow \infty ) 
\end{align*}
\endgroup

\item{Case 2. $\mathcal{J}\setminus \mathcal{K}$:} For $z=(k_z,a_1,a_2,b_1,b_2) \in \mathcal{J}\setminus \mathcal{K}$, $\delta n \le k_z \le \frac{n}{4}$ since $z\notin \mathcal{K}$. So $2(\frac{n}{2}-k_z)k_z \ge 2(\frac{n}{2}-\frac{n}{4})\delta n=\frac{\delta}{2}n^2$. As $z\in \mathcal{J}$, $D_z=k_z\{a_1+(\lceil \gamma m\rceil-a_2)+b_1+(\lceil \gamma m \rceil-b_2)\}+(\frac{n}{2}-k_z)(a_1+a_2+b_1+b_2)\ge k_z\{a_1+(\lceil \gamma m\rceil-a_2)+b_1+(\lceil \gamma m \rceil-b_2)\} \ge k_z\{(\lceil \gamma m\rceil-\delta m)+(\lceil \gamma m \rceil-\delta m)\}\ge 2\delta (\gamma -\delta)mn$. By applying the argument of Case 1, we have 
$$\underset{z \in \mathcal{J}\setminus \mathcal{K}}{\Sigma}|\chi(z)| e^{(1+o(1))(-D_z I_r -2(\frac{n}{2}-k_z)k_z I_s)} \rightarrow 0 \text{ as } n\rightarrow \infty.$$

\item{Case 3. $\mathcal{K}\setminus \mathcal{J}$:} As $z\in \mathcal{K}$, $k_z< \delta n$, which implies $2(\frac{n}{2}-k_z)k_z \ge 2(\frac{n}{2}-\delta n)k_z=k_z(1-2\delta)n$. As $z\notin \mathcal{J}$, assume $a_1\ge \delta m$ without loss of generality. Then $D_z=k_z\{a_1+(\lceil \gamma m\rceil-a_2)+b_1+(\lceil \gamma m \rceil-b_2)\}+(\frac{n}{2}-k_z)(a_1+a_2+b_1+b_2)\ge (\frac{n}{2}-\delta n)(\delta m)= (\frac{1}{2}-\delta)\delta nm$. This case is a simple version of Case 4, and one can show that $$\underset{z \in \mathcal{K}\setminus \mathcal{J}}{\Sigma}|\chi(z)| e^{(1+o(1))(-D_z I_r -2(\frac{n}{2}-k_z)k_z I_s)} \rightarrow 0 \text{ as } n\rightarrow \infty.$$

\item{Case 4. $\mathcal{J}\cap \mathcal{K}$:} As $z\in \mathcal{K}$, $k_z< \delta n$, which implies $2(\frac{n}{2}-k_z)k_z \ge 2(\frac{n}{2}-\delta n)k_z=k_z(1-2\delta)n$. As $z\in \mathcal{J}$, $D_z=k_z\big\{a_1+(\lceil \gamma m\rceil-a_2)+b_1+(\lceil \gamma m \rceil-b_2)\big\}+(\frac{n}{2}-k_z)(a_1+a_2+b_1+b_2)\ge k_z\big\{(\lceil \gamma m\rceil-\delta m)+(\lceil \gamma m \rceil-\delta m)\big\}+(\frac{n}{2}-\delta n)(a_1+a_2+b_1+b_2)\ge 2k_z(\gamma -\delta )m+(\frac{n}{2}-\delta n)(a_1+a_2+b_1+b_2)$. Then 
\begingroup
\allowdisplaybreaks
\begin{align*}
&\quad\underset{z \in \mathcal{J}\cap \mathcal{K}}{\Sigma}|\chi(z)| e^{(1+o(1))(-D_z I_r -2(\frac{n}{2}-k_z)k_z I_s)}\\
&\le \underset{z \in \mathcal{J}\cap \mathcal{K}}{\Sigma}|\chi(z)| e^{(1+o(1))\big[-\big\{2k_z(\gamma -\delta )m+(\frac{n}{2}-\delta n)(a_1+a_2+b_1+b_2)\big\} I_r -k_z(1-2\delta)n I_s\big]} \\
&= \underset{z \in \mathcal{J}\cap \mathcal{K}}{\Sigma}|\chi(z)| e^{(1+o(1))\big[-2k_z\big\{(\gamma -\delta )m I_r +(\frac{1}{2}-\delta)n I_s\big\}-(\frac{1}{2}-\delta )n(a_1+a_2+b_1+b_2) I_r\big]} \\
&\overset{(i)}{\le} \underset{z \in \mathcal{J}\cap \mathcal{K}}{\Sigma}|\chi(z)| e^{(1+o(1))\big[-2k_z\big\{(1+\frac{\epsilon}{2})\log n\big\}-(1+\frac{\epsilon}{2})\log m(a_1+a_2+b_1+b_2)\big]} \\
&= \underset{z \in \mathcal{J}\cap \mathcal{K}}{\Sigma}|\chi(z)| n^{-2k_z(1+\frac{\epsilon}{2})(1+o(1))}m^{-(1+\frac{\epsilon}{2})(a_1+a_2+b_1+b_2)(1+o(1))} \\
&\le \underset{z \in \mathcal{J}\cap \mathcal{K}}{\Sigma}|\chi(z)| n^{-2k_z(1+\frac{\epsilon}{4})}m^{-(1+\frac{\epsilon}{4})(a_1+a_2+b_1+b_2)} \\
&~~~~(\because (1+\frac{\epsilon}{2})(1+o(1))\ge (1+\frac{\epsilon}{4}) \text{ for sufficiently large } n)\\
&\overset{(ii)}{\le} \underset{z \in \mathcal{J}\cap \mathcal{K}}{\Sigma}n^{2k_z}m^{a_1+a_2+b_1+b_2}(d-1)^{a_1+a_2+b_1+b_2} n^{-2k_z(1+\frac{\epsilon}{4})}m^{-(1+\frac{\epsilon}{4})(a_1+a_2+b_1+b_2)} \\
&= \underset{z \in \mathcal{J}\cap \mathcal{K}}{\Sigma}n^{-\frac{\epsilon}{2}k_z} m^{-\frac{\epsilon}{4}(a_1+a_2+b_1+b_2)}(d-1)^{a_1+a_2+b_1+b_2} \\
&= \underset{z \in \mathcal{J}\cap \mathcal{K}}{\Sigma}n^{-\frac{\epsilon}{2}k_z} \{m^{-\frac{\epsilon}{4}}(d-1)\}^{(a_1+a_2+b_1+b_2)} \\
&\overset{(iii)}{\le} \underset{k_z\in [0,\delta n]}{\Sigma}n^{-\frac{\epsilon}{2}k_z}\big(\underset{a_1\in [0,\delta m]}{\Sigma}\{m^{-\frac{\epsilon}{4}}(d-1)\}^{a_1}\big)^4-1 \\
&\le \underset{k_z\in \mathbb{N}\cup \{0\}}{\Sigma}n^{-\frac{\epsilon}{2}k_z}\big(\underset{a_1\in \mathbb{N}\cup \{0\}}{\Sigma}\{m^{-\frac{\epsilon}{4}}(d-1)\}^{a_1}\big)^4-1 \\
&= \Big(\frac{1}{1-n^{-\frac{\epsilon}{2}}}\Big)\Big(\frac{1}{1-m^{-\frac{\epsilon}{4}}(d-1)}\Big)-1 \rightarrow 0 \text{ as } n\rightarrow \infty. 
\end{align*}
\endgroup
Here, three inequalities hold for the following reasons.

\textit{(i)} : It follows from the assumptions $\frac{1}{2}n I_s+\gamma m I_r\ge (1+\epsilon)\log n, \frac{1}{2}n I_r\ge (1+\epsilon)\log m$, and they imply $(\gamma -\delta )m I_r +(\frac{1}{2}-\delta)n I_s \ge (1+\frac{\epsilon}{2})\log n, (\frac{1}{2}-\delta )n I_r \ge (1+\frac{\epsilon}{2})\log m $ for sufficiently small $\delta$. In explicit, $\delta=\text{min}\big\{\frac{1}{2}\frac{\frac{\epsilon}{2}}{1+\epsilon}, \gamma\frac{\frac{\epsilon}{2}}{1+\epsilon}\big\}$ satisfies above inequalities.

\textit{(ii)} : For $z=(k_z,a_1,a_2,b_1,b_2)$, a direct calculation yields $|\chi(z)|=\binom{\frac{n}{2}}{k_z}^2\binom{m-\lceil \gamma m\rceil}{a_1} (d-1)^{a_1}\binom{m-\lceil \gamma m\rceil}{b_1}(d-1)^{b_1}\binom{\lceil \gamma m\rceil}{a_2}(d-1)^{a_2}\binom{\lceil \gamma m\rceil}{b_2}(d-1)^{b_2}$ and it can be upper bounded by $n^{2k_z}m^{a_1+a_2+b_1+b_2}(d-1)^{a_1+a_2+b_1+b_2}$

\textit{(iii)} : Note that $\underset{z \in \mathcal{J}\cap \mathcal{K}}{\Sigma}n^{-\frac{\epsilon}{2}k_z} \big\{m^{-\frac{\epsilon}{4}}(d-1)\big\}^{(a_1+a_2+b_1+b_2)} \le \big(\underset{k_z\in [0,\delta n]}{\Sigma}n^{-\frac{\epsilon}{2}k_z}\big)\big(\underset{a_1\in [0,\delta m]}{\Sigma}\{m^{-\frac{\epsilon}{4}}(d-1)\}^{a_1}\big)\big(\underset{a_2\in [0,\delta m]}{\Sigma}\{m^{-\frac{\epsilon}{4}}(d-1)\}^{a_2}\big)\big(\underset{b_1\in [0,\delta m]}{\Sigma}\{m^{-\frac{\epsilon}{4}}(d-1)\}^{b_1}\big)(\underset{b_2\in [0,\delta m]}{\Sigma}\{m^{-\frac{\epsilon}{4}}(d-1)\}^{b_2})-1$ and the last $-1$ comes from the fact that $(0,0,0,0,0)\notin \mathcal{J}\cap \mathcal{K}$. Then apply $\underset{a_1\in [0,\delta m]}{\Sigma}\{m^{-\frac{\epsilon}{4}}(d-1)\}^{a_1}=\underset{a_2\in [0,\delta m]}{\Sigma}\{m^{-\frac{\epsilon}{4}}(d-1)\}^{a_2}=\underset{b_1\in [0,\delta m]}{\Sigma}\{m^{-\frac{\epsilon}{4}}(d-1)\}^{b_1}=\underset{b_2\in [0,\delta m]}{\Sigma}\{m^{-\frac{\epsilon}{4}}(d-1)\}^{b_2}$.

\end{itemize}

\textbf{(II) MLE Converse}

\noindent\textbf{Overview of the proof}:  We show that if the observation rate is below a certain threshold, then the worst-case probability of error does not approach to $0$ as $n\rightarrow \infty$ for \emph{any} estimator. 
To begin with, Lemma~\ref{lem:4} shows that it suffices to prove the statement above for the constrained version of MLE. 
Then the rest of the proof is similar to the proof in~\cite{Ahn_2018}. Specifically, we consider a genie-aided MLE by providing the constrained MLE with additional information of a small set where the ground-truth latent preference matrix lies. 
We first make our analysis tractable by designing a proper set that reveals a just-about-right amount of information about the ground-truth latent preference matrix.
We then show that the error probability for the genie-aided MLE becomes strictly larger than $0$. 
Since the error probability of a genie-aided MLE is always lower than that of MLE, we can conclude the error probability of the constrained MLE does not approach to $0$ as $n\rightarrow \infty$. The following diagram visualizes the proof dependencies.

\begin{tcolorbox}
\textbf{Proof dependencies}:\\
(II) of Theorem 1 $\longleftarrow$ Lemma 4, \textbf{Case 1, 2}\\
\phantom ~~~~~~ \textbf{Case 1} $\longleftarrow$ Lemma 5\\
\phantom ~~~~~~ \textbf{Case 2} $\longleftarrow$ Lemma 6, 7
\end{tcolorbox}

We need to show the following : for arbitrary $\epsilon>0$, if $p\le \frac{1}{(d^{\min}_H)^2} \max\big\{\frac{(1-\epsilon) \log n - \frac{n}{2}I_s}{\gamma m},\frac{(1-\epsilon)2\log m}{n}\big\}$, then $P_e^{\gamma}(\psi)\nrightarrow 0$ as $n \rightarrow \infty$ for any $\psi$. To prove the statement for any $\psi$, we need to consider the constrained maximum likelihood estimator.

Suppose $d_H(p_i, p_j)$ achieves the minimum Hellinger distance when $p_i = p_{d_0}, p_j = p_{d_0 +1}$.
Let $\mathcal{D}^{\gamma}:=\big\{X\in \{p_{d_0},p_{d_0+1}\}^{[n]\times [m]}:\|u_X-v_X\|_0=\lceil \gamma m \rceil\big\}$. Consider the maximum likelihood estimator $\psi_{ML, \mathcal{D}^{\gamma}}$ whose output is constrained in $\mathcal{D}^{\gamma}$. Let $R'$ be a ground-truth latent preference matrix chosen in $\mathcal{D}^{\gamma}$ where $A_{R'}=[\frac{n}{2}], B_{R'}=[n]\setminus [\frac{n}{2}]$,  $(u_{R'})_j=(v_{R'})_j=p_{d_0}$ for $j=1,2,\dots, (m-\lceil \gamma m\rceil)$ and $(u_{R'})_j=p_{d_0}, (v_{R'})_j=p_{d_0+1}$ for $j=(m-\lceil \gamma m\rceil +1),\dots, m$.

\begin{lemma}\label{lem:4}
$\underset{\psi}{\text{inf }}P_e^{\gamma}(\psi)\ge \Pr(\psi_{ML, \mathcal{D}^{\gamma}}(N^\Omega, G)\neq R | R=R')$.
\end{lemma}

\begin{proof}
\begin{align*}
\underset{\psi}{\text{inf }}P_e^{\gamma}(\psi) &= \underset{\psi}{\text{inf}} \text{ max}\big\{\Pr(\psi(N^\Omega, G)\neq R) : R\in \{p_1, p_2, \dots , p_d\}^{n\times m}, \|u_R-v_R\|_0=\lceil \gamma m \rceil \big\} \\
&\ge \underset{\psi}{\text{inf }} \Pr(\psi(N^\Omega, G)\neq R | R=R')  \\
&= \underset{\psi : \psi (N^{\Omega},G)\in \mathcal{D}^{\gamma}}{\text{inf }} \Pr(\psi(N^\Omega, G)\neq R | R=R')  \\
&~~~~(\because \text{ if there exist } N_0^{\Omega}, G_0 \text{ such that } \psi(N_0^{\Omega}, G_0) \notin \mathcal{D}^{\gamma}, \text{ we can decrease } \\
&\quad \quad~~ \Pr(\psi(N^{\Omega}, G)\neq R| R=R') \text{ by replacing } \psi(N_0^{\Omega}, G_0) \text{ with any element in } \mathcal{D}^{\gamma}.) \\
&= \Pr(\psi_{ML, \mathcal{D}^{\gamma}}(N^\Omega, G)\neq R | R=R').  \\
&~~~~(\because \text{maximum likelihood estimator is optimal under uniform prior.})
\end{align*}
\end{proof}

By Lemma~\ref{lem:4}, it suffices to show that $\Pr(\psi_{ML, \mathcal{D}^{\gamma}}(N^\Omega, G)\neq R | R=R') \nrightarrow 0$ as $n \rightarrow \infty$. Let $S:=\underset{X\neq R', X\in \mathcal{D}^{\gamma}}{\cap}\big[L(X)>L(R')\big]$ which is the success event of $\psi_{ML, \mathcal{D}^{\gamma}}$ where $R=R'$. Then it suffices to prove $\Pr (S)\rightarrow 0$. \big( $\because$ $\Pr (S)\rightarrow 0$ implies $\Pr (S^c)\rightarrow 1$, and $\Pr\big(\psi_{ML, \mathcal{D}^{\gamma}}(N^\Omega, G)\neq R | R=R'\big) \ge \frac{1}{2} \Pr (S^c)$. The last inequality comes from the fact that $\Pr\big(\psi_{ML, \mathcal{D}^{\gamma}}(N^\Omega, G)\neq R | R=R'\big)\ge \frac{1}{2}$ conditioned on $S^c$.\big)

Now we consider genie-aided ML estimators to prove $\Pr (S)\rightarrow 0$; $\psi_{ML}^{(1)}$ is given with the information that the ground-truth latent preference matrix belongs to $\mathcal{D}^{\gamma} \cap \chi(0,0,0,1,1)$, and $\psi_{ML}^{(2)}$ is given with the information that the ground-truth latent preference matrix belongs to $\mathcal{D}^{\gamma} \cap \chi(1,0,0,0,0)$. Let $S^{(i)}$ be the success event of $\psi_{ML}^{(i)}$ for $i=1,2$. Then $S^{(1)}=\underset{X\in \mathcal{D}^{\gamma}\cap \chi(0,0,0,1,1)}{\cap}\big[L(X)>L(R')\big],~ S^{(2)}=\underset{X\in \mathcal{D}^{\gamma}\cap \chi(1,0,0,0,0)}{\cap}\big[L(X)>L(R')\big]$, and it is straightforward that $\Pr(S) \le \Pr(S^{(1)})$ and $\Pr(S) \le \Pr(S^{(2)})$. Note that $p\le \frac{1}{(d^{\min}_H)^2} \max\big\{\frac{(1-\epsilon) \log n - \frac{n}{2}I_s}{\gamma m},\frac{(1-\epsilon)2\log m}{n}\big\} \Leftrightarrow \frac{1}{2}np(d^{\min}_H)^2\le (1-\epsilon) \log m$ or $\frac{1}{2}nI_s+\gamma m p(d^{\min}_H)^2\le (1-\epsilon) \log n$. Hence it is enough to show the following: (i) if $\frac{1}{2}np(d^{\min}_H)^2\le (1-\epsilon) \log m$, then $\Pr(S^{(1)}) \rightarrow 0$ as $n \rightarrow \infty$, and (ii) if $\frac{1}{2}nI_s+\gamma m p(d^{\min}_H)^2\le (1-\epsilon) \log n$, then $\Pr(S^{(2)})\rightarrow 0$ as $n \rightarrow \infty$.

\textbf{Case 1.} $\frac{1}{2}np(d^{\min}_H)^2\le (1-\epsilon) \log m$:

We first need to observe the following fact. Consider $X_1\in\chi(0,0,0,1,0)\cap \mathcal{D}^{\gamma}, X_2\in\chi(0,0,0,0,1)\cap \mathcal{D}^{\gamma}, X_3\in\chi(0,0,0,1,1)\cap \mathcal{D}^{\gamma}$ where $v_{X_1}$ differs from $v_{R'}$ at $i_1$-th coordinate, $v_{X_2}$ differs from $v_{R'}$ at $i_2$-th coordinate, $v_{X_3}$ differs from $v_{R'}$ at $i_1$ and $i_2$-th coordinates. Then $\big[L(X_1)\le L(R') \text{ and } L(X_2)\le L(R')\big]$ implies $\big[L(X_3)\le L(R')\big]$ since $L(X_3)-L(R')=(L(X_2)-L(R'))+(L(X_1)-L(R'))$ by Lemma 1. Hence $\big[L(X_3)> L(R')\big]$ implies $\big[L(X_1)> L(R') \text{ or } L(X_2)> L(R')\big]$. From this observation, we can show that 
$$\underset{X\in\chi(0,0,0,1,1)\cap \mathcal{D}^{\gamma}}{\cap}\big[L(X)>L(R')\big] \subset \big\{\underset{X_1\in\chi(0,0,0,1,0)\cap \mathcal{D}^{\gamma}}{\cap}\big[L(X_1)>L(R')\big]\big\} \cup \big\{\underset{X_2\in\chi(0,0,0,0,1)\cap \mathcal{D}^{\gamma}}{\cap}\big[L(X_2)>L(R')\big]\big\}$$
\big($\because$ Suppose $L(X)>L(R')$ for all $X\in \chi(0,0,0,1,1)\cap \mathcal{D}^{\gamma}$. If $L(X_1)>L(R')$ for all $X_1\in \chi(0,0,0,1,0)\cap \mathcal{D}^{\gamma}$, we are done. If not, there exists $X_1\in \chi(0,0,0,1,0)\cap \mathcal{D}^{\gamma}$ such that $L(X_1)\le L(R')$. Let $v_{X_1}$ differs from $v_{R'}$ at $i_1$-th coordinate. Consider $X_{3,j}\in \chi(0,0,0,1,1)\cap \mathcal{D}^{\gamma}$ where $v_{X_{3,j}}$ differs from $v_{R'}$ at $i_1$ and $(m-\lceil \gamma m \rceil +j)$-th coordinates, and $X_j\in \chi(0,0,0,0,1)\cap \mathcal{D}^{\gamma}$ where $v_{X_j}$ differs from $v_{R'}$ at $(m-\lceil \gamma m \rceil +j)$-th coordinate ($j=1,\dots, \lceil \gamma m \rceil$). Using the observation above and the fact that $L(X_{3,j})>L(R')$ (by the assumption) together, we can conclude that $L(X_j)>L(R')$ for all $j=1,\dots, \lceil \gamma m \rceil$\big)\\
Applying the union bound, we get
$\Pr\big(\underset{X\in\chi(0,0,0,1,1)\cap \mathcal{D}^{\gamma}}{\cap}\big[L(X)>L(R')\big]\big) \le \Pr\big(\underset{X_1\in\chi(0,0,0,1,0)\cap \mathcal{D}^{\gamma}}{\cap}\big[L(X_1)>L(R')\big]\big) + \Pr\big(\underset{X_2\in\chi(0,0,0,0,1)\cap \mathcal{D}^{\gamma}}{\cap}\big[L(X_2)>L(R')\big]\big)$. Now it suffices to show that $\Pr\big(\underset{X_1\in\chi(0,0,0,1,0)\cap \mathcal{D}^{\gamma}}{\cap}\big[L(X_1)>L(R')\big]\big)\rightarrow 0$ as $n\rightarrow \infty$. \big(Identical argument can be applied to $\Pr\big(\underset{X_2\in\chi(0,0,0,0,1)\cap \mathcal{D}^{\gamma}}{\cap}\big[L(X_2)>L(R')\big]\big)\rightarrow 0$ as $n \rightarrow \infty$.\big)

\begin{lemma}\label{lem:5}
For integers $K,L>0$, Let $\{A_i\}_{i=1}^{K} \overset{i.i.d.}{\sim} \text{Bern}(\alpha), \{B_i\}_{i=1}^{K} \overset{i.i.d.}{\sim} \text{Bern}(\beta), \{\mathbb{P}_i\}_{i=1}^{L} \overset{i.i.d.}{\sim} \text{Bern}(p), \{\mathbb{P}_{d_0,i}\}_{i=1}^{L} \overset{i.i.d.}{\sim} \text{Bern}(p_{d_0})$. Assume that $\alpha, \beta, p=o(1)$ and $\max\{\sqrt{\alpha\beta}K,pL\}=\omega(1)$. Then the following holds for sufficiently large K if $\sqrt{\alpha\beta}K>pL$; sufficiently large L otherwise:\\
$\Pr\Big(\log \big(\frac{\alpha(1-\beta)}{(1-\alpha)\beta}\big) \underset{1\le i \le K}{\sum}(B_i-A_i)+\underset{1\le i \le L}{\sum}\mathbb{P}_i\big\{\mathbb{P}_{d_0,i} \log \big(\frac{p_{d_0+1}}{p_{d_0}}\big)+ (1-\mathbb{P}_{d_0,i}) \log \big(\frac{1-p_{d_0+1}}{1-p_{d_0}}\big)\big\}\ge 0\Big) \ge \frac{1}{4}e^{-(1+o(1))(KI_s+LI_r)}$\\
where $I_s:=-2\log \big(1- d^2_H(\alpha, \beta)\big), I_r=p(d^{\min}_H)^2$.
\end{lemma}

\begin{proof}
Can be proved similarly by applying the argument of Lemma 4 in~\cite{Ahn_2018}.
\end{proof}

Then 
\begingroup
\allowdisplaybreaks
\begin{align*}
&\Pr\big(\underset{X_1\in\chi(0,0,0,1,0)\cap \mathcal{D}^{\gamma}}{\cap}\big[L(X_1)>L(R')\big]\big)\\
&=\underset{X_1\in\chi(0,0,0,1,0)\cap \mathcal{D}^{\gamma}}{\Pi}\Pr\big(\big[L(X_1)>L(R')\big]\big) ~~\big(\because \big\{\big[L(X_1)>L(R')\big]\big\}_{X_1\in\chi(0,0,0,1,0)\cap \mathcal{D}^{\gamma}} \text{ is mutually independent}.\big)\\
&=\underset{X_1\in\chi(0,0,0,1,0)\cap \mathcal{D}^{\gamma}}{\Pi}\big(1-\Pr\big(\big[L(X_1)>L(R')\big]\big)\big)\\
&=\underset{X_1\in\chi(0,0,0,1,0)\cap \mathcal{D}^{\gamma}}{\Pi}\Big(1-\Pr\Big(\underset{1\le i\le \frac{n}{2}}{\sum}\mathbb{P}_i\big\{\mathbb{P}_{d_0,i} \log \big(\frac{p_{d_0+1}}{p_{d_0}}\big)+ (1-\mathbb{P}_{d_0,i}) \log \big(\frac{1-p_{d_0+1}}{1-p_{d_0}}\big)\big\}\ge 0\Big)\Big) ~~\big(\because \text{Lemma~\ref{lem:2}}\big)\\
&\le \underset{X_1\in\chi(0,0,0,1,0)\cap \mathcal{D}^{\gamma}}{\Pi}\Big(1-\frac{1}{4}e^{-(1+o(1))(\frac{n}{2} I_r)}\Big)\quad \big(\because \text{apply Lemma~\ref{lem:5} with } K=0, L=\frac{n}{2} \big)\\
&\le \underset{X_1\in\chi(0,0,0,1,0)\cap \mathcal{D}^{\gamma}}{\Pi}\Big(e^{-\frac{1}{4}e^{-(1+o(1))(\frac{n}{2} I_r)}}\Big) \quad \big(\because 1-x\le e^{-x}\big)\\
&= \exp\Big\{-\frac{1}{4}|\chi(0,0,0,1,0)\cap \mathcal{D}^{\gamma}|e^{-(1+o(1))(\frac{n}{2} I_r)}\Big\} \\
&= \exp\Big\{-\frac{1}{4}(1-\gamma)m e^{-(1+o(1))(\frac{n}{2} I_r)}\Big\} \\
&= \exp\Big\{-\frac{1}{4}(1-\gamma) e^{-(1+o(1))(\frac{n}{2} I_r)+\log m}\Big\} \rightarrow 0 \text{ as }n\rightarrow \infty \qquad \big(\because \frac{1}{2}nI_r\le (1-\epsilon) \log m \big)
\end{align*}
\endgroup

\textbf{Case 2.} $\frac{1}{2}nI_s+\gamma m I_r=\frac{1}{2}nI_s+\gamma m p(d^{\min}_H)^2\le (1-\epsilon) \log n$:
    
From the assumption, $\alpha,\beta=O(\frac{\log n}{n})$.

\begin{lemma}\label{lem:6}
 Suppose $\alpha=O(\frac{\log n}{n})$, and consider the following procedure:\\
1) For $r=\frac{n}{\log^3 n}$, let $T:=\big\{1,2,\dots, 2r\big\}\cup\big\{\frac{n}{2}+1,\frac{n}{2}+2,\dots, \frac{n}{2}+2r\big\}$.\\
2) Within T, we will delete every pair of two nodes which are adjacent.\\
3) Denote the remaining nodes by U.\\
Then the above procedure results in $|U|\ge 3\frac{n}{\log^3 n}$, with probability approaching to 1.
\end{lemma}
\begin{proof}
Lemma 5 in~\cite{Ahn_2018}. 
\end{proof}

Let $\Delta$ be the event \big[$|U| \ge 3\frac{n}{\log^3 n}$\big]. Let $\Psi$ be the event \big[there extist subsets $A_P\subset A_R$ and $A_Q\subset B_R$ such that (i) $|A_P|=|A_Q|=\frac{n}{\log^3 n}$ and (ii) there is no edge between nodes in $A_P\cup A_Q$\big]. One can show that $\Delta \subset \Psi$. As $Pr(\Delta)=1-o(1)$ by Lemma~\ref{lem:6}, $Pr(\Psi)=1-o(1)$. 

Let $X^{(i)}$ be the latent preference matrix obtained from X by replacing $i$-th row with $v_{X}$ if $i\in A_X$; with $u_{X}$ otherwise. In explicit, $A_{X^{(i)}}=A_X\bigtriangleup \{i\}$ and $B_{X^{(i)}}=B_X\bigtriangleup \{i\}$

\begin{lemma}\label{lem:7}
 Suppose that $L(R^{(i)})\le L(R)$ and $L(R^{(j)})\le L(R)$ hold for $i\in A_P$ and $j\in A_Q$. Then, conditioned on $\Psi$, $L\big((R^{(i)})^{(j)}\big)\le L(R)$.
\end{lemma}
\begin{proof}
Lemma 6 in~\cite{Ahn_2018}. 
\end{proof}

Now we can find the upper bound of $\Pr\big(S^{(2)}\big)=\Pr\big(\underset{X\in \mathcal{D}^{\gamma}\cap \chi(1,0,0,0,0)}{\cap}\big[L(X)>L(R')\big]\big)$ as follows. 
\begingroup
\allowdisplaybreaks
\begin{align*}
&\Pr\big(S^{(2)}\big)=\Pr\big(\underset{X\in \mathcal{D}^{\gamma}\cap \chi(1,0,0,0,0)}{\cap}\big[L(X)>L(R')\big]\big) \\
&=\Pr\Big(\big\{\big(\underset{X\in \mathcal{D}^{\gamma}\cap \chi(1,0,0,0,0)}{\cap}\big[L(X)>L(R')\big]\big)\cap \Psi \big\}\cup \big\{\big(\underset{X\in \mathcal{D}^{\gamma}\cap \chi(1,0,0,0,0)}{\cap}\big[L(X)>L(R')\big]\big)\cap \Psi^c \big\}\Big) \\
&=\Pr\Big(\big(\underset{X\in \mathcal{D}^{\gamma}\cap \chi(1,0,0,0,0)}{\cap}\big[L(X)>L(R')\big]\big)\cap \Psi \Big)+ \Pr\Big(\big(\underset{X\in \mathcal{D}^{\gamma}\cap \chi(1,0,0,0,0)}{\cap}\big[L(X)>L(R')\big]\big)\cap \Psi^c \Big) \\
&\le\Pr\Big(\big(\underset{X\in \mathcal{D}^{\gamma}\cap \chi(1,0,0,0,0)}{\cap}\big[L(X)>L(R')\big]\big)\cap \Psi \Big)+\Pr\big(\Psi^c\big) \\
&=\Pr\Big(\big(\underset{X\in \mathcal{D}^{\gamma}\cap \chi(1,0,0,0,0)}{\cap}\big[L(X)>L(R')\big]\big)\cap \Psi \Big)+\big(1-(1-o(1))\big) \\
&=\Pr\Big(\underset{X\in \mathcal{D}^{\gamma}\cap \chi(1,0,0,0,0)}{\cap}\big[L(X)>L(R')\big]\big| \Psi \Big)\cdot \Pr\big(\Psi\big)+o(1) \\
&\le\Pr\Big(\underset{i\in A_P, j\in A_Q}{\cap}\big[L\big((R^{(i)})^{(j)}\big)>L(R')\big]\big|\Psi\Big)\cdot \big(1-o(1)\big)+o(1) \\
&\le\Pr\Big(\underset{i\in A_P}{\cap}\big[L(R^{(i)})>L(R')\big]\big|\Psi\Big)\cdot \big(1-o(1)\big)+\Pr\Big(\underset{j\in A_Q}{\cap}\big[L(R^{(j)})>L(R')\big]\big|\Psi\Big)\cdot \big(1-o(1)\big)+o(1) \\
&~~~~\big(\because \text{By Lemma~\ref{lem:7}, } L\big((R^{(i)})^{(j)}\big)>L(R') \text{ implies either one of the following}  \\
&~~\quad\text{(i) } L(R^{(i)})>L(R') \text{ for any } i\in A_P, \text{ or } \text{(ii) }L(R^{(j)})>L(R') \text{ for any } j\in A_Q \big) \\
&=2\Pr\Big(\underset{i\in A_P}{\cap}\big[L(R^{(i)})>L(R')\big]\big|\Psi\Big)\cdot \big(1-o(1)\big)+o(1) ~~~~ \big(\because \textsl{by symmetry}\big)\\
&=2\Pr\Big(\big[L(R^{(1)})>L(R')\big]\big|\Psi\Big)^{|A_P|}\cdot \big(1-o(1)\big)+o(1)\\
&~~~~\big(\because\text{WLOG, $1 \in A_P$.}\{[L(R^{(i)})>L(R')]\}_{i\in A_P} \text{ is mutually independent since there is no edge in $A_P$.}\big)
\end{align*}
\endgroup
\\
The upper bound of $\Pr\Big(\big[L(R^{(1)})>L(R')\big]\big|\Psi\Big)^{|A_P|}$ can be computed as follows. Let $c_s := \log \big(\frac{\alpha(1-\beta)}{(1-\alpha)\beta}\big)$
\begingroup
\allowdisplaybreaks
\begin{align*}
&\quad\Pr\Big(L(R^{(1)})>L(R')\big|\Psi\Big)^{|A_P|}  \\
&\le\Big[\Pr\Big(c_s\big\{e(A_R,B_R)-e(A_{R^{(1)}},B_{R^{(1)}})\big\}+ \underset{1\le i\le \gamma m}{\sum}\mathbb{P}_i\big\{\mathbb{P}_{d_0,i} \log \big(\frac{p_{d_0+1}}{p_{d_0}}\big)+ (1-\mathbb{P}_{d_0,i}) \log \big(\frac{1-p_{d_0+1}}{1-p_{d_0}}\big)\big\} <0\big|\Psi\Big)\Big]^{|A_P|}  \\
&\le\Big[\Pr\Big(c_s\big\{e(\{1\},B_R)-e(\{1\},A_R\Delta \{1\})\big\}+\underset{1\le i\le \gamma m}{\sum}\mathbb{P}_i\big\{\mathbb{P}_{d_0,i} \log \big(\frac{p_{d_0+1}}{p_{d_0}}\big)+ (1-\mathbb{P}_{d_0,i}) \log \big(\frac{1-p_{d_0+1}}{1-p_{d_0}}\big)\big\} <0\big|\Psi\Big)\Big]^{|A_P|}  \\
&\le\Big[\Pr\Big(c_s\sum_{i=1}^{\frac{n}{2}-r}(B_i-A_i)+\underset{1\le i\le \gamma m}{\sum}\mathbb{P}_i\big\{\mathbb{P}_{d_0,i} \log \big(\frac{p_{d_0+1}}{p_{d_0}}\big)+ (1-\mathbb{P}_{d_0,i}) \log \big(\frac{1-p_{d_0+1}}{1-p_{d_0}}\big)\big\} <0\Big)\Big]^{|A_P|}   \qquad (r=|A_P|)\\
&\le\big\{1-\frac{1}{4}e^{-(1+o(1))(\frac{n}{2}-r)I_s-(1+o(1))\gamma mI_r}\big\}^{|A_P|}  
\quad \big(\because \text{apply Lemma~\ref{lem:5} with }  K=\frac{n}{2}-r \text{ and } L=\gamma m \big) \\
&\le \Big[exp\big\{-\frac{1}{4}e^{-(1+o(1))(\frac{n}{2}-r)I_s-(1+o(1))\gamma mI_r}\big\}\Big]^{|A_P|}   \qquad \big(\because 1-x\le e^{-x}\big)\\
&= \exp\big\{-\frac{n}{ \log^3 n}\frac{1}{4}e^{-(1+o(1))(\frac{n}{2}-r)I_s-(1+o(1))\gamma mI_r}\big\}  \qquad \big(\because |A_P|=\frac{n}{ \log^3 n}\big)\\
&\le \exp\big\{-\frac{n}{ \log^3 n}\frac{1}{4}e^{-(1-\frac{\epsilon}{2}) \log n}\big\}  \quad \text{for sufficiently large n} \\
&~~~~\big(\because -(1+o(1))\big(\frac{n}{2}-r\big)I_s-(1+o(1))\gamma mI_r \ge -\big(1-\frac{\epsilon}{2}\big) \log n\quad \text{for sufficiently large }n\big) \\
&= \exp\big\{-\frac{1}{4}\frac{n^{\frac{\epsilon}{2}}}{ \log^3 n}\big\} \rightarrow 0 \text{ as } n\rightarrow \infty
\end{align*}
\endgroup
Hence we can conclude that $\Pr\big(S^{(2)}\big)\rightarrow 0$ as $n \rightarrow \infty$.

\subsection{Proof of Theorem~\ref{thm:2}} \label{sec:D:2}
\noindent\textbf{Overview of the proof}: The proof of Theorem~\ref{thm:2} consists of two parts; MLE achievability and MLE converse. Both parts can be proved by combining the technique developed in the proof of Theorem~\ref{thm:1} and the technique of~\cite{8636058}.

In this section, we provide the full statement and the proof of  Theorem~\ref{thm:2}. Recall that $R\in \{p_1,p_2, \dots, p_d\}^{n\times m}$ is a ground-truth latent preference matrix, $C : [n] \rightarrow [K]$ is a cluster assignment function, $u_k \in \{p_1,\dots,p_d\}^m$ is a latent preference vector whose cluster assignment is $k\in [K]$, $c_k := |C^{-1}(\{k\})|, c_{i,j}:=\frac{c_i +c_j}{2}$, $d_0 := \underset{i\in \{1,2, \dots, d-1\}  }{\arg \max}\big(\sqrt{p_i p_{i+1}}+\sqrt{(1-p_i)(1-p_{i+1})} \big)$.
Define $\mathbb{p}:\{p_1, \dots, p_d\}^m \rightarrow \{p_{d_0}, p_{d_0+1}\}^m$ that sends each coordinate $x_i$  to $p_{d_0}$ if $x_i \le p_{d_0}$; $p_{d_0+1}$ if $x_i \ge p_{d_0+1}$.

\begin{theorem} Let $m=\omega (\log n)$, $\log m=o(n)$, $ \underset{n\rightarrow \infty}{\lim\inf} \frac{c_k}{n}>0~\text{for all}~k\in [K]$, $ \underset{m\rightarrow \infty}{\lim\inf} \frac{\|\mathbb{p}(u_i)-\mathbb{p}(u_j)\|_0}{m}>0~\text{for all}~i\neq j \in[K]$. Then, the following holds for arbitrary $\epsilon>0$.

(I) If $p\ge$ $\frac{1}{(d^{\min}_H)^2} \max \{\underset{i\neq j \in [K]}{\max}\{\frac{(1+\epsilon)\log n - c_{i,j}I_s}{\|\mathbb{p}(u_i)-\mathbb{p}(u_j)\|_0}\}, \underset{k\in [K]}{\max}\{\frac{(1+\epsilon)\log m}{c_k}\}\}$,
then there exists an estimator $\psi$ such that $\Pr(\psi(N^\Omega, G)\neq R)\rightarrow 0$ as $n \rightarrow \infty$

(II) Suppose $R\in \{p_{d_0},p_{d_0+1}\}^{n\times m}$, $\alpha = O(\frac{\log n}{n})$. 
 If $p\le$ $\frac{1}{(d^{\min}_H)^2} \max \{\underset{i\neq j \in [K]}{\max}\{\frac{(1-\epsilon)\log n - c_{i,j}I_s}{\|\mathbb{p}(u_i)-\mathbb{p}(u_j)\|_0}\}, \underset{k\in [K]}{\max}\{\frac{(1-\epsilon)\log m}{c_k}\}\}$, then $\Pr(\psi(N^\Omega, G)\neq R)\nrightarrow 0$ as $n \rightarrow \infty$ for any $\psi$.
\end{theorem}

We will first show that the maximum likelihood estimator $\psi_{ML}$ satisfies (I), and then show that there does not exist an estimator $\psi$ satisfying (II).

\textbf{(I) MLE Achievability}\\
We introduce a few more notations that will be used in the proof. For an arbitrary latent preference matrix $X$, let 
$C^X : [n] \rightarrow [K]$ be a cluster assignment function, $u^X_k \in \{p_1,\dots,p_d\}^m$ be a latent preference vector whose cluster assignment is $k\in [K]$, $C^X_k:= (C^X)^{-1}(\{k\})$, $C^X_{i,j}:=C^{R}_i\cap C^{X}_j$. Then $\{C^X_{i,j}: 1\le i,j \le K\}$ is a $K^2$-partition of $[n]$. In light of the proof of \textbf{Claim 1} in~\cite{8636058} together with Lemma~\ref{lem:2}, we get 
\begin{align*}
\Pr(L(X)\le L(R))&=\Pr\bigg(\underset{1\le a\neq b \le d}{\sum}\Big[\underset{(i,j)\in Q_{ab}}{\sum}\mathbb{P}_{ij}\big\{\mathbb{P}_{a,ij}\log \frac{p_b}{p_a} + (1-\mathbb{P}_{a,ij}) \log \frac{1-p_b}{1-p_a} \big\}\Big]\\
&~~~~~~~~~~~~~~ +(N_B -N_A)\log \Big(\frac{1-\alpha}{1-\beta}\Big) +\log \Big(\frac{\alpha(1-\beta)}{(1-\alpha)\beta}\Big) \Big( \underset{1\le i \le N_B}{\sum B_i} -\underset{1\le i \le N_A}{\sum A_i} \Big)\ge 0\bigg)
\end{align*}
where $Q_{ab}:=\big\{(i,j)\in [n]\times [m] : R_{ij}=p_a, X_{ij}=p_b\big\},$ $N_A :=\underset{i,j,k\in [K], j\neq k}{\sum}|C^X_{i,j}|\cdot |C^X_{i,k}|$, $N_B:= \underset{i,j,k\in [K], j\neq k}{\sum}|C^X_{j,i}|\cdot |C^X_{k,i}|$, and $A_i\overset{i.i.d.}{\sim} \text{Bern}(\alpha)$, $B_i\overset{i.i.d.}{\sim} \text{Bern}(\beta)$, $\mathbb{P}_{ij}\overset{i.i.d.}{\sim} \text{Bern}(p)$, $\mathbb{P}_{a, ij}\overset{i.i.d.}{\sim} \text{Bern}(p_a)$.
Applying the technique of Lemma 2 in~\cite{8636058} together with Lemma~\ref{lem:3}, we get 
$$\Pr(L(X)\le L(R)) \le \{1-p(d^{\min}_H)^2\}^{D_X}\exp\Big(-\frac{N_A+N_B}{2}I_s\Big)$$
where $D_X:=|\{(i,j)\in [n]\times [m] : R_{ij}\neq X_{ij}\}|$. Then 
\begin{align*}
\Pr(\psi_{ML}(N^\Omega, G)\neq R)&\le \Pr\big(\underset{X\neq R}{\cup}[L(X)\le L(R)]\big)\overset{\text{union bound}}{\le} \underset{X\neq R}{\sum}\Pr(L(X)\le L(R))\\
& \le \underset{X\neq R}{\sum}\{1-p(d^{\min}_H)^2\}^{D_X}\exp(-\frac{N_A+N_B}{2}I_s)
\end{align*}
Let $\mathcal{X}:=\{X\neq R : X ~\text{is a latent preference matrix with $K$ clusters} \}$. It suffices to show that the last summation converges to $0$ as $n\rightarrow \infty$. We divide it into three partial sums over subsets $\mathcal{X}_1 := \{X\in \mathcal{X}: \exists i,j\neq k\in [K]~\text{such that}~ |C^X_{i,j}|,|C^X_{i,k}|\ge \delta n\} \cup \{X\in \mathcal{X}: \exists i,j\neq k\in [K]~\text{such that}~ |C^X_{j,i}|,|C^X_{k,i}|\ge \delta n\}$, $\mathcal{X}_2 := \{X\in \mathcal{X}\setminus \mathcal{X}_1 : \exists i\in [K] ~\text{such that}~ \|\mathbb{p}(u^R_i)-\mathbb{p}(u^R_j)\|_0 \ge \delta m \}$, $\mathcal{X}_3 := \mathcal{X}\setminus \{\mathcal{X}_1 \cup \mathcal{X}_2\}$. Applying the technique used in the proof of \textbf{Claim 1} in~\cite{8636058}, one can show that each partial sum converges to $0$ as $n\rightarrow \infty$.

\textbf{(II) MLE Converse}\\
$R\in \{p_{d_0},p_{d_0+1}\}^{n\times m}$ by the assumption. Let $\mathcal{D}:=\{p_{d_0},p_{d_0+1}\}^{n\times m}$. In light of Lemma~\ref{lem:4}, one can show that $\Pr(\psi(N^\Omega, G)\neq R) \ge \Pr(\psi_{ML}|_{\mathcal{D}}(N^\Omega, G)\neq R)$ where $\psi_{ML}|_{\mathcal{D}}$ is the maximum likelihood estimator whose output is constrained in $\mathcal{D}$. Hence it suffices to show that $\Pr(\psi_{ML}|_{\mathcal{D}}(N^\Omega, G)\neq R) \nrightarrow 0$ as $n \rightarrow \infty$. Let $S:=\underset{X\neq R, X\in \mathcal{D}}{\cap}\big[L(X)>L(R')\big]$ which is the success event of $\psi_{ML}|_{\mathcal{D}}$. Then one can observe that  $\Pr (S)\rightarrow 0$ implies $\Pr(\psi_{ML}|_{\mathcal{D}}(N^\Omega, G)\neq R) \nrightarrow 0$.\\
Note that $p\le$ $\frac{1}{(d^{\min}_H)^2} \max \{\underset{i\neq j \in [K]}{\max}\{\frac{(1-\epsilon)\log n - c_{i,j}I_s}{\|\mathbb{p}(u_i)-\mathbb{p}(u_j)\|_0}\}, \underset{k\in [K]}{\max}\{\frac{(1-\epsilon)\log m}{c_k}\}\} \Leftrightarrow $ ``$p(d^{\min}_H)^2\|\mathbb{p}(u_i)-\mathbb{p}(u_j)\|_0 + c_{i,j}I_s \le (1-\epsilon)\log n$ for some $i\neq j \in [K]$'' or ``$c_k p (d^{\min}_H)^2 \le (1-\epsilon)\log m $ for some $k\in [K]$''.
\begin{itemize}
    \item{Case 1. $p(d^{\min}_H)^2\|\mathbb{p}(u_i)-\mathbb{p}(u_j)\|_0 + c_{i,j}I_s \le (1-\epsilon)\log n$ for some $i\neq j \in [K]$:} Without loss of generality, assume that $i=1, j=2$. We consider a genie-aided ML estimator $\psi^{(1)}_{ML}$ which is given with the information that the ground-truth latent preference matrix belongs to $\mathcal{D}_1:=\{X\in \mathcal{D}: C^X_1\setminus C^R_1= C^R_1\setminus C^X_1 =1, C^X_2 = (C^R_1\cup C^R_2)\setminus C^X_1, C^X_i = C^R_i ~\text{for}~ i=3,\dots,K\}$. Then one can show that $\Pr(S):=\Pr(\underset{X\neq R, X\in \mathcal{D}}{\cap}\big[L(X)>L(R')\big]) \le \Pr(\underset{X\neq R, X\in \mathcal{D}_1}{\cap}\big[L(X)>L(R')\big]) \rightarrow 0$ by using the technique developed in~\cite{8636058}.
       
    \item{Case 2. $c_k p (d^{\min}_H)^2 \le (1-\epsilon)\log m $ for some $k\in [K]$:} Without loss of generality, assume that $k=1$. We consider a genie-aided ML estimator $\psi^{(2)}_{ML}$ which is given with the information that the ground-truth latent preference matrix belongs to $\mathcal{D}_2:=\{X\in \mathcal{D}: \|u^X_1-u^R_1\|_0 =1,  u^X_i = u^R_i ~\text{for}~ i=2,\dots,K\}$. Then one can show that $\Pr(S):=\Pr(\underset{X\neq R, X\in \mathcal{D}}{\cap}\big[L(X)>L(R')\big]) \le \Pr(\underset{X\neq R, X\in \mathcal{D}_2}{\cap}\big[L(X)>L(R')\big]) \rightarrow 0$ by using the technique developed in~\cite{8636058}.
    
\end{itemize}

\subsection{Proof of Theorem~\ref{thm:3}} \label{sec:D:3}
\noindent\textbf{Overview of the proof}: For \textbf{Stage 1}, we make use of the standard performance guarantee of spectral clustering algorithms, and this step is identical to the argument in~\cite{Ahn_2018}. 

\noindent Our theoretical contribution lies in the analysis of \textbf{Stage 2-(i)}.
In~\cite{Ahn_2018}, the authors have considered the symmetric case with $d=2$, i.e., $p_1 + p_2 = 1$. 
Thus, one just needs to estimate a single parameter $p_1$, making the entire parameter estimation part straightforward. 
On the other hand, we do have to estimate $d$ parameters $p_1, p_2, \ldots, p_d$, making our estimation algorithm more complicated and complicating our analysis. 

\noindent To obtain a theoretical guarantee of \textbf{Stage 2-(i)}, we first show that $2d\lceil \log m \rceil$ number of estimations $a_j$'s and $a_j'$'s satisfy the following in Lemma~\ref{lem:8}; (i) every $p_t\in \{p_1,\dots,p_d\}$ will be estimated by at least one of $a_j$'s or $a_j'$'s with probability approaching $1$ as $n \rightarrow \infty$, (ii) $a_j$'s and $a_j'$'s are located in the $o(1)$-radius neighborhoods of ground-truth latent preference levels $p_1,\dots, p_d$ with probability approaching $1$ as $n \rightarrow \infty$. The next step is a distance-based clustering on the distribution of $p_1,\dots,p_d$ which will give us $\hat{p_1},\dots,\hat{p_d}$ where $\hat{p_k}$ is indeed the average of numbers whose distance from $p_k$ can be arbitrarily small as $n\rightarrow\infty$. Then for each pair of cluster and column, we assign one of $\hat{p_1},\dots,\hat{p_d}$ whose likelihood is maximum for that pair. This gives us $\hat{u_R}, \hat{v_R}$ which are the estimations of latent preference vectors, and Lemma~\ref{lem:9} ensures that $\hat{u_R}^{(j)}\rightarrow u_R^{(j)}, \hat{v_R}^{(j)}\rightarrow v_R^{(j)}$ for all $j=1,\dots, m$ with probability approaching to $1$ as $n\rightarrow \infty$.

\noindent \textbf{Stage 2-(ii)} is a local refinement step in which we compare estimated likelihood values and update cluster assignments, and the analysis is similar to the proof in~\cite{Ahn_2018}. We prove that under the conditions of Thm.~\ref{thm:2}, the number of wrongly classified users can be halved in each iteration, and hence one can successively improve the quality of the estimation, eventually achieving the perfect recovery. Note that the proof procedure is based on a standard successive refinement technique. The following diagram visualizes the proof dependencies.
 
 \begin{tcolorbox}
\textbf{Proof dependencies}:\\
Theorem 3 $\longleftarrow$ Analysis of \textbf{Stage 1, 2-(i), 2-(ii)} \\
\phantom ~~~~~~ Analysis of \textbf{Stage 2-(i)} $\longleftarrow$ Lemma 8, 9\\
\phantom ~~~~~~ Analysis of \textbf{Stage 2-(ii)} $\longleftarrow$ Lemma 10, 11, 12
\end{tcolorbox}

\begin{theorem}
Let $\ell_{\max} = 1, K=2, |C^{-1}(\{1\})|=|C^{-1}(\{2\})|=\frac{n}{2}, \gamma \in (0,1)$, $m=\omega (\log n)$, $\log m=o(n)$, $(\sqrt{\alpha}-\sqrt{\beta})^2=\omega(\frac{1}{n})$, $m=O(n)$, and $\alpha = O(\frac{\log n}{n})$. 
Let $\phi_j \cdot 2m$ be the number of $p_j$'s among $(u_R)_{1},\dots,(u_R)_{m},(v_R)_{1},\dots,(v_R)_{m}$ for $j=1,\dots,d$, and assume that $\phi_j \nrightarrow 0$ as $n\rightarrow \infty$. If $$p\ge \frac{1}{(d^{\min}_H)^2} \max\left\{\frac{(1+\epsilon) \log n - \frac{n}{2}I_s}{\gamma m},\frac{2(1+\epsilon)\log m}{n}\right\}$$ for some $\epsilon >0$, then our algorithm outputs $\hat{R}$ where the following holds with probability approaching to $1$ as $n$ goes to $\infty$ : $\|\hat{R}-R\|_{\max} := \underset{(i,j)\in [n]\times [m]}{\max}|\hat{R}_{ij}-R_{ij}|=o(1)$.
\end{theorem}

In Algorithm 1, (Stage 1) we first use spectral clustering to get $A_R^{(0)}, B_R^{(0)}$, (Stage 2-(i)) then get almost exact recovery of latent preference vectors $\hat{u_R}, \hat{v_R}$, (Stage 2-(ii)) and eventually get exact recovery of clusters $\hat{A_R}, \hat{B_R}$.

\textbf{Analysis of Stage 1.} Let $\eta :=\frac{|A_R^{(0)} \setminus A_R|}{n}$. Then $\eta \rightarrow 0$ as $n\rightarrow \infty$ with probability approaching to 1.

\begin{proof}
Since $(\sqrt{\alpha}-\sqrt{\beta})^2=\omega(\frac{1}{n})$ satisfies the assumption of Theorem 6 in~\cite{Gao_2017}, $\eta \rightarrow 0$ as $n\rightarrow \infty$ with probability approaching to $1$.
\end{proof}

\textbf{Analysis of Stage 2-(i).} Under the success of Stage 1, $\hat{u_R}\rightarrow u_R, \hat{v_R}\rightarrow v_R$ as $n\rightarrow \infty$ with probability approaching to 1.

\begin{proof}
It follows directly from Lemma~\ref{lem:9}, and we need to prove Lemma~\ref{lem:8} first.

\begin{lemma}\label{lem:8}
Sample $m_0:= d\lceil \log m \rceil$ elements $j_1,\dots, j_{m_0}$ from $[m]$ with replacement. Define $a_j=\frac{\underset{i\in A_R^{(0)}}{\sum}\mathbb{1}(N_{ij}^{\Omega}=1)}{\underset{i\in A_R^{(0)}}{\sum}\mathbb{1}(N_{ij}^{\Omega}=1  \text{ or }-1)}$, and $a_{j}'=\frac{\underset{i\in B_R^{(0)}}{\sum}\mathbb{1}(N_{ij}^{\Omega}=1)}{\underset{i\in B_R^{(0)}}{\sum}\mathbb{1}(N_{ij}^{\Omega}=1  \text{ or }-1)}$ for $j=j_1,\dots,j_{m_0}$. Let $q_1, \dots, q_{2m_0}$ be ground-truth latent preference levels corresponding to $a_{j_1},\dots,a_{j_{m_0}},a_{j_1}',\dots,a_{j_{m_0}}'$ respectively. (i) Then $\{q_1, \dots, q_{2m_0}\}=\{p_1, \dots, p_{d}\}$ with probability approaching to $1$ as $n\rightarrow \infty$. (ii) Moreover, for any constant $\delta >0$, the following holds with probability approaching to $1$ as $n\rightarrow \infty$: for all $i=1, \dots, m_0$, $|a_{j_i}-q_i|< \delta$ and $|a'_{j_i}-q_i|< \delta$.
\end{lemma}

\begin{proof}
(i) As there are $\phi_j \cdot 2m$ $p_j$'s among $(u_R)_{1},\dots,(u_R)_{m},(v_R)_{1},\dots,(v_R)_{m}$, $\Pr\big(\big[p_j\notin \{q_1, \dots, q_{2m_0}\}\big]\big)=(1-\phi_j)^{2m_0}\le (1-\delta_0)^{2m_0}$ ($\because$ $\phi_j \nrightarrow 0$ implies $\exists \delta_j>0$ such that $\phi_j\ge \delta_j$ for all $n$. Then define $\delta_0 := \text{min}\{\delta_1,\dots,\delta_d\}$). By union bound, $\Pr\big(\{q_1, \dots, q_{2m_0}\}\neq \{p_1, \dots, p_{d}\}\big)\le \underset{1\le j\le d}{\sum}\Pr\big(\big[p_j\notin \{q_1, \dots, q_{2m_0}\}\big]\big)\le \underset{1\le j\le d}{\sum} (1-\delta_0)^{2m_0}=d(1-\delta_0)^{2m_0}\rightarrow 0$ as $n\rightarrow \infty$. So $\{q_1, \dots, q_{2m_0}\}=\{p_1, \dots, p_{d}\}$ with probability approaching to $1$ as $n\rightarrow \infty$.

(ii) For $j=j_1$, $\Pr(|a_{j_1}-q_1|\ge \delta)=\Pr(a_{j_1}-q_1\ge \delta)+\Pr(a_{j_1}-q_1\le \delta)$. 
We first find the upper bound of $\Pr(a_{j_1}-q_1\ge \delta)$. 
\begingroup
\allowdisplaybreaks
\begin{align*}
&\Pr(a_{j_1}-q_1\ge\delta)=\Pr(a_{j_1}\ge q_1+\delta)\\
&=\Pr\Bigg(\frac{\underset{i\in A_R^{(0)}}{\sum}\mathbb{1}(N_{i{j_1}}^{\Omega}=1)}{\underset{i\in A_R^{(0)}}{\sum}\mathbb{1}(N_{i{j_1}}^{\Omega}=1  \text{ or }-1)}\ge q_1+\delta\Bigg)
=\Pr\Bigg(\frac{\underset{i\in A_R^{(0)}\cap A_R}{\sum}\mathbb{P}_i\mathbb{Q}_{1,i}+\underset{i\in A_R^{(0)}\setminus A_R}{\sum}\mathbb{P}_i\mathbb{P}_{f,i}}{\underset{i\in A_R^{(0)}}{\sum}\mathbb{P}_i}\ge q_1+\delta\Bigg)\\
&~~~~\big(\mathbb{P}_i\overset{i.i.d.}{\sim}\text{Bern}(p), \mathbb{Q}_{j,i}\overset{i.i.d.}{\sim}\text{Bern}(q_1), \mathbb{Q}_{f,i}\overset{i.i.d.}{\sim}\text{Bern}(p_f) ~\text{where} \\
&~~~~~~~~q_1,p_f \text{ are ground-truth latent preference levels of $j_1$-th column of } A_R, B_R \text{ respectively}.\big)\\
&=\Pr\Big(\underset{i\in A_R^{(0)}\cap A_R}{\sum}\mathbb{P}_i(\mathbb{Q}_{1,i}-q_1-\delta)+\underset{i\in A_R^{(0)}\setminus A_R}{\sum}\mathbb{P}_i(\mathbb{P}_{f,i}-q_1-\delta)\ge 0\Big)\\
&=\Pr\big(Y \ge 0\big) \quad~~~~ \Big(Y:=\underset{i\in A_R^{(0)}\cap A_R}{\sum}\mathbb{P}_i(\mathbb{Q}_{1,i}-q_1-\delta)+\underset{i\in A_R^{(0)}\setminus A_R}{\sum}\mathbb{P}_i(\mathbb{P}_{f,i}-q_1-\delta)\Big)\\
&= \Pr\big(e^{tY} \ge 1\big) \quad~~ \big(\text{ for all~~ t$>$0}\big) \quad \le \quad \underset{t>0}{\inf} \E[e^{tY}] \quad~~ \big(\because \text{Markov's inequality}\big)\\
&\le \E[e^{t_0Y}] \quad~~ \Big(t_0:=\log \Big(\frac{(q_1+\delta)(1-q_1)}{q_1(1-q_1-\delta)}\Big)\Big)\\
&= \underset{1\le i\le L_1}{\Pi}e^{t_0\mathbb{P}_i(\mathbb{Q}_{1,i}-q_1-\delta)} \underset{1\le i\le L_2}{\Pi}e^{t_0 \mathbb{P}_i(\mathbb{P}_{f,i}-q_1-\delta)}\quad~~\big(L_1:=|A_R^{(0)}\cap A_R|=\big(\frac{1}{2}-\eta\big)n, L_2:=|A_R^{(0)}\setminus A_R|=\eta n \big)\\
&= \underset{1\le i\le L_1}{\Pi}\Big\{(1-p)+p\Big(\frac{q_1^{q_1+\delta}(1-q_1)^{1-q_1-\delta}}{(q_1+\delta)^{q_1+\delta}(1-q_1-\delta)^{1-q_1-\delta}}\Big)\Big\} \underset{1\le i\le L_2}{\Pi}e^{t_0 \mathbb{P}_i(\mathbb{P}_{f,i}-q_1-\delta)}\\
&= \underset{1\le i\le L_1}{\Pi}(1-a_{\delta}p) \underset{1\le i\le L_2}{\Pi}e^{t_0 \mathbb{P}_i(\mathbb{P}_{f,i}-q_1-\delta)} \quad \text{for some constant }a_{\delta}\in (0,1)  \\
&~~~~\big(\because \text{Let }  G(x):=x^{q_1+\delta}(1-x)^{1-q_1-\delta}, \text{then } G'(x)=x^{q_1+\delta}(1-x)^{-q_1-\delta}\big\{(1-q_1-\delta)+(q_1+\delta)\frac{1-x}{x}\big\}. \\
&~~~~~~~~~~\text{Note that } G'(x)>0\text{ for } x\in (0,1), \text{ so } G(q_1+\delta)>G(q_1) \text{ which means }a_{\delta}=1-\frac{G(q_1)}{G(q_1+\delta)}\in (0,1). \big)\\
&\le \underset{1\le i\le L_1}{\Pi}(1-a_{\delta}p) \underset{1\le i\le L_2}{\Pi}(1+b_{\delta}p) \quad \text{for some constant }b_{\delta}\ge 0  \\
&~~~~~\big(\text{by calculating } \underset{1\le i\le L_2}{\Pi}e^{t_0 \mathbb{P}_i(\mathbb{P}_{f,i}-q_1-\delta)} \text{ directly, it is clear that such $b_{\delta}$ exists.}\big) \\
&= e^{(\frac{1}{2}-\eta)n\log (1-a_{\delta}p)}e^{\eta n\log (1+b_{\delta}p)} \\
&= e^{(\frac{1}{2}-\eta)n \big(-\frac{a_{\delta}}{(d^{\min}_H)^2}+o(1)\big)I_r} e^{\eta n \big(\frac{b_{\delta}}{(d^{\min}_H)^2}+o(1)\big)I_r} \\
&~~~~\big(\because \log (1-a_{\delta}p)=-a_{\delta}p+O(p^2)=I_r\big(-\frac{a_{\delta}}{(d^{\min}_H)^2}+o(1)\big)~~ \text{ where }~~I_r=p(d^{\min}_H)^2.\big)\\
&= e^{-\frac{a_{\delta}}{(d^{\min}_H)^2}nI_r\{(\frac{1}{2}-\eta) (1+o(1))-\eta  (\frac{b_{\delta}}{a_{\delta}}+o(1))\}}  \\
&\le e^{-\frac{a_{\delta}}{(d^{\min}_H)^2}(1+\frac{\epsilon}{2})\log m}  \\
&~~~~\big(\because nI_r\big\{\big(\frac{1}{2}-\eta\big) (1+o(1))-\eta  \big(\frac{b_{\delta}}{a_{\delta}}+o(1)\big)\big\}\ge \big(1+\frac{\epsilon}{2}\big)\log m \text{ for sufficiently large }n \\
&~~~~~~~~~~\text{since}~~~ \frac{1}{2}nI_r\ge (1+\epsilon)\log m \text{ and } \eta\rightarrow 0 \text{ as }n\rightarrow \infty. \big)\\
&= m^{-\frac{a_{\delta}}{(d^{\min}_H)^2}(1+\frac{\epsilon}{2})} = o(m^{-A}) ~~~~~~~~\big( A:=\frac{a_{\delta}}{(d^{\min}_H)^2}>0\big)
\end{align*}
\endgroup

So we get $\Pr(a_{j_1}-q_1\ge \delta)=o(m^{-A})$ for some $A>0$, and similarly, $\Pr(a_{j_1}-q_1\le \delta)=o(m^{-B})$ for some $B>0.$ Hence $\Pr\big(|a_{j_1}-q_1|\ge\delta\big)=o\big(m^{-\min\{A,B\}}\big)$. Note that A,B depend only on $p_1,\dots,p_d, \delta$ ~~\big($\because A=\frac{a_{\delta}}{(d^{\min}_H)^2}, a_{\delta}=1-\frac{G(q_1)}{G(q_1+\delta)}$ where $q_1\in \{p_1,\dots, p_d\}$\big), which means we can find a constant $A_{\delta}>0$ such that $\Pr(|a_{j_i}-q_i|\ge\delta)=o(m^{-A_{\delta}})$ and $\Pr(|a'_{j_i}-q_{m_0+i}|\ge\delta)=o(m^{-A_{\delta}})$ for $i=1,\dots, m_0$. Then
\begin{eqnarray}
&\quad& 1-\Pr\big(\underset{1\le i\le m_0}{\cap}\big[|a_{j_i}-q_i|< \delta\big] \cap \underset{1\le i \le m_0}{\cap}\big[|a'_{j_i}-q_{m_0+i}|< \delta\big] \big) \nonumber\\
&\le& \underset{1\le i\le m_0}{\sum}\Pr(|a_{j_i}-q_i|\ge \delta)+\underset{1\le i\le m_0}{\sum}\Pr(|a'_{j_i}-q_{m_0+i}|\ge \delta) \quad (\because \text{Union bound}) \nonumber\\
&=& \underset{1\le i\le m_0}{\sum}o(m^{-A_{\delta}})+\underset{1\le i\le m_0}{\sum}o(m^{-A_{\delta}}) \nonumber\\
&=& 2m_0 \cdot o(m^{-A_{\delta}}) = 2d \lceil \log m \rceil o(m^{-A_{\delta}}) \rightarrow 0 \text{ as } n\rightarrow \infty \nonumber
\end{eqnarray}
So we can conclude that $\Pr\big(\underset{1\le i\le m_0}{\cap}\big[|a_{j_i}-q_i|< \delta\big] \cap \underset{1\le i \le m_0}{\cap}\big[|a'_{j_i}-q_{m_0+i}|< \delta\big] \big)\rightarrow 1$ as $n\rightarrow \infty$.
\end{proof}

Applying Lemma~\ref{lem:8} with $\delta=\frac{1}{l}$, we have the following with probability approaching to $1$ as $n\rightarrow \infty$ : for all $i=1, \dots, m_0$, $|a_{j_i}-q_i|< \frac{1}{l}$ and $|a'_{j_i}-q_i|< \frac{1}{l}$. If we choose large enough $l$ satisfying $\frac{1}{l}<\frac{1}{5}\text{min}\{p_{i+1}-p_i:i=1,2,\dots,d-1\}$, one can show that $\hat{p_k}$ is a correct estimation of $p_k$ for $k=1,\dots,d$ (see Algorithm 1 for the definition of $\hat{p_k}$). In explicit, $\hat{p_k}$ is indeed the average of numbers whose distance from $p_k$ is less than $\frac{1}{l}$, hence we get $|\hat{p_k}-p_k|<\frac{1}{l}$. As we can choose arbitrary large $l$ by Lemma~\ref{lem:8}, we can observe $|\hat{p_k}-p_k|=o(1)$. Moreover, as there are finite number of choices for k,
the following holds with probability approaching to $1$ as $n\rightarrow \infty$:
\begin{equation}
    |\hat{p_k}-p_k|=o(1) ~~\text{for all }k=1,\dots, d.
\end{equation}
Note that $\frac{y}{x}$ and $\frac{1-y}{1-x}$ are continuous functions on $\mathbb{R^2}\setminus \{(x,y)| x\neq 0, x\neq 1\}$. Together with the facts that $0<p_1,\dots,p_d<1$ and that there are finite number of choices for $(i,j)$ where $1\le i,j\le d$, the following holds with probability approaching to $1$ as $n\rightarrow \infty$:
\begin{equation}\label{eq:4}
\frac{\hat{p_j}}{\hat{p_i}}=\frac{p_j}{p_i}(1+o(1)) ~\text{ and }~\frac{1-\hat{p_j}}{1-\hat{p_i}}=\frac{1-p_j}{1-p_i}(1+o(1))~\text{for all } i,j=1,\dots, d.
\end{equation}

\begin{lemma}\label{lem:9}
Define $\hat{u_R}^{(j)}=\underset{\hat{p_k}:k\in [d]}{\arg\min}\big(\underset{i\in A_R^{(0)}}{\sum}\big\{\mathbb{1}(N^{\Omega}_{ij}=1)(-\log \hat{p_k}) + \mathbb{1}(N^{\Omega}_{ij}=-1) (-\log (1-\hat{p_k}))\big\}\big)$, and $\hat{v_R}^{(j)}=\underset{\hat{p_k}:k\in [d]}{\arg\min}\big(\underset{i\in B_R^{(0)}}{\sum}\big\{\mathbb{1}(N^{\Omega}_{ij}=1)(-\log \hat{p_k}) + \mathbb{1}(N^{\Omega}_{ij}=-1) (-\log (1-\hat{p_k}))\big\}\big)$ for $j=1,\dots,m$. Then the following holds with probability approaching to $1$ as $n\rightarrow \infty$ : for all $j=1,\dots, m$, $\hat{u_R}^{(j)}\rightarrow u_R^{(j)}, \hat{v_R}^{(j)}\rightarrow v_R^{(j)}$.
\end{lemma}
\begin{proof}
Without loss of generality, assume $u_R^{(j)}=p_1$. Let $L_{\hat{u_R}^{(j)}}(\hat{p_k}):=\underset{i\in A_R^{(0)}}{\sum}\big\{\mathbb{1}(N^{\Omega}_{ij}=1)(-\log \hat{p_k}) + \mathbb{1}(N^{\Omega}_{ij}=-1) (-\log (1-\hat{p_k}))\big\})$. Then

\begingroup
\allowdisplaybreaks
\begin{align}
&\Pr(\hat{u_R}^{(j)}\neq \hat{p_1})\le \underset{2\le k\le d}{\sum}\Pr\big(L_{\hat{u_R}^{(j)}}(\hat{p_k})\le L_{\hat{u_R}^{(j)}}(\hat{p_1})\big)\quad~~~~\big(\because \text{ Union bound}\big)\nonumber\\
&= \underset{2\le k\le d}{\sum}\Big\{\Pr\big(\underset{i\in A_R^{(0)}\cap A_R}{\sum}\big\{\mathbb{1}(N^{\Omega}_{ij}=1)\log \frac{\hat{p_k}}{\hat{p_1}} + \mathbb{1}(N^{\Omega}_{ij}=-1) \log \frac{1-\hat{p_k}}{1-\hat{p_1}}\big\} \ge 0\big) \nonumber\\
&~~~~~~~~~~~~~~+ \Pr\big(\underset{i\in A_R^{(0)}\setminus A_R }{\sum}\big\{\mathbb{1}(N^{\Omega}_{ij}=1)\log \frac{\hat{p_k}}{\hat{p_1}} + \mathbb{1}(N^{\Omega}_{ij}=-1) \log \frac{1-\hat{p_k}}{1-\hat{p_1}}\big\} \ge 0\big)\Big\}\nonumber\\
&= \underset{2\le k\le d}{\sum}\Big\{\Pr\big(\underset{i\in A_R^{(0)}\cap A_R}{\sum}\big\{\mathbb{P}_i\mathbb{P}_{1,i}\log \frac{\hat{p_k}}{\hat{p_1}} + \mathbb{P}_i(1-\mathbb{P}_{1,i}) \log \frac{1-\hat{p_k}}{1-\hat{p_1}}\big\} \ge 0\big) \nonumber\\
&~~~~~~~~~~~~~~+ \Pr\big(\underset{i\in A_R^{(0)}\setminus A_R }{\sum}\big\{\mathbb{P}_i\mathbb{P}_{f,i}\log \frac{\hat{p_k}}{\hat{p_1}} + \mathbb{P}_i(1-\mathbb{P}_{f,i}) \log \frac{1-\hat{p_k}}{1-\hat{p_1}}\big\} \ge 0\big)\Big\}\nonumber\\
&~~~~~~~~\big(\mathbb{P}_i\overset{i.i.d.}{\sim}\text{Bern}(p), \mathbb{P}_{1,i}\overset{i.i.d.}{\sim}\text{Bern}(p_1), \mathbb{P}_{f,i}\overset{i.i.d.}{\sim}\text{Bern}(p_f)\nonumber\\
&~~~~~~~~~~\text{where $p_f$ is the ground-truth latent preference level corresponding to $j$-th column of $B_R$}\big)\nonumber\\
&= \underset{2\le k\le d}{\sum}\big\{\Pr(Z_{k,1}\ge 0)+\Pr(Z_{k,2}\ge 0)\big\} = \underset{2\le k\le d}{\sum}\big\{\Pr(e^{\frac{1}{2}Z_{k,1}}\ge 1)+\Pr(e^{\frac{1}{2}Z_{k,2}}\ge 1)\big\} \nonumber\\
&\le \underset{2\le k\le d}{\sum}\big(\E[e^{\frac{1}{2}Z_{k,1}}]+\E[e^{\frac{1}{2}Z_{k,2}}]\big) \quad~~ \big(\because \text{ Markov's inequality}\big)\nonumber\\
&= \underset{2\le k\le d}{\sum}\Big[\underset{i\in A_R^{(0)}\cap A_R}{\Pi}\Big\{pp_1\sqrt{\frac{\hat{p_k}}{\hat{p_1}}} + p(1-p_1)\sqrt{\frac{1-\hat{p_k}}{1-\hat{p_1}}} +(1-p)\Big\} \underset{i\in A_R^{(0)}\setminus A_R}{\Pi}\Big\{pp_f\sqrt{\frac{\hat{p_k}}{\hat{p_1}}} + p(1-p_f)\sqrt{\frac{1-\hat{p_k}}{1-\hat{p_1}}} +(1-p)\Big\} \Big]\nonumber\\
&=\underset{2\le k\le d}{\sum}\Big[\underset{i\in A_R^{(0)}\cap A_R}{\Pi}\Big\{pp_1\sqrt{\frac{p_k}{p_1}}(1+o(1)) + p(1-p_1)\sqrt{\frac{1-p_k}{1-p_1}}(1+o(1)) +(1-p)\Big\}\nonumber\\
&~~~~~~~~~~~~~~~\cdot \underset{i\in A_R^{(0)}\setminus A_R}{\Pi}\Big\{pp_f\sqrt{\frac{p_k}{p_1}}(1+o(1)) + p(1-p_f)\sqrt{\frac{1-p_k}{1-p_1}}(1+o(1)) +(1-p)\Big\}\Big] \quad~~ \big(\because (\ref{eq:4})\big)\nonumber\\
&=\underset{2\le k\le d}{\sum}\Big[\underset{i\in A_R^{(0)}\cap A_R}{\Pi}\Big\{p\sqrt{p_1 p_k}(1+o(1)) + p\sqrt{(1-p_1)(1-p_k)}(1+o(1)) +(1-p)\Big\}\nonumber\\
&~~~~~~~~~~~~~~~\cdot \underset{i\in A_R^{(0)}\setminus A_R}{\Pi}\Big\{pp_f\sqrt{\frac{p_k}{p_1}}(1+o(1)) + p(1-p_f)\sqrt{\frac{1-p_k}{1-p_1}}(1+o(1)) +(1-p)\Big\}\Big]\nonumber\\
&\le \underset{2\le k\le d}{\sum} \underset{i\in A_R^{(0)}\cap A_R}{\Pi}\big\{p (1-(d^{\min}_H)^2 ) (1+o(1))+(1-p)\big\}\cdot  \underset{i\in A_R^{(0)}\setminus A_R}{\Pi}(1+B_k p)\quad~~\big(\text{it is clear that such $B_k>0$ exists}\big)\nonumber\\
&= \underset{2\le k\le d}{\sum} e^{(\frac{1}{2}-\eta)n \log \{p(1-(d^{\min}_H)^2 )(1+o(1))+(1-p)\}}\cdot  e^{\eta n \log (1+B_k p)}\nonumber\\
&= \underset{2\le k\le d}{\sum}e^{(\frac{1}{2}-\eta)n (-1+o(1))I_r} e^{\eta n \big(\frac{B_{k}}{(d^{\min}_H)^2}+o(1)\big)I_r} \nonumber\\
&~~~~~~\big(\because \log \big(1-p\big\{1-(1-(d^{\min}_H)^2 )\big(1+o(1)\big)\big\}\big)=-p(d^{\min}_H)^2 (1+o(1))+O(p^2)=I_r(-1+o(1)) \text{ where }I_r=p(d^{\min}_H)^2\big)\nonumber\\
&= \underset{2\le k\le d}{\sum}e^{(\frac{1}{2}-\eta)n (-1+o(1))I_r+\eta n \big(\frac{B_{k}}{(d^{\min}_H)^2}+o(1)\big)I_r}=\underset{2\le k\le d}{\sum}e^{-nI_r\big\{(\frac{1}{2}-\eta) (1+o(1))-\eta  \big(\frac{B_{k}}{(d^{\min}_H)^2}+o(1)\big)\big\}} \nonumber\\
&\le\underset{2\le k\le d}{\sum}e^{-(1+\frac{\epsilon}{2})\log m}\quad~~ \big(\because -nI_r\big\{(\frac{1}{2}-\eta) (1+o(1))-\eta  \big(\frac{B_{k}}{(d^{\min}_H)^2}+o(1)\big)\big\} \ge (1+\frac{\epsilon}{2})\log m ~~ \text{ for sufficiently large $n$}\big)\nonumber\\
&=\underset{2\le k\le d}{\sum} O(m^{-1-\frac{\epsilon}{2}})=O(m^{-1-\frac{\epsilon}{2}})\nonumber
\end{align}
\endgroup
If we set $u_R^{(j)}=p_{a_j}, v_R^{(j)}=p_{b_j}$ for $j=1,\dots,m$, above result means $\Pr(\hat{u_R}^{(j)}\neq \hat{p}_{a_j})=O(m^{-1-\frac{\epsilon}{2}})$ for $j=1,\dots, m$. Similarly, $\Pr(\hat{v_R}^{(j)}\neq \hat{p}_{b_j})=O(m^{-1-\frac{\epsilon}{2}})$ for $j=1,\dots, m$. Then 
$$\underset{1\le j \le m}{\sum}\big\{\Pr(\hat{u_R}^{(j)}\neq \hat{p}_{a_j})+\Pr(\hat{v_R}^{(j)}\neq \hat{p}_{b_j})\big\}=2m\cdot O(m^{-1-\frac{\epsilon}{2}})=o(1)$$
which implies $\hat{u_R}^{(j)}= \hat{p}_{a_j}, \hat{v_R}^{(j)}= \hat{p}_{b_j}$ for all $j=1,\dots,m$ with probability approaching to $1$ as $n\rightarrow \infty$. As $\hat{p}_{a_j}\rightarrow p_{a_j}, \hat{p}_{b_j}\rightarrow p_{b_j}$ for all $j=1,\dots,m$ with probability approaching to $1$ as $n\rightarrow \infty$ by (6), we can conclude that $\hat{u_R}^{(j)}\rightarrow p_{a_j} (=u_R^{(j)}), \hat{v_R}^{(j)}\rightarrow p_{b_j} (=v_R^{(j)})$ for all $j=1,\dots,m$ with probability approaching to $1$ as $n\rightarrow \infty$.
\end{proof}
Lemma~\ref{lem:9} implies the success of Stage 2-(i).
\end{proof}

\textbf{Analysis of Stage 2-(ii).} With probability approaching to $1$ as $n\rightarrow \infty$, $T:=\lceil \log n\rceil$ iterations ensure that $\hat{A_R} = A_R$ and $\hat{B_R} = B_R$ will be recovered exactly.

\begin{proof}
Let $L(i;A,B):=-\log\big(\frac{{\alpha}(1-{\beta})}{{\beta}(1-{\alpha})}\big)e(\{i\},A)-\big[\underset{j:N^{\Omega}_{ij}=1}{\sum}\log(u_R^{(j)})+\underset{j:N^{\Omega}_{ij}=-1}{\sum}\log(1-u_R^{(j)})\big]+\log\big(\frac{{\alpha}(1-{\beta})}{{\beta}(1-{\alpha})}\big)e(\{i\},B)+\big[\underset{j:N^{\Omega}_{ij}=1}{\sum}\log(v_R^{(j)})+\underset{j:N^{\Omega}_{ij}=-1}{\sum}\log(1-v_R^{(j)})\big]$,\\
and $\hat{L}(i;A,B):=-\log\big(\frac{\hat{\alpha}(1-\hat{\beta})}{\hat{\beta}(1-\hat{\alpha})}\big)e(\{i\},A)-\big[\underset{j:N^{\Omega}_{ij}=1}{\sum}\log(\hat{u_R}^{(j)})+\underset{j:N^{\Omega}_{ij}=-1}{\sum}\log(1-\hat{u_R}^{(j)})\big]+\log\big(\frac{\hat{\alpha}(1-\hat{\beta})}{\hat{\beta}(1-\hat{\alpha})}\big)e(\{i\},B)+\big[\underset{j:N^{\Omega}_{ij}=1}{\sum}\log(\hat{v_R}^{(j)})+\underset{j:N^{\Omega}_{ij}=-1}{\sum}\log(1-\hat{v_R}^{(j)})\big]$.

There are $T$ iterations in Stage 2-(ii), and at $t$-th iteration, Algorithm 1 updates every user's affiliation by the following rule : put user $i$ to $A_R^{(t)}$ if $\hat{L}(i;A_R^{(t-1)},B_R^{(t-1)})<0$; put user $i$ to $B_R^{(t)}$ otherwise. In Lemma 11, we show the following holds with probability approaching to $1$ as $n\rightarrow \infty$: if we use $L(i;A,B)$ instead of $\hat{L}(i;A,B)$, $A_R,B_R$ can be recovered exactly within $1$ iteration of Stage 2-(ii). 

\begin{lemma}\label{lem:10}
Suppose $\frac{1}{2}nI_s+\gamma m p(d^{\min}_H)^2\ge (1+\epsilon)\log n$. Then there exist a constant $\tau >0$ such that $L(i;A_R,B_R)<-\tau \log n$ if $i\in A_R$; $L(i;A_R,B_R)>\tau \log n$ if $i\in B_R$ with probability $1-O(n^{-\frac{\epsilon}{2}})$
\end{lemma}

\begin{proof}
This lemma can be proved similarly by applying the argument of Lemma 9 in~\cite{Ahn_2018}.
\end{proof}

Our goal is to show $A_R,B_R$ can be recovered exactly by using $\hat{L}(i;A,B)$ in Stage 3. Define $\mathcal{Z}_{\delta}:=\big\{(A,B) : A\cup B=[n], A\cap B=\emptyset, |A\bigtriangleup A_R|=|B\bigtriangleup B_R|<\delta n\big\}$ for $\delta \in [\frac{1}{n},\frac{1}{2})$. 

\begin{lemma}\label{lem:11}
Suppose $\alpha=\Theta (\frac{\log n}{n})$ For arbitrary $\tau>0$, there exists $\delta_0<\frac{1}{2}$ such that if $\delta<\delta_0$, the following holds with probability $1-O(n^{-1})$ : for all $(A,B)\in Z_{\delta}$, $|L(i;A,B)-L(i;A_R,B_R)|\le \frac{\tau}{2}\log n$, for all except $\frac{\delta}{2}$ many $i's$.
\end{lemma}

\begin{proof}
This lemma can be proved similarly by applying the argument of Lemma 10 in~\cite{Ahn_2018}.
\end{proof}

\begin{lemma}\label{lem:12}
Suppose $\alpha=\Theta (\frac{\log n}{n}), p=\Theta(\frac{\log n}{m}+\frac{\log m}{n}), m=O(n)$. For arbitrary $\tau>0$, the following holds with probability approaching to 1 as $n\rightarrow \infty$: for all $A,B\subset [n]$, and $i\in [n]$, $|\hat{L}(i;A,B)-L(i;A,B)|\le \frac{\tau}{2}\log n$.
\end{lemma}

\begin{proof}
This lemma can be proved similarly by applying the argument of Lemma 11 in~\cite{Ahn_2018}.
\end{proof}

By Lemma~\ref{lem:11}, there exists $\delta_0 < \frac{1}{2}$ such that if $\delta<\delta_0$, the following holds with probability $1-O(n^{-1})$ : for all $(A,B)\in Z_{\delta}$, $|L(i;A,B)-L(i;A_R,B_R)|\le \frac{\tau}{2}\log n$, for all except $\frac{\delta}{2}$ many $i's$. At the same time, by applying Lemma~\ref{lem:12} to $(A,B)\in \mathcal{Z}_{\delta}$, the following holds with probability approaching to 1 as $n\rightarrow \infty$: for all $(A,B)\in \mathcal{Z}_{\delta}$, and $i\in [n]$, $|\hat{L}(i;A,B)-L(i;A,B)|\le \frac{\tau}{2}\log n$. Combining these two results, the following holds with probability approaching to 1 as $n\rightarrow \infty$ : for all $(A,B)\in Z_{\delta}$, $|\hat{L}(i;A,B)-L(i;A_R,B_R)|\le|\hat{L}(i;A,B)-L(i;A,B)|+|L(i;A,B)-L(i;A_R,B_R)|\le \frac{\tau}{2}\log n+\frac{\tau}{2}\log n = \tau\log n$, for all except $\frac{\delta}{2}$ many $i's$. Then together with Lemma~\ref{lem:10}, we eventually get the following holds with probability approaching to 1 as $n\rightarrow \infty$ : for all $(A,B)\in Z_{\delta}$,
$$\hat{L}(i;A,B)=
\begin{cases}
\le L(i;A_R,B_R)+|\hat{L}(i;A,B)-L(i;A_R,B_R)|<-\tau \log n+\tau \log n=0 \quad \text{if } i\in A_R;\\
\ge L(i;A_R,B_R)-|\hat{L}(i;A,B)-L(i;A_R,B_R)|>\tau \log n-\tau \log n=0 \quad~~~ \text{if } i\in B_R;
\end{cases}$$
for all except $\frac{\delta}{2}$ many $i's$. This means that at each iteration of Stage 2-(ii), every user's affiliation will be updated to the correct one except for $\frac{\delta}{2}$ many $i's$. So the following holds with probability approaching to 1 as $n\rightarrow \infty$ : whenever $(A,B)$ belongs to $\mathcal{Z}_{\delta}$, the result of single iteration of Stage 2-(ii) belongs to $\mathcal{Z}_{\frac{\delta}{2}}$. Then $T=\lceil \frac{\log (\delta_0 n)}{\log 2}\rceil$ iterations guarantee the exact recovery of $A_R,B_R$.
\end{proof}

\section{An Alternative Algorithm}\label{sec:E} As we mentioned in Remark~\ref{rmk:12}, we suggest an alternative algorithm, which utilizes both rating and graph data at Stage~1.
Analyzing the performance of this new algorithm is an interesting open problem.

\begin{algorithm}
   \caption{}
   \label{alg:2}

\begin{enumerate}[topsep=0pt,leftmargin=15mm,label=(\textbf{Stage \arabic*}),wide, labelindent=0pt]
    \item[\textbf{Input:}] $N^{\Omega}\in \{-1, 0, +1\}^{n\times m}$, $G=([n], E)$, $K$, $d$
    \item[\textbf{Output:}] Clusters of users $A_1^{(1)},\dots, A_K^{(1)}$, latent preference vectors $\hat{u_1}^{(1)},\dots, \hat{u_K}^{(1)}$
    
    \item[\textbf{Preprocessing:}] We first concatenate $G$ and $N^{\Omega}$ to get a new matrix $[G|N^{\Omega}] \in \{-1, 0, +1\}^{n\times (n+m)}$. We denote $[G|N^{\Omega}]$ by $I_0$. To make \textbf{Stage 1} and \textbf{Stage 2-(i)} independent, we split the information of $I_0$ by using the technique used in~\cite{Abbe_2016}. In specific, we generate a matrix $M_1\in \{0, +1\}^{n\times (n+m)}$ where each entry is drawn independently from the Bernoulli distribution with parameter $\frac{1}{\sqrt{\log n}}$. A matrix $M_2\in \{0, +1\}^{n\times (n+m)}$ is defined as $(\mathbf{1}_{n, (n+m)}-M_1)$. Then, we let $I_1 = I_0 \circ M_1$ and $I_2 = I_0 \circ M_2$, where $\circ$ is the Hadamard product.
    
    \item[\textbf{Stage 1. Partial recovery of clusters}] We apply \textbf{Part I} of the algorithm proposed in~\cite{Awasthi_2012} to $I_1$: (i) we project $I_1$ onto the subspace spanned by the top $K$ singular vectors, and we denote the projected matrix by $\hat{I_1}$; (ii) run a 10-approximate $k$-means algorithm on $\hat{I_1}$, and obtain an initial clustering result $A_1^{(0)},\dots, A_K^{(0)}$.
    
    \item[\textbf{Stage 2-(i)}] Run \textbf{Stage 2-(i)} of Alg.~\ref{alg:1} on $I_2$.
    
    \item[\textbf{Stage 2-(ii)}] Run \textbf{Stage 2-(ii)} of Alg.~\ref{alg:1} on $I_0$.
\end{enumerate}
\end{algorithm}

\end{document}